\documentclass[11pt,letterpaper]{article}

\usepackage[margin=1in]{geometry}
\usepackage[T1]{fontenc}
\usepackage{lmodern}
\usepackage{microtype}
\renewenvironment{abstract}
  {\begin{center}\normalfont\normalsize\bfseries Abstract\end{center}%
   \begin{quotation}\normalfont\normalsize}
  {\end{quotation}}

\usepackage{amsmath,amssymb,amsthm,mathtools,bm}
\numberwithin{equation}{section}

\usepackage{graphicx}
\usepackage{booktabs}
\usepackage{enumitem}
\DeclareMathOperator{\Tr}{Tr}
\newcommand{\ket}[1]{\lvert#1\rangle}
\newcommand{\bra}[1]{\langle#1\rvert}

\usepackage{xcolor}
\definecolor{citegreen}{HTML}{208054}
\definecolor{linkblue}{HTML}{0055CC}
\usepackage[
    colorlinks=true,
    linkcolor=linkblue,
    citecolor=citegreen,
    urlcolor=linkblue,
    bookmarksnumbered=true
]{hyperref}
\usepackage[capitalize,nameinlink,noabbrev]{cleveref}

\theoremstyle{plain}
\newtheorem{theorem}{Theorem}
\newtheorem{lemma}[theorem]{Lemma}
\newtheorem{proposition}[theorem]{Proposition}
\newtheorem{corollary}[theorem]{Corollary}

\theoremstyle{definition}

\theoremstyle{remark}

\newif\ifanonymous
\anonymousfalse

\title{Optimal single-copy estimation of quantum state moments:\\[0.15em]
why $\Tr(\rho^3)$ and $\Tr(\rho^4)$ are equally hard}

\ifanonymous
    \author{Anonymous Author}
\else
    \author{Zhenhuan Liu\thanks{Quantum Research Center, Technology Innovation Institute (TII), Abu Dhabi, United Arab Emirates.
    Email: \href{mailto:qubithuan@gmail.com}{\texttt{qubithuan@gmail.com}}.}}
\fi

\date{}

\begin{document}

\hypersetup{pageanchor=false}
\pagenumbering{gobble}
\pagestyle{empty}
\maketitle
\thispagestyle{empty}

\begin{abstract}
Estimating nonlinear properties of an unknown quantum state with restrictive experimental accessibility is a fundamental problem in quantum learning.
We study the sample complexity of estimating the state moments $\operatorname{Tr}(\rho^t)$, allowing arbitrary adaptive single-copy measurements.
While purity estimation has optimal sample complexity $\Theta(\sqrt d)$ at constant additive error, the optimal dimension dependence for higher moments has remained unclear.
We resolve this question for every fixed integer $t\ge 2$: for any sufficiently small constant additive error, the optimal sample complexity is
\[
\Theta\!\left(d^{\frac{\lceil\log_2 t\rceil}{1+\lceil\log_2 t\rceil}}\right)
\]
with constant success probability.
The complexity thus follows a dyadic hierarchy: for each integer $h\geq1$, all moment orders $2^{h-1}<t\leq2^h$ share the same exponent.
In particular, the third and fourth moments both require $\Theta(d^{2/3})$ copies.
Our results also apply to multistate overlaps, weighted and partial-transpose moments, fixed-rank spectrum estimation, and distributed estimation of nonlinear quantities.
We also establish an $\Omega(d^{1-1/t})$ lower bound for arbitrary nonadaptive single-copy protocols, matching the known upper bound and demonstrating an advantage from adaptivity for every $t\geq4$.
Our upper bound uses adaptive state filtering to reduce higher-order moment estimation to lower-order moment estimation in a smaller state space.
For the lower bound, we directly compare moment-matched ensembles through a smooth interpolation, bypassing the maximally mixed state as an intermediate hypothesis and yielding matching bounds beyond $\sqrt d$.
\end{abstract}

\clearpage

\tableofcontents
\clearpage

\pagenumbering{arabic}
\pagestyle{plain}
\hypersetup{pageanchor=true}

\section{Introduction}
\label{sec:introduction}

Quantum learning seeks to extract classical information about quantum systems, providing the readout stage of all quantum information-processing tasks.
A central question is how many copies of an unknown quantum state, or queries to an unknown quantum channel, are required to learn a property of interest.
Examples range from reconstructing a full description through tomography~\cite{Vogel1989QuantumTomography,James2001MeasurementQubits} to estimating selected properties, such as fidelities to target states~\cite{Flammia2011DirectFidelity,daSilva2011PracticalCharacterization}, correlation functions~\cite{Huang2020ClassicalShadows}, and various quantum resources~\cite{Chitambar2019QuantumResourceTheories}.
The answer depends not only on the target property but also on how the learner can access the system.
Understanding learning complexity under experimentally motivated restrictions, including limited quantum memory~\cite{Chen2022QuantumMemory,Aharonov2022QuantumAlgorithmicMeasurement}, incoherent channel queries~\cite{Chen2023UnitarityEstimation,liu2026exponential}, and noisy access to quantum systems~\cite{Cotler2026NoisyLearning,Zhou2018Heisenberg}, is therefore a fundamental and practically important problem.

Nonlinear properties of quantum states are of particular interest.
Quantities such as purity, entropy, and entanglement negativity~\cite{Vidal2002Negativity} are used to study many-body dynamics~\cite{Islam2015EntanglementEntropy,Brydges2019RandomizedEntropy}, detect entanglement~\cite{Guhne2009EntanglementDetection}, and mitigate errors through virtual purification~\cite{Huggins2021VirtualDistillation,Koczor2021ExponentialErrorSuppression}.
While bounded linear observables can be estimated efficiently using single-copy measurements, nonlinear properties generally require more elaborate protocols.
For polynomial functionals, a standard approach is to express the target as $\operatorname{Tr}(O_t\rho^{\otimes t})$ and perform a joint measurement on several copies~\cite{Ekert2002DirectEstimation,Brun2004PolynomialFunctions}.
The SWAP test for purity and its cyclic-permutation generalizations are canonical examples.
Such protocols require coherent access to multiple copies and operations coupling them, imposing substantial experimental demands.
This motivates the search for efficient ways to estimate nonlinear properties.

Along this line, a basic yet unresolved question is the optimal sample complexity of estimating a single state moment $p_t(\rho)=\operatorname{Tr}(\rho^t)$ using single-copy measurements.
We consider an unknown $d$-dimensional state and a fixed integer $t\geq 2$, and ask how many copies are needed to estimate $p_t(\rho)$ to a sufficiently small constant additive error with success probability at least $2/3$.
In the single-copy model, each copy is measured separately; the measurement may depend arbitrarily on previous classical outcomes, but no quantum information is retained between copies.
For purity ($t=2$), matching upper and lower bounds establish an optimal complexity of $\Theta(\sqrt d)$ copies, with the upper bound already achievable by a nonadaptive protocol~\cite{Chen2022QuantumMemory,Aharonov2022QuantumAlgorithmicMeasurement,Gong2026PurityInnerProduct}.
The same $\Omega(\sqrt d)$ lower bound applies to every fixed higher moment, even when adaptive measurements are allowed: the hard discrimination problem between a Haar-random pure state and the maximally mixed state remains applicable, since their $t$-th moments are $1$ and $d^{1-t}$, respectively.
Existing upper bounds, however, do not match this lower bound for $t\geq 3$.
Classical-shadow methods also allow higher-order moment estimation~\cite{Huang2020ClassicalShadows,Elben2020MixedStateEntanglement,Seif2023ShadowDistillation}, with global randomized measurements giving an $\mathcal{O}(d)$ sample upper bound at fixed accuracy~\cite{Pelecanos2026UnentangledSpectrum}.
As part of their work on entanglement negativity, Zhou, Zeng, and Liu obtained an $\mathcal{O}(d^{2/3})$ single-copy protocol for the third moment~\cite{Zhou2020Negativity}.
More recently, Li et al.\ developed collision-based nonlinear estimation, a nonadaptive single-copy protocol achieving an $\mathcal{O}(d^{1-1/t})$ upper bound for every fixed integer $t\geq 2$~\cite{Li2026SingleSetting}.
Allowing joint operations on two copies enables estimation of both the third and fourth moments using $\mathcal{O}(\sqrt d)$ copies at constant additive error~\cite{Zhou2024HybridFramework,Ye2025ReplicaAdvantage}.
Within the single-copy model, however, the available bounds leave a polynomial gap for higher moments and do not determine whether adaptivity alone can improve the sample complexity.
In this work, we address the following question:
\begin{quote}
\emph{What is the optimal sample complexity of estimating $\operatorname{Tr}(\rho^t)$ using adaptive single-copy measurements?}
\end{quote}

\subsection{Our results}
\label{sec:our-results}

For every fixed integer moment order, we establish matching upper and lower bounds that apply to arbitrary quantum states and allow fully adaptive single-copy measurements.

\begin{theorem}[Optimal single-copy estimation of state moments]
\label[theorem]{thm:optimal-moments}
For every fixed integer $t\geq2$ and sufficiently small constant $\epsilon>0$, the optimal worst-case sample complexity of estimating $\operatorname{Tr}(\rho^t)$ for an unknown $d$-dimensional quantum state $\rho$, to additive error $\epsilon$ with success probability at least $2/3$, is
\[
    \Theta\!\left(
        d^{\frac{\lceil\log_2 t\rceil}{1+\lceil\log_2 t\rceil}}
    \right).
\]
Here measurements act on one copy at a time and may adapt to previous classical outcomes, but no quantum information is retained between copies.
\end{theorem}

Our result reveals that the optimal dimension exponent is constant over dyadic ranges of moment orders.
In particular, the third and fourth moments both require $\Theta(d^{2/3})$ copies, while all orders from five through eight require $\Theta(d^{3/4})$ copies.
Thus, increasing the degree of the nonlinear quantity need not increase its asymptotic dependence on the state dimension; the hidden constants can still depend on the moment order and the target precision.

\begin{figure}[htbp]
    \centering
    \includegraphics[width=0.8\linewidth]{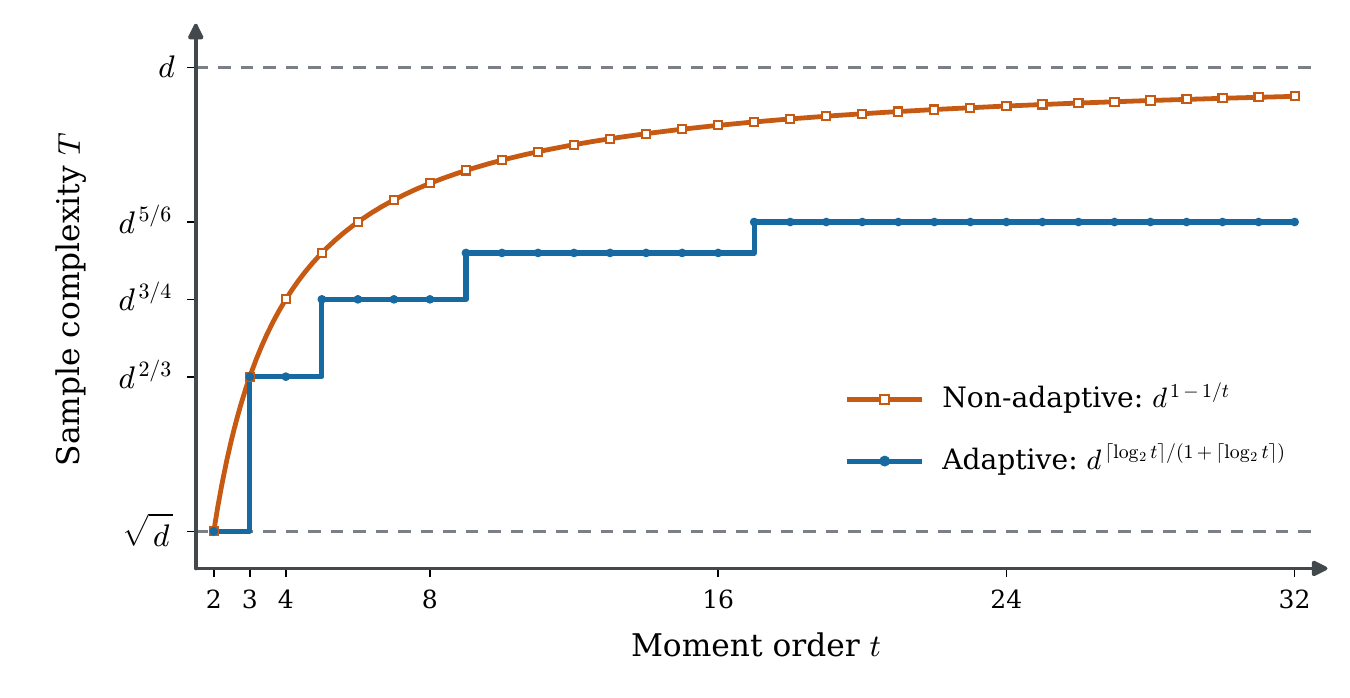}
    \caption{Comparison between sample complexities of adaptive and nonadaptive protocols.}
    \label{fig:overview}
\end{figure}

Our main technical contribution is a lower-bound method that directly compares two ensembles of quantum states.
A common approach in learning-tree analyses bounds their transcript distance by comparing each ensemble with a single reference state or channel, which are typically chosen as the maximally mixed state or fully depolarizing channel, and applying the triangle inequality~\cite{Chen2022QuantumMemory}.
In our setting, however, each hard ensemble can already be distinguished from the maximally mixed state using $\mathcal{O}(\sqrt d)$ copies through purity estimation, even while the two ensembles remain difficult to distinguish from each other.
This reference-based comparison therefore cannot establish lower bounds beyond the $\sqrt d$ scale.

We overcome this obstruction by constructing a smooth interpolation between hard ensembles built from spectra with matching lower-order moments but separated target moments.
We differentiate the measurement-record probabilities along this interpolation.
On typical records, moment matching cancels the low-order terms, while martingale second-moment estimates control the remaining matrix contributions to the derivative, uniformly over adaptive measurement strategies.
Integrating this bound and adding the small probabilities of exceptional records at the two endpoints controls their total variation distance directly.
Related smooth-path techniques have previously been used to characterize statistical distinguishability~\cite{Wootters1981StatisticalDistance} and derive lower bounds for Hamiltonian learning~\cite{Zhou2026AnsatzFreeHamiltonian}.
Our combination of moment matching and martingale estimates yields matching lower bounds beyond $\sqrt d$, without requiring either ensemble to remain close to the maximally mixed state, and provides a method for studying other learning problems.

Our upper bound uses adaptive state filtering to reduce higher-order moment estimation to simpler estimation problems in smaller state spaces.
Related projection-based filtering techniques have also been used for quantum spectrum estimation~\cite{Pelecanos2026UnentangledSpectrum}.
Preliminary single-copy measurement outcomes specify filters onto lower-dimensional subspaces, which are then applied to fresh copies of the unknown state.
The moments of the resulting conditional states encode information about higher-order moments of the original state, allowing us to recover the target through lower-order estimates and the corresponding filtering probabilities.
We apply this reduction recursively, with purity estimation as the base case.
The analysis accounts for the normalization introduced by conditioning, the copies required to obtain successful filtering outcomes, and the statistical errors accumulated through the recursion.
Optimizing the filter dimensions and the allocation of copies balances these costs and gives the claimed upper bound.
All filtering and measurement operations act on individual copies, with only classical information retained between rounds.

The main theorem also yields consequences for multistate overlaps, weighted moments, partial-transpose moments, and spectrum estimation at fixed rank.
Ordered multistate overlaps, such as $\operatorname{Tr}(\rho_0\rho_1\rho_2\cdots)$, enter measures of noncommutativity~\cite{Ferro2018MeasuringQuantumness} and experimentally accessible bounds on quantum fidelity~\cite{Miszczak2009SubSuperFidelity}.
Weighted moments, such as $\operatorname{Tr}(O\rho^2)$, are used in quantum error mitigation through virtual purification~\cite{Huggins2021VirtualDistillation,Koczor2021ExponentialErrorSuppression}, while partial-transpose moments provide entanglement witnesses~\cite{Elben2020MixedStateEntanglement,Zhou2020Negativity}.
We also obtain an upper bound for spectrum estimation~\cite{Keyl2001Spectrum,ODonnell2021QuantumSpectrumTesting} under a fixed-rank promise.

\begin{corollary}[Consequences for nonlinear properties and spectra]
\label[corollary]{cor:moment-consequences}
Let $t,r\geq2$ be fixed integers, and consider adaptive single-copy estimation with sufficiently small constant error $\epsilon>0$ and success probability at least $2/3$.
Estimating the ordered overlap $\operatorname{Tr}(\rho_0\cdots\rho_{t-1})$ of $t$ unknown $d$-dimensional states, with error measured in complex modulus, has optimal worst-case total sample complexity $\Theta(d^{\lceil\log_2 t\rceil/(1+\lceil\log_2 t\rceil)})$.
For a known Hermitian observable $O$ with $\|O\|_{\mathrm{op}}\leq1$, estimating $\operatorname{Tr}(O\rho^t)$ has the same optimal worst-case sample complexity over $(O,\rho)$.
Estimating $\operatorname{Tr}[(\rho_{AB}^{\mathrm{T}_A})^t]$ also has this optimal sample complexity when one subsystem has fixed dimension, where $d=\dim(A)\dim(B)$.
If $\operatorname{rank}(\rho)\leq r$, its sorted spectrum can be estimated to $\ell_1$ error $\epsilon$ using $\mathcal{O}(d^{\lceil\log_2 r\rceil/(1+\lceil\log_2 r\rceil)})$ copies, without any eigenvalue-gap assumption.
\end{corollary}

Ordered overlaps are recovered from ordinary moments of mixtures augmented with suitable known auxiliary states to isolate the desired ordering; these operations preserve single-copy access, and the matching lower bound follows by taking all input states equal.
Weighted moments follow by applying filters determined by $O$ and recovering the desired quantity through polynomial interpolation.
For partial-transpose moments with one subsystem of fixed dimension, we express the transpose map as a linear combination of physical channels and expand the target into ordered overlaps of their output states, which are estimated using the preceding overlap protocol.
Finally, a rank-$r$ spectrum is determined by its first $r$ moments, and estimating these moments to sufficiently small, dimension-independent accuracy gives the stated spectrum-estimation upper bound.

Our filtering construction also applies to distributed estimation.
Alice and Bob hold copies of unknown $d$-dimensional states $\rho$ and $\sigma$, respectively.
Each party may perform arbitrary joint operations on their own copies, and they may exchange classical messages, but no quantum communication or preshared entanglement is available.
Anshu, Landau, and Liu established the optimal sample complexity of estimating $\operatorname{Tr}(\rho\sigma)$ in this model, obtaining $\Theta(\sqrt d)$ at constant additive accuracy~\cite{Anshu2022DistributedInnerProduct}.
We obtain the same dimension dependence for every fixed higher moment.

\begin{corollary}[Distributed estimation of alternating moments]
\label[corollary]{cor:distributed-alternating-moments}
For every fixed integer $t\geq2$ and sufficiently small constant $\epsilon>0$, the optimal worst-case total sample complexity of estimating
$\operatorname{Tr}[(\rho\sigma)^t]$
to additive error $\epsilon$ with success probability at least $2/3$ is $\Theta(\sqrt d)$ in the distributed model above.
The sample complexity counts all copies of both input states, including unsuccessful filtering attempts.
\end{corollary}

The key difference from the single-copy setting is that Alice and Bob may each perform arbitrary joint measurements on their own copies.
Our state-filtering construction lets Alice use her measurement outcomes to specify filters on Bob's copies, whose filtered moments encode $\operatorname{Tr}[(\rho\sigma)^t]$.
Bob can estimate these moments using local joint measurements with a dimension-independent number of successful outputs at fixed order and accuracy.
Thus, moment estimation introduces no additional dimension dependence: balancing Alice's preliminary measurements against Bob's filtering attempts gives $\mathcal{O}(\sqrt d)$ total copies for every fixed $t$.
The matching lower bound follows from the same-versus-independent Haar-random pure-state discrimination problem of~\cite{Anshu2022DistributedInnerProduct}.

Our state-filtering protocol is inherently adaptive: earlier measurement outcomes determine the filters applied to later copies, and this dependence recurs at each level of the construction.
A natural question is therefore whether adaptivity is necessary to achieve the optimal complexity.
We answer this by determining the optimal complexity when all measurements must be chosen in advance.

\begin{theorem}[Optimal nonadaptive estimation of state moments]
\label[theorem]{thm:nonadaptive-moments}
For every fixed integer $t\geq2$ and sufficiently small constant $\epsilon>0$, the optimal worst-case sample complexity of estimating $\operatorname{Tr}(\rho^t)$ for an unknown $d$-dimensional quantum state, to additive error $\epsilon$ with success probability at least $2/3$, using nonadaptive single-copy measurements is $\Theta(d^{1-1/t})$.
Conditioned on an input-independent classical random seed, all positive operator-valued measures (POVMs) are fixed before any outcomes are observed; arbitrary classical postprocessing is allowed.
\end{theorem}

The upper bound is attained by nonadaptive collision-based estimation~\cite{Li2026SingleSetting}, while we prove the matching lower bound for arbitrary nonadaptive single-copy POVMs in \Cref{sec:proof-nonadaptive}.
Together with \Cref{thm:optimal-moments}, at any fixed accuracy below both theorems' thresholds, this establishes a strict advantage from adaptivity for every fixed $t\geq4$.
In particular, without adaptivity, the fourth moment requires $\Theta(d^{3/4})$ copies and is asymptotically harder to estimate than the third moment, which requires $\Theta(d^{2/3})$ copies.
Thus, the equality of their optimal dimension exponents relies on adaptive measurements.
We summarize the results of \Cref{thm:optimal-moments,thm:nonadaptive-moments} in \Cref{fig:overview}.

\section{Overview of techniques}
\label{sec:overview}

We first construct a single-copy protocol based on state filtering that achieves the upper bound.
We then prove the matching lower bound by smoothly interpolating between hard ensembles.
Martingale second-moment estimates control the probability derivative on typical records, while the probabilities of exceptional records at the endpoints bound their contribution to total variation.
Exploiting conditional independence within the same framework yields the stronger lower bound for nonadaptive protocols.
We also apply our upper-bound techniques to multistate overlaps, weighted and partial-transpose moments, fixed-rank spectrum estimation, and distributed nonlinear estimation.
Throughout this section, moment orders and target accuracies are fixed independently of the dimension, and we write $p_j=p_j(\rho)=\operatorname{Tr}(\rho^j)$.

\subsection{Upper bound: recursive state filtering}
\label{sec:overview-upper}

The central idea is to encode higher-order moments of $\rho$ into lower-order moments of a state supported on a smaller, classically known subspace.
The resulting reduction doubles the range of accessible moment orders at each recursive step.

\begin{lemma}[Doubling the accessible moment orders]
\label[lemma]{lem:moment-doubling}
Fix an integer $m\geq2$ and $1/2\leq a<1$.
Suppose that, for every dimension $d\geq1$, $p_m$ can be estimated on every $d$-dimensional state using $\mathcal{O}(d^a)$ adaptive single-copy measurements at any fixed additive accuracy and failure probability.
Then all moments $p_{m+1},\ldots,p_{2m}$ can be estimated using $\mathcal{O}(d^{1/(2-a)})$ copies of an arbitrary $d$-dimensional state, under the same accuracy convention.
The implicit constants may depend on $m,a$, the accuracy, and the failure probability, but not on the dimension or the input state.
\end{lemma}

The ingredient behind this reduction is a tunable transformation of moments.
For $x\in[0,1]$, consider the normalized matrix
\[
    \tau_x=\frac{\rho+x\rho^2}{1+xp_2(\rho)}.
\]
Its $m$th moment contains the moments of $\rho$ from order $m$ through order $2m$.
Our protocol constructs a low-dimensional state whose $m$th moment approximates that of $\tau_x$.

\begin{lemma}[Moment transformation by state filtering]
\label[lemma]{lem:moment-filter}
Fix $m\geq2$, $0<b<1$, and constants $0<\eta,\delta<1$, and let $n=\lceil d^b\rceil$.
For each $x\in[0,1]$, a preliminary stage using at most $n$ copies constructs a known filter whose conditional output $\sigma_x$ is supported on a known subspace of dimension at most $n$.
For sufficiently large $d$, with probability at least $1-\delta$, the filter succeeds with probability $\Theta(n/d)$ and
\[
    \left|
        p_m(\sigma_x)
        -\frac{\operatorname{Tr}[(\rho+x\rho^2)^m]}
                    {(1+xp_2(\rho))^m}
    \right|\leq\eta.
\]
On this event, a protocol requiring $M\geq1$ single-copy measurements of $\sigma_x$ can be simulated using $\mathcal{O}((d/n)M)$ additional copies of $\rho$, with additional failure probability at most $\delta$.
The copy budget includes unsuccessful filtering outcomes.
The implicit constants and dimension threshold may depend on the fixed parameters $m,b,\eta,\delta$.
\end{lemma}

To construct the filter, we use a covariant rank-one measurement on $\rho$~\cite{Harrow2013SymmetricSubspace}.
Its classical outcome is a unit vector $v$ satisfying $\mathbb E[vv^\dagger]=(I+\rho)/(d+1)$.
Mixing these outcomes with independently generated Haar-random vectors lets us tune the covariance to $(I+x\rho)/(d+x)$ for any $x\in[0,1]$.
This covariance is chosen because $\sqrt\rho(I+x\rho)\sqrt\rho=\rho+x\rho^2$, whose $m$th power contains the moments we wish to estimate.
From $n$ independent vectors we form the known positive matrix $S_x=\sum_{i=1}^n v_iv_i^\dagger$ and use the filter $K_x=\sqrt{S_x/L}$, where $L$ is a sufficiently large constant depending only on the fixed exponent $b$.
We abort if $\|S_x\|_{\mathrm{op}}>L$, an event whose probability vanishes with $d$; otherwise this is a valid quantum operation.
Its conditional output is $\sigma_x=\sqrt{S_x}\rho\sqrt{S_x}/\operatorname{Tr}(S_x\rho)$, supported on the known subspace $\operatorname{span}\{v_1,\ldots,v_n\}$ of dimension at most $n$.

To relate this output to the desired moments, define $B_x=(d+x)\sqrt\rho S_x\sqrt\rho/n$ and $H_x=\rho+x\rho^2$.
These matrices are used only for analysis; implementing the filter requires only the known matrix $S_x$.
The covariance identity gives $\mathbb E B_x=H_x$, and both the mean-square Hilbert--Schmidt error and the mean-square trace error are $\mathcal{O}(1/n)$.
Moreover, $\sqrt{S_x}\rho\sqrt{S_x}$ and $\sqrt\rho S_x\sqrt\rho$ have the same nonzero eigenvalues.
Consequently,
\[
    p_m(\sigma_x)
    =\frac{\operatorname{Tr}(B_x^m)}{(\operatorname{Tr}B_x)^m}
    \approx
    \frac{\operatorname{Tr}[(\rho+x\rho^2)^m]}
         {(1+xp_2(\rho))^m}.
\]
Here trace concentration controls the normalization, which remains bounded away from zero because $\operatorname{Tr}H_x=1+xp_2(\rho)\in[1,2]$.
The same concentration bound gives the filtering success probability
$\operatorname{Tr}(S_x\rho)/L=n\operatorname{Tr}B_x/[L(d+x)]=\Theta(n/d)$, establishing the guarantees in \Cref{lem:moment-filter}.
The filter acts on fresh copies and depends only on the preliminary classical record.
It can be combined with each subsequent measurement into a single POVM, so no quantum information needs to be retained between rounds.

We can now see the cost of the reduction in \Cref{lem:moment-doubling}.
The assumed $m$th-moment estimator needs $\mathcal{O}(n^a)$ successful filtered samples.
Each success costs $\mathcal{O}(d/n)$ original copies, giving a subsequent cost of $\mathcal{O}(dn^{a-1})$.
For $a<1$, increasing $n$ reduces this cost but increases the preliminary cost.
Balancing the two gives
\[
    \mathcal{O}\!\left(n+dn^{a-1}\right)
    =\mathcal{O}\!\left(d^{1/(2-a)}\right)
    \qquad\text{for }n=\left\lceil d^{1/(2-a)}\right\rceil.
\]
We also estimate $p_2(\rho)$ using $\mathcal{O}(\sqrt d)$ copies~\cite{Gong2026PurityInnerProduct}, which does not change this scaling.
Multiplying the estimated filtered moment by $(1+xp_2(\rho))^m$ then estimates the polynomial
\begin{equation}
\label{eq:filter-moment-polynomial}
    F_m(x):=\operatorname{Tr}[(\rho+x\rho^2)^m]
    =\sum_{j=0}^m\binom mj x^j p_{m+j}(\rho).
\end{equation}
Thus the desired moments appear as its coefficients and can be recovered by interpolation from $m+1$ fixed values of $x$.
Since $m$, the target error, and the failure probability are fixed, the number of filters and the accuracy required for interpolation contribute only dimension-independent factors.

The third and fourth moments provide a simple example.
Taking $m=2$ gives $F_2(x)=p_2+2xp_3+x^2p_4$, so both higher moments can be recovered from purity measurements on the filtered states.
After estimating $p_2$, the two values $F_2(1/2)$ and $F_2(1)$ determine
\[
    p_3=2F_2(1/2)-\tfrac12F_2(1)-\tfrac32p_2,
    \qquad
    p_4=2F_2(1)-4F_2(1/2)+2p_2.
\]
Purity estimation on the filtered subspace needs $\mathcal{O}(\sqrt n)$ successful samples, so the total cost is $\mathcal{O}(n+d/\sqrt n)$.
Choosing $n=\lceil d^{2/3}\rceil$ gives $\mathcal{O}(d^{2/3})$ copies for both moments.
Their equal scaling follows from recovering two coefficients of the same polynomial, without an additional filtering stage for the fourth moment.
More generally, iterating the exponent transformation $a\mapsto1/(2-a)$ from the purity exponent $1/2$ gives $1/2,2/3,3/4,\ldots$, while doubling the largest accessible moment order at each step.
This proves the upper bound in \Cref{thm:optimal-moments}.

\subsection{Lower bound: smooth interpolation of hard ensembles}
\label{sec:overview-lower}

We illustrate the argument with the third and fourth moments.
Consider the probability vectors $\lambda^{(0)}=(2/3,1/6,1/6)$ and $\lambda^{(1)}=(1/2,0,1/2)$.
Writing $p_k(\lambda)=\sum_j\lambda_j^k$, both have purity $p_2=1/2$, but their third and fourth moments differ by $1/18$ and $2/27$, respectively.
The key is to connect them by a smooth path that preserves the matching purity throughout:
\begin{equation}
\label{eq:third-fourth-spectral-path}
    \lambda_j(\theta)
    =\frac{1+\cos(\theta+2\pi j/3)}3,
    \qquad j=0,1,2,\quad 0\leq\theta\leq\pi/3.
\end{equation}

These weights define a smooth family of random states.
For independent standard complex Gaussian vectors $g_0,g_1,g_2\in\mathbb C^d$ and a sufficiently small fixed constant $\alpha>0$, set
\begin{equation}
\label{eq:overview-gaussian-prior}
    \rho_\theta
    =\frac{1-R_\theta}{d}I
      +\frac\alpha d\sum_{j=0}^2\lambda_j(\theta)g_jg_j^\dagger,
    \qquad
    R_\theta=\frac\alpha d\sum_{j=0}^2
                  \lambda_j(\theta)\|g_j\|^2.
\end{equation}
We weight the Gaussian sampling density by a continuously differentiable cutoff $\chi(R_\theta)$ and renormalize, changing only an exponentially unlikely tail and ensuring $R_\theta<1/2$ and $\rho_\theta\succeq I/(2d)$.
For every fixed $k\geq2$, the state moment $p_k(\rho_\theta)$ converges in probability to $\alpha^k p_k(\lambda(\theta))$ as $d$ grows.
Thus the endpoint ensembles retain constant gaps in both target moments.

The distinguishing problem chooses either endpoint with equal prior probabilities, draws one state from that ensemble, and supplies copies of that same state to the learner.
An accurate moment estimator would distinguish the endpoints by thresholding its output.
By Le Cam's testing identity~\cite{LeCam1986,Tsybakov2009}, success probability $2/3-o(1)$ requires the endpoint transcript laws to have total variation distance at least $1/3-o(1)$.
Our goal is to show that too few single-copy measurements cannot produce this separation.

Let $P_\theta$ be the law of the complete classical record $H$, including the random seed, and let $p_\theta$ be its density relative to the fixed law $\nu$ obtained by running the same protocol on $I/d$.
For discrete records, this means $p_\theta(H)=P_\theta(H)/\nu(H)$.
The fundamental theorem of calculus gives
\begin{equation}
\label{eq:overview-tv-derivative}
    \operatorname{TV}(P_u,P_v)
    \leq\frac12\int_u^v\int
          |\partial_\theta p_\theta(H)|\,d\nu(H)\,d\theta.
\end{equation}
Here $I/d$ only defines a reference measure: we compare the endpoints directly, without requiring either to be close to that reference.
Related smooth-path comparisons underlie statistical-distance geometry~\cite{Wootters1981StatisticalDistance} and lower bounds for Hamiltonian learning~\cite[App.~G]{Zhou2026AnsatzFreeHamiltonian}.
Our task is to control the derivative while preserving the cancellations from moment matching.

\paragraph{From the probability derivative to spectral blocks.}
Refine each measurement into rank-one effects $w_iP_i$, as in learning-tree analyses~\cite{Chen2022QuantumMemory}.
Once a record is fixed, its measurement choices are fixed even for an adaptive protocol.
Writing $f_i=d\operatorname{Tr}(P_i\rho_\theta)$, the relative likelihood for fixed latent vectors is $\prod_i f_i$, and hence
\[
    p_\theta(H)
    =\frac{\mathbb E_G[\chi(R_\theta)\prod_{i=1}^T f_i]}
           {\mathbb E_G\chi(R_\theta)},
\]
where $\mathbb E_G$ denotes expectation under the original Gaussian law.
Up to the exponentially small cutoff correction, Gaussian pairing and the moment--cumulant relation organize this product into blocks~\cite{Isserlis1918GaussianMoments,McCullagh1987TensorMethods}; \Cref{pre:gaussian-blocks} gives the expansion used here.
Differentiating a block of size $\ell$ produces the spectral coefficient $\alpha^\ell p_\ell'(\lambda(\theta))/\ell$.
Thus every block whose spectral moment is constant along the path disappears from the derivative; here the expansion starts at block size three.

The matrix statistics associated with these blocks have a useful martingale structure.
Define $\widehat Y_i=(P_i-I/d)/f_i$.
Conditional on the latent vectors, the seed, and the preceding outcomes, its mean is zero: the outcome probability $w_if_i/d$ cancels the denominator, and POVM completeness gives $\sum_yw_y(P_y-I/d)=0$.
This conditional cancellation holds for arbitrary adaptive measurements.

To obtain uniform matrix bounds, we use the auxiliary differences
$Y_i=\boldsymbol1_{\{\|\sum_{j<i}P_j\|_{\mathrm{op}}<\kappa\}}\widehat Y_i$, initially with $\kappa=6$.
Their multiplier depends only on the past, so the martingale-difference property is preserved.
On the good records $G=\{H:\|\sum_{i=1}^T P_i\|_{\mathrm{op}}<\kappa\}$, they agree with the original matrices; the actual protocol still runs for all $T$ rounds.
Define the block statistics, or cycle sums, by
\[
    \mathcal C_{\ell,T}
    =\sum_{\substack{i_1,\ldots,i_\ell\leq T\\\text{all distinct}}}
       \operatorname{Tr}(Y_{i_1}\cdots Y_{i_\ell}),
\]
and write $\widehat{\mathcal C}_{\ell,T}$ for the same sum formed from $\widehat Y_i$.
The two agree on $G$.
For the copy budgets considered below, the Gaussian expansion and Bayes' formula give, uniformly on these records,
\begin{equation}
\label{eq:overview-marked-cycle-derivative}
    \partial_\theta\log p_\theta(H)
    =\mathbb E_\theta\!\left[
       \left.\sum_{\ell=3}^{T}
          \frac{\alpha^\ell p_\ell'(\lambda(\theta))}{\ell}
          \mathcal C_{\ell,T}\right|H\right]+o(1),
    \qquad H\in G.
\end{equation}
The error comes from the cutoff and is exponentially small.
Conditional Jensen and Cauchy--Schwarz therefore reduce the good-record derivative bound to estimates of $\mathbb E_\theta|\mathcal C_{\ell,T}|^2$ under the original joint law of the latent vectors and outcomes.

\paragraph{Martingale estimates and the dyadic hierarchy.}
We use the same auxiliary truncation with a general fixed threshold $\kappa$, suppressing its dependence in the matrix-bound constants below.
Let $\mathcal A_{s,n}$ be the sum of all ordered products of length $s$ with distinct indices among $Y_1,\ldots,Y_n$.
Cyclicity of the trace gives
$\mathcal C_{\ell,n}-\mathcal C_{\ell,n-1}
=\ell\operatorname{Tr}(Y_n\mathcal A_{\ell-1,n-1})$.
The chain multiplying $Y_n$ is determined by the past, so this increment has conditional mean zero.
The martingale isometry~\cite[Sec.~12.1]{Williams1991ProbabilityMartingales} and the one-step variance bound imply
\[
    \mathbb E_\theta|\mathcal C_{\ell,T}|^2
    \leq\frac{2\ell^2}{d}\sum_{n=\ell}^T
       \mathbb E_\theta\|\mathcal A_{\ell-1,n-1}\|_{\mathrm{HS}}^2.
\]
The key fact for the third and fourth moments is that both required chain lengths satisfy
\[
    \mathbb E_\theta\|\mathcal A_{s,n}\|_{\mathrm{HS}}^2
       =\mathcal{O}(n^2/d),\qquad s=2,3.
\]
Thus both cycle variances are $\mathcal{O}(T^3/d^2)$, and their contributions to the derivative are $\mathcal{O}(T^{3/2}/d)$.

This common scaling follows from a recursion that also explains the hierarchy at higher orders.
For a chain of length $s$, independently assign each round a uniformly random color from $\{1,\ldots,s\}$ and form
\[
    B_{a,n}=\sum_{\substack{i\leq n\\\mathrm{color}(i)=a}}Y_i,
    \qquad
    \mathcal A_{s,n}
       =s^s\mathbb E_{\mathrm{colors}}
          [B_{1,n}\cdots B_{s,n}].
\]
Each product selects one index from each color class, and each ordered distinct tuple receives the prescribed colors with probability $s^{-s}$.
Jensen's inequality therefore reduces the chain estimate to bounding these products uniformly over fixed colorings.

Fix a coloring and write
$M_{w,n}=B_{a_1,n}\cdots B_{a_u,n}$ for an ordered list
$w=(a_1,\ldots,a_u)$ of distinct colors, where $1\leq u\leq s$.
At each round, only one group sum changes, by the addition of $Y_n$.
The product's increment is consequently either zero or $LY_nR$, with $L,R$ determined by the past.
Conditional centering of $Y_n$ makes $M_{w,n}$ a matrix martingale.

For a nonzero increment with $u\geq2$, the two sides of $Y_n$ contain $u-1$ old factors in total, so at least one side contains $\lfloor u/2\rfloor$ factors.
We retain a consecutive subproduct of this length and bound the remaining factors in operator norm using the auxiliary truncation.
The one-step variance estimate gives
\[
    \mathbb E_\theta
       \|M_{w,n}-M_{w,n-1}\|_{\mathrm{HS}}^2
    \leq\frac{C_u}{d}\,
       \mathbb E_\theta\|M_{w'_n,n-1}\|_{\mathrm{HS}}^2,
    \qquad |w'_n|=\lfloor u/2\rfloor.
\]
The retained sublist $w'_n$ is chosen from the fixed coloring and the updated position.
The bound is uniform over these choices, with $C_u$ depending only on $u$ and $\kappa$.

Summing the increment bounds by the martingale isometry contributes a factor of order $n/d$ at each halving.
Starting from the $\mathcal{O}(n)$ bound for a single group sum, lengths two and three both reduce to length one in a single step, explaining their common $\mathcal{O}(n^2/d)$ bound.
More generally, lengths $2^j\leq s<2^{j+1}$ undergo the same number of halvings.
Applying the coloring identity and Jensen's inequality yields
\[
    \mathbb E_\theta\|\mathcal A_{s,n}\|_{\mathrm{HS}}^2
       =\mathcal{O}_s\!\left(\frac{n^{j+1}}{d^j}\right),
       \qquad 2^j\leq s<2^{j+1}.
\]
The cycle estimate uses a chain of length $\ell-1$ and contributes one further factor of $T/d$, giving
\[
    \mathbb E_\theta|\mathcal C_{\ell,T}|^2
       =\mathcal{O}_\ell\!\left(\frac{T^{h+1}}{d^h}\right),
       \qquad 2^{h-1}<\ell\leq2^h.
\]

For a target order $2^{h-1}<t\leq2^h$, we choose the spectral path so that the first surviving blocks lie in this range.
Set $r=2^{h-1}+1$ and
\[
    \lambda_j(\theta)
       =\frac{1+\cos(2\pi j/r+\theta)}r,
    \qquad 0\leq j<r,\quad 0\leq\theta\leq\pi/r.
\]
The ensemble construction uses $r$ Gaussian vectors in place of the three in \eqref{eq:overview-gaussian-prior}.
Summing over the equally spaced phases cancels every Fourier frequency not divisible by $r$.
Thus moments below order $r$ remain constant along the path, whereas $p_t(\lambda(\theta))$ contains a nonzero $\cos(r\theta)$ term and differs between the endpoints.
Moment matching removes all derivative blocks of degree below $r$.
The first surviving range, $r\leq\ell\leq2^h$, therefore contributes
$\mathcal{O}_h(T^{(h+1)/2}/d^{h/2})$ to the derivative bound once $\kappa$ is fixed as below.

To complete the argument, set $\kappa=2(h+1)$ and
$T=\lfloor cd^{h/(h+1)}\rfloor$ with fixed $0<c\leq1$.
Choose $\alpha>0$ sufficiently small depending only on $h$, in particular with $\alpha\leq1/(8\kappa)$.
Let $G=\{H:\|\sum_{i=1}^T P_i\|_{\mathrm{op}}<\kappa\}$.
Two contributions remain: higher-degree blocks on $G$, and the endpoint probabilities of records outside $G$.

Beyond the leading dyadic range, the second-moment bounds gain additional powers of $T/d$.
Tracking the dependence of their constants on the block size controls the intermediate degrees.
For the longest blocks, a deterministic bound grows only exponentially with their size and is suppressed by the coefficients' $\alpha^\ell$ factors.
Together, these estimates make the higher-order contribution $o(1)$ at the chosen copy budget.

For exceptional records, we bound their probabilities under $P_\theta$, including the endpoint laws $P_0$ and $P_{\pi/r}$.
Before the projector sum reaches $\kappa$, a posterior covariance bound ensures that the conditional mean of the next projector is at most $2I/d$~\cite{Carlen2013BrascampLieb}.
A matrix exponential-moment argument based on the Golden--Thompson inequality~\cite{Forrester2014GoldenThompson} then gives
\[
    P_\theta(G^c)
    \leq d\left(\frac{2eT}{\kappa d}\right)^\kappa
    =\mathcal{O}_h(d^{-1}),
\]
uniformly in $\theta$.

Integrating the derivative on $G$ and adding the endpoint contribution
$[P_0(G^c)+P_{\pi/r}(G^c)]/2$ therefore yields
\[
    \operatorname{TV}(P_0,P_{\pi/r})
    \leq C_h\frac{T^{(h+1)/2}}{d^{h/2}}+o(1)
    \leq C_hc^{(h+1)/2}+o(1).
\]
Here $C_h$ depends only on $h$, with $\kappa$ and $\alpha$ fixed as above.
Taking $c$ sufficiently small contradicts the testing requirement and proves the matching lower bound.

\subsection{The role of adaptivity}
\label{sec:overview-nonadaptive}

The nonadaptive proof follows the same smooth-path and block-expansion framework.
Conditional on the latent vectors and random seed, the original, untruncated matrices $\widehat Y_i$ are independent as well as centered, even when the POVMs differ between rounds.
We therefore work directly with $\widehat Y_i$ and their cycle sums $\widehat{\mathcal C}_{\ell,T}$ to exploit this independence.
The auxiliary matrices $Y_i$ need not remain independent, since their truncation multipliers depend on previous outcomes.

Expand the squared cycle sum into products of two trace terms.
If their underlying index sets differ, a matrix appearing in only one term is independent of all the others and has mean zero, so the cross expectation vanishes.
For example, the cycles indexed by $(1,2,3,5)$ and $(1,2,4,5)$ have zero cross expectation after averaging over $\widehat Y_3$.
For adaptive measurements, the fifth measurement can depend on the third outcome, so this cancellation is no longer guaranteed.
Only pairs with the same index set can contribute in the nonadaptive case; counting these pairs and applying the one-step variance estimates gives
\[
    \mathbb E_\theta|\widehat{\mathcal C}_{\ell,T}|^2
    \leq\ell!2^\ell\frac{T^\ell}{d^{\ell-1}}.
\]
Inserting the indicator of $G$ can only decrease this second moment, so the same bound controls its contribution on good records without conditioning on $G$.

Choose the spectral path with $r=t$, so the moments through order $t-1$ agree while the target moments remain separated.
The first surviving block then has degree $t$ and contributes at most $\mathcal{O}_t(T^{t/2}/d^{(t-1)/2})$ to the derivative integral on $G$.
The stronger bound comes from cancellations between different index sets, beyond the martingale cancellations available in the adaptive analysis.
For cycles of length four, it improves the second-moment bound from $\mathcal{O}(T^3/d^2)$ to $\mathcal{O}(T^4/d^3)$.
Together with matching the first three spectral moments, this raises the fourth-moment lower-bound threshold from $d^{2/3}$ to $d^{3/4}$.

As in the adaptive proof, we must also control the higher-degree tail on $G$ and the endpoint probabilities of the exceptional records $G^c$.
Similar tail and concentration estimates make both contributions negligible, with a sufficiently large fixed threshold $\kappa$ and sufficiently small fixed $\alpha$.
Integrating the derivative on $G$ and adding the exceptional-record probabilities therefore gives endpoint total variation distance at most $C_tc^{t/2}+o(1)$ at $T=\lfloor cd^{1-1/t}\rfloor$.
Taking $c$ sufficiently small proves the nonadaptive lower bound in \Cref{thm:nonadaptive-moments}.

\subsection{Extensions to other nonlinear properties and spectra}
\label{sec:overview-extensions}

The four consequences in \Cref{cor:moment-consequences} follow by reducing the corresponding tasks to ordinary state-moment estimation.
For fixed moment order, rank, and local dimension, these reductions introduce only dimension-independent overheads.

\paragraph{Ordered overlaps of several states.}
The main difficulty is that $\operatorname{Tr}(\rho_0\cdots\rho_{t-1})$ can be complex, whereas ordinary state moments are real.
We therefore introduce known auxiliary states in the construction to select the desired ordering and supply a tunable phase.
For $t\geq3$, choose $A_j=\ket{a_j}\bra{a_j}$ with $\ket{a_j}=(\ket{j}+e^{i\phi_j}\ket{j+1})/\sqrt2$, where basis labels are taken modulo $t$.
Consider the mixtures $\Omega_S=|S|^{-1}\sum_{j\in S}\rho_j\otimes A_j$ for nonempty subsets $S\subseteq\{0,\ldots,t-1\}$.
An inclusion--exclusion combination of their $t$th moments cancels terms with repeated input labels.
The orthogonality of nonadjacent auxiliary states then leaves only the forward and reverse cyclic orderings, giving
\[
    t\,2^{1-t}\operatorname{Re}\!\left[
        e^{-i\sum_j\phi_j}
        \operatorname{Tr}(\rho_0\cdots\rho_{t-1})
    \right].
\]
Choosing the total phase to be $0$ or $\pi/2$ recovers the real or imaginary part, respectively.
For fixed $t$, the number of mixtures, their dimension enlargement, and the error amplification are constants, so \Cref{thm:optimal-moments} gives the claimed upper bound.

The auxiliary states need not be physically supplied in our single-copy POVM model.
After selecting source $j$, any measurement effect $E_y$ on the enlarged system can be replaced by $(I\otimes\bra{a_j})E_y(I\otimes\ket{a_j})$ on the original input, with exactly the same outcome probability.
Thus the protocol uses only source-dependent single-copy measurements and classical memory.
For $t=2$, ordinary overlap estimation reduces directly to purity estimation of the two inputs and their equal mixture.
The matching worst-case lower bound follows by setting all input states equal to $\rho$, so that the target becomes $p_t(\rho)$.

\paragraph{Weighted moments.}
For a known Hermitian observable $O$ with $\|O\|_{\mathrm{op}}\leq1$, apply the filter $K_x=\sqrt{(I+xO)/2}$, where $x\in[-1/2,1/2]$.
Its success probability $q_x=\operatorname{Tr}(K_x\rho K_x)$ lies in $[1/4,3/4]$, so obtaining the conditional state $\sigma_x=K_x\rho K_x/q_x$ incurs only constant copy overhead.
The quantity
\[
    F(x)=(2q_x)^t\operatorname{Tr}(\sigma_x^t)
    =\operatorname{Tr}\!\left[
        \bigl(\sqrt\rho(I+xO)\sqrt\rho\bigr)^t
      \right]
\]
is a polynomial of degree at most $t$ whose linear coefficient is $t\operatorname{Tr}(O\rho^t)$, by cyclicity of the trace.
Estimating the filtering probability and the ordinary $t$th moment at $t+1$ fixed values of $x$ therefore recovers the weighted moment by interpolation.
The interpolation and normalization require only dimension-independent precision, giving the same dimension exponent as in \Cref{thm:optimal-moments}.
The worst-case lower bound over $(O,\rho)$ follows immediately by taking $O=I$.

\paragraph{Partial-transpose moments.}
Suppose that $q=\dim(A)$ is fixed.
The transpose map on $A$ admits the decomposition
$\mathrm{T}_A=(q+1)\Phi_A-q\mathcal D_A$ into the two quantum channels
\[
    \Phi_A(X)=\frac{X^\mathrm{T}+\operatorname{Tr}(X)I_A}{q+1},
    \qquad
    \mathcal D_A(X)=\frac{\operatorname{Tr}(X)}q I_A.
\]
Write $\omega_1=(\Phi_A\otimes\operatorname{id}_B)(\rho_{AB})$
and $\omega_0=(\mathcal D_A\otimes\operatorname{id}_B)(\rho_{AB})$.
Then
\[
    \operatorname{Tr}[(\rho_{AB}^{\mathrm{T}_A})^t]
    =\operatorname{Tr}[((q+1)\omega_1-q\omega_0)^t].
\]
Expanding the power while preserving the order of the factors
expresses the target as a linear combination of $2^t$ ordered
overlaps of $\omega_0$ and $\omega_1$.
Each overlap can be estimated using the preceding protocol.
The absolute coefficients sum to $(2q+1)^t$, so both the number
of terms and the required error amplification depend only on
the fixed parameters $q,t$.
Each single-copy measurement on either output state can be
composed with the corresponding channel to give a measurement
on one original copy, preserving the access model.

The matching lower bound follows by restricting to product states
$\ket{0}\!\bra{0}_A\otimes\rho_B$, for which the target is
$\operatorname{Tr}(\rho_B^t)$ and $\dim(B)=d/q$.
If $B$ has fixed dimension instead, exchange the two subsystems,
using the fact that their partial transposes have identical spectra.

\paragraph{Spectrum estimation at fixed rank.}
For $\operatorname{rank}(\rho)\leq r$, pad its nonzero eigenvalues with zeros to form a sorted vector $\lambda\in\mathbb R^r$.
Newton's identities show that $p_1=1,p_2,\ldots,p_r$ uniquely determine this vector.
Moreover, the map $\lambda\mapsto(p_2,\ldots,p_r)$ is continuous and injective on the compact set of sorted probability vectors, so its inverse is uniformly continuous.
Consequently, for every fixed $r$ and target $\ell_1$ error $\epsilon$, estimating these moments to a sufficiently small accuracy $\delta(r,\epsilon)>0$ suffices to recover the spectrum, including at repeated or zero eigenvalues.
For example, one can select a sorted probability vector whose moments best fit the estimates.
Estimating all moments through order $r$ costs $\mathcal{O}(d^{\lceil\log_2 r\rceil/(1+\lceil\log_2 r\rceil)})$ copies, proving the claimed upper bound without an eigenvalue-gap assumption.
All filters and channels above can be composed with the subsequent measurement on the same input copy, so each reduction preserves the single-copy access model.

\subsection{Distributed alternating moments}
For \Cref{cor:distributed-alternating-moments}, Alice constructs the same pilot matrix $S_x$ from $n$ pilot trials on $\rho$, with the pilot-vector covariance $(I+x\rho)/(d+x)$, and sends its classical description to Bob.
Bob applies the corresponding filter to his copies of $\sigma$.
The same concentration argument shows that the filtered moments, after correcting for normalization using the filtering success probability, approximate
\[
    F_t(x)=\operatorname{Tr}\!\left[
        \bigl(\sigma+x\sqrt{\sigma}\rho\sqrt{\sigma}\bigr)^t
    \right].
\]
This is a degree-at-most-$t$ polynomial whose leading coefficient is
$[x^t]F_t(x)=\operatorname{Tr}[(\rho\sigma)^t]$.
Thus, estimates at $t+1$ fixed values of $x$ recover the target by interpolation.

The crucial difference from our single-copy protocol is that Bob may jointly measure his filtered copies.
At fixed order and accuracy, estimating their moments requires only a constant number of successful outputs, independently of the filtered dimension.
Since each filtering attempt succeeds with probability $\Theta(n/d)$, the total cost, including normalization estimation, is
\[
    \mathcal{O}(n+d/n)=\mathcal{O}(\sqrt d)
    \qquad\text{for }n=\lceil\sqrt d\rceil.
\]
For the lower bound, pure inputs satisfy
$\operatorname{Tr}[(\rho\sigma)^t]=|\langle\psi|\phi\rangle|^{2t}$.
An accurate estimator therefore distinguishes identical from independent Haar-random pure states, requiring $\Omega(\sqrt d)$ copies even with arbitrary local joint operations and classical communication~\cite{Anshu2022DistributedInnerProduct}.

\section{Related works}
\label{sec:related-work}

Randomized measurements provide a practical route to estimating nonlinear properties with single-copy access~\cite{Elben2023RandomizedMeasurementToolbox}.
By combining outcomes from separately measured copies in classical postprocessing, these protocols construct estimators for nonlinear functions of the state.
Closely related classical-shadow methods allow many properties to be predicted from a common measurement record~\cite{Huang2020ClassicalShadows}.
Applications include state moments~\cite{VanEnk2012RandomMeasurements,Elben2018RandomQuenches,Elben2019LocalRandomizedMeasurements} and weighted moments $\operatorname{Tr}(O\rho^t)$, with protocols developed for both quadratic~\cite{Du2026OptimalRandomized} and higher-order quantities~\cite{Li2026SingleSetting}.
Randomized measurements also give access to state overlaps $\operatorname{Tr}(\rho\sigma)$, enabling comparisons of quantum states prepared on different experimental platforms~\cite{Elben2020CrossPlatformVerification}.
For correlation analysis and entanglement detection, they enable the estimation of the correlation overlap $\operatorname{Tr}[\rho_{AB}(\rho_A\otimes\rho_B)]$~\cite{Liu2022Correlation}, partial-transpose moments $\operatorname{Tr}[(\rho^{\mathrm{T}_A})^t]$~\cite{Zhou2020Negativity,Elben2020MixedStateEntanglement}, and more general permutation moments~\cite{Liu2022PermutationMoments}.
Further applications include measurable polynomial lower bounds on quantum Fisher information for quantum metrology~\cite{Rath2021QuantumFisherInformation,Yu2021RandomizedQuantumFisherInformation} and nonlinear many-body topological invariants for characterizing quantum phases~\cite{Elben2020TopologicalInvariants}.

The availability of such protocols, however, does not guarantee a small sample complexity.
For purity estimation, exponential separations between single-copy and joint measurements were established by Aharonov, Cotler, and Qi and by Chen, Cotler, Huang, and Li~\cite{Aharonov2022QuantumAlgorithmicMeasurement,Chen2022QuantumMemory}.
Huang et al.\ established similar separations for quantum principal component analysis, including the nonlinear task of estimating observables in the leading eigenstate of an unknown mixed state~\cite{Huang2022LearningExperiments}.
Ye, Liu, and Deng further showed that, for certain nonlinear functionals of fixed degree, access to just one additional replica can reduce the sample complexity at constant precision from exponential to constant in the number of qubits~\cite{Ye2025ReplicaAdvantage}.
Together, these results show that the number of copies that can be coherently processed together can fundamentally change the complexity of nonlinear property estimation.

Further work explores more detailed resource tradeoffs and broader estimation tasks.
Gong et al.\ established upper and lower bounds for purity estimation using two-copy measurements with limited quantum memory, quantifying the dependence on memory size and precision~\cite{Gong2026PurityInnerProduct}.
For low-rank states, access to copies of a purification can yield exponential savings over single-copy measurements of the target mixed state~\cite{Liu2024Purification}.
In the stronger model of coherent access to a purification-preparation oracle, estimating state moments admits a quadratic improvement in the dependence on moment order over sample-based estimation at fixed precision~\cite{Zhang2026MeasuringLess}.
The range of estimation tasks has also expanded beyond individual polynomial quantities.
Protocols using collective measurements or reusable quantum registers address the simultaneous estimation of multiple nonlinear quantities~\cite{Chen2026Simultaneous,Shi2025Simultaneous}, while studies allowing collective measurements analyze the complexity of entropy estimation~\cite{Acharya2020QuantumEntropy,Chen2026RenyiTsallis} and noninteger trace powers~\cite{Chen2026TracePowers}.

\section{Open problems}
\label{sec:discussion}

Our results determine the optimal dimension dependence of state-moment estimation with adaptive and nonadaptive single-copy measurements.
The comparison shows that classical adaptation is itself a useful resource: it makes the third and fourth moments equally costly, although their nonadaptive complexities differ.

A natural next question is how this picture changes when a small quantum memory is allowed to connect successive copies.
Ancillary qubits retained between rounds permit limited coherent interaction across copies and therefore go beyond the single-copy POVM model studied here.
Tradeoffs between memory and sample complexity have been established for Pauli-observable estimation and purity testing~\cite{Chen2024PauliTradeoffs,Gong2026PurityInnerProduct}.
These results motivate extending the present hierarchy to higher moments under explicit memory constraints.
In particular, it would be interesting to determine whether the relation between the complexities of third- and fourth-moment estimation changes at intermediate memory sizes.

We hope that our filtering protocols and lower-bound techniques will also find applications in other quantum learning problems.
Our analysis of filtering focuses on additive error; a natural question is how it performs under relative error and whether it can improve on existing classical-shadow-based bounds for moment estimation~\cite{Pelecanos2026UnentangledSpectrum}.
Combining filtering with classical shadows~\cite{Huang2020ClassicalShadows} may also broaden the range of efficiently estimable properties, particularly for tasks involving the simultaneous estimation of many nonlinear observables~\cite{Chen2026Simultaneous}.
Finally, it would be interesting to extend our smooth comparison of ensembles to channel-property learning, including unitarity estimation~\cite{Chen2023UnitarityEstimation} and Choi-state moments under restricted probe and ancilla access.

\section{Preliminaries}
\label{sec:proof-preliminaries}

We collect the tools used below and then specify the learning and testing models.
All asymptotic statements concern $d\to\infty$, unless another variable is specified.
For nonnegative functions $f,g$, with $g$ eventually positive, $f=\mathcal{O}(g)$ means $f(d)\leq Cg(d)$ for all sufficiently large $d$ and some constant $C>0$; $f=\Omega(g)$ means $f(d)\geq cg(d)$ under the same convention for some $c>0$.
We write $f=\Theta(g)$ when both bounds hold, and $f=o(g)$ when $f(d)/g(d)\to0$; in particular, $o(1)$ denotes a quantity tending to zero.
For signed or complex quantities, these definitions apply to their absolute values.
Subscripts specify permitted dependence of the hidden constants and dimension thresholds: for example, $\mathcal{O}_{t,\epsilon,\delta}(d^a)$ allows dependence on $t,\epsilon,\delta$, but not on $d$.
Without subscripts, constants may depend on the fixed moment order, accuracy, failure probability, and other parameters explicitly held fixed; bounds stated uniformly over states or measurement strategies have no dependence on those choices.

A state on $\mathbb C^d$ is a positive semidefinite matrix $\rho$ with $\operatorname{Tr}\rho=1$.
We write $p_j(\rho)=\operatorname{Tr}(\rho^j)$, $\|A\|_{\mathrm{HS}}=(\operatorname{Tr}A^\dagger A)^{1/2}$, and $\|A\|_{\mathrm{op}}=\|A\|_\infty$ for the operator norm.
The order $A\preceq B$ means that $B-A$ is positive semidefinite.
For a probability vector $\lambda$, write $p_j(\lambda)=\sum_i\lambda_i^j$; for a polynomial $F$, $[x^j]F(x)$ denotes its coefficient of $x^j$.
All logarithms are natural unless their base is displayed.
Moment orders, accuracies, and failure probabilities are fixed independently of $d$.
We use $T$ for the total copy budget and $n,N$ for local sample counts or prefix lengths as specified in each argument.

\subsection{Spectral estimates and polynomial reconstruction}

\begin{lemma}[Continuity and normalization of moments]
\label[lemma]{pre:moment-lipschitz}
For states $A,B$ and every integer $m\geq2$,
$|p_m(A)-p_m(B)|\leq m\|A-B\|_{\mathrm{HS}}$.
If $B,H\succeq0$, $1\leq\operatorname{Tr}H\leq2$, $0<\eta<1/2$, and both $\|B-H\|_{\mathrm{HS}}$ and $|\operatorname{Tr}(B-H)|$ are at most $\eta$, then
\[
 \operatorname{Tr}B\geq\tfrac12,\qquad
 \left\|\frac{B}{\operatorname{Tr}B}
            -\frac{H}{\operatorname{Tr}H}\right\|_{\mathrm{HS}}
 \leq4\eta.
\]
\end{lemma}
\begin{proof}
Expand $A^m-B^m=\sum_{j=0}^{m-1}A^{m-1-j}(A-B)B^j$, take traces, and apply Hilbert--Schmidt Cauchy--Schwarz.
Each factor $B^jA^{m-1-j}$ has Hilbert--Schmidt norm at most one because $m\geq2$ and states have both operator and Hilbert--Schmidt norm at most one.
For the second assertion, write $b=\operatorname{Tr}B$ and $h=\operatorname{Tr}H$.
Since $b\geq h-\eta\geq1/2$ and $\|H\|_{\mathrm{HS}}\leq h$,
\[
 \left\|B/b-H/h\right\|_{\mathrm{HS}}
 \leq \frac{\|B-H\|_{\mathrm{HS}}}{b}
       +\frac{|h-b|\|H\|_{\mathrm{HS}}}{bh}
 \leq4\eta.
\]
\end{proof}

We also use that $VV^\dagger$ and $V^\dagger V$ have the same nonzero eigenvalues, including multiplicities; this follows from the singular-value decomposition.

\begin{lemma}[Interpolation at fixed degree]
\label[lemma]{pre:interpolation}
Fix distinct real nodes $x_0,\ldots,x_m$ and a polynomial $F$ of degree at most $m$.
If $|\widehat F_i-F(x_i)|\leq\eta$ at every node, every coefficient of $F$ and its value at any fixed real $z$ can be estimated with error at most $C\eta$, where $C$ depends only on the nodes, degree, and evaluation point.
\end{lemma}
\begin{proof}
Set $L_i(x)=\prod_{j\ne i}(x-x_j)/(x_i-x_j)$ and $\widehat F(x)=\sum_i\widehat F_iL_i(x)$.
The errors in coefficient $k$ and in evaluation at $z$ are at most $\eta\sum_i|[x^k]L_i(x)|$ and $\eta\sum_i|L_i(z)|$, respectively.
These constants are finite, including for extrapolation outside the node interval.
\end{proof}

\begin{lemma}[Recovering a fixed-size spectrum]
\label[lemma]{pre:spectrum-continuity}
Let $r\geq2$ be fixed and let $\Delta_r^\downarrow$ be the set of sorted probability vectors of length $r$.
For every $\epsilon>0$ there exists $\delta>0$, depending only on $r,\epsilon$, such that
\[
 \lambda,\mu\in\Delta_r^\downarrow,\qquad
 \max_{2\leq j\leq r}|p_j(\lambda)-p_j(\mu)|<\delta
 \quad\Longrightarrow\quad \|\lambda-\mu\|_1<\epsilon.
\]
\end{lemma}
\begin{proof}
Let $e_k$ be the $k$th elementary symmetric polynomial in the entries of $\lambda$, with $e_0=1$.
Newton's identities $ke_k=\sum_{j=1}^k(-1)^{j-1}e_{k-j}p_j$, together with $p_1=1$, show that $p_2,\ldots,p_r$ determine the polynomial $\prod_i(z-\lambda_i)$.
The polynomial determines the eigenvalue multiset, and sorting removes the permutation ambiguity.
Thus the power-sum map is continuous and injective on the compact set $\Delta_r^\downarrow$, and its inverse on its image is uniformly continuous.
This argument includes zero and repeated eigenvalues.
\end{proof}

\subsection{Haar and Gaussian identities}

\begin{lemma}[Moments of a Haar vector]
\label[lemma]{pre:haar}
For a Haar-random unit vector $v\in\mathbb C^d$, $P_v=vv^\dagger$, and integer $k\geq1$,
\[
 \mathbb E P_v^{\otimes k}
 =\frac{\sum_{\pi\in S_k}V_\pi}{d(d+1)\cdots(d+k-1)},
\]
where $V_\pi$ permutes the tensor factors~\cite[Prop.~6]{Harrow2013SymmetricSubspace}.
In particular, for $z=v^\dagger\rho v$,
\begin{align*}
 \mathbb Ez^2&=\frac{1+p_2(\rho)}{d(d+1)},&
 \mathbb Ez^3&=\frac{1+3p_2(\rho)+2p_3(\rho)}{d(d+1)(d+2)},\\
 d\mathbb E[zP_v]&=\frac{I+\rho}{d+1}.&&
\end{align*}
\end{lemma}
\begin{proof}
Write a standard complex Gaussian vector as $g=Rv$, with independent radius and direction.
Gaussian pairing, the complex form of the Wick--Isserlis formula~\cite{Isserlis1918GaussianMoments}, gives $\mathbb E(gg^\dagger)^{\otimes k}=\sum_{\pi\in S_k}V_\pi$, while $R^2$ has the Gamma distribution of shape $d$ and $\mathbb ER^{2k}=d(d+1)\cdots(d+k-1)$.
Dividing yields the tensor identity.
Contracting against $\rho^{\otimes k}$ gives the scalar formulas, with each permutation contributing the moments indexed by its cycles; a partial trace of the $k=2$ identity gives the matrix formula.
\end{proof}

\begin{lemma}[Complex Gaussian quadratic forms]
\label[lemma]{pre:gaussian}
Let $g_1,\ldots,g_r$ be independent standard complex Gaussian vectors in $\mathbb C^d$, each with density $\pi^{-d}e^{-\|g_j\|^2}$.
Then $\mathbb E g_jg_j^\dagger=I$ and, for Hermitian $M$ with $\|M\|_{\mathrm{op}}<1$,
\[
 \mathbb E e^{g_j^\dagger M g_j}=\det(I-M)^{-1}.
\]
For $U=\sum_j\|g_j\|^2$ and integer $m\geq1$,
$\mathbb EU^m=\prod_{a=0}^{m-1}(rd+a)\leq(rd+m)^m$.
For fixed $r$, the Gram matrix $(g_j^\dagger g_k/d)_{j,k=1}^r$ converges in probability to $I_r$.
\end{lemma}
\begin{proof}
Diagonalizing $M$ reduces the determinant identity to one-dimensional Gaussian integrals, each equal to $(1-\lambda)^{-1}$.
Squared Gaussian coordinates are independent exponential variables with mean one, so $U$ is Gamma distributed with shape $rd$, giving its moments.
After normalization by $d$, each diagonal Gram entry has mean one and variance $1/d$, and each off-diagonal entry has mean zero and squared-modulus expectation $1/d$.
Chebyshev's inequality and a union bound prove the last assertion.
The determinant identity also holds analytically near the origin for complex linear combinations of Hermitian matrices, by analytic continuation.
\end{proof}

\begin{lemma}[Gaussian block expansion]
\label[lemma]{pre:gaussian-blocks}
Let $X_1,\ldots,X_N$ be fixed traceless Hermitian matrices on
$\mathbb C^d$, and let $g_1,\ldots,g_r$ be independent standard
complex Gaussian vectors.
For a probability vector $\lambda=(\lambda_1,\ldots,\lambda_r)$
and $\alpha>0$, define
\[
    W_i=\alpha\sum_{j=1}^r\lambda_j g_j^\dagger X_i g_j.
\]
For each nonempty subset $J\subseteq\{1,\ldots,N\}$ of size $\ell$,
define
\begin{equation}
\label{pre:gaussian-block-coefficient}
    b_J=\frac{\alpha^\ell p_\ell(\lambda)}{\ell}
       \sum_{\substack{(i_1,\ldots,i_\ell)\\
                       \text{an ordering of }J}}
          \operatorname{Tr}(X_{i_1}\cdots X_{i_\ell}),
    \qquad p_\ell(\lambda)=\sum_{j=1}^r\lambda_j^\ell.
\end{equation}
Then, for every $S\subseteq\{1,\ldots,N\}$,
\begin{equation}
\label{pre:gaussian-block-moments}
    \mathbb E_G\prod_{i\in S}W_i
       =\sum_{\pi\in\mathcal P(S)}\prod_{J\in\pi}b_J,
\end{equation}
where $\mathbb E_G$ denotes expectation over the Gaussian vectors
and $\mathcal P(S)$ is the set of partitions of $S$ into nonempty
disjoint blocks.
For $S=\varnothing$, both sides equal one.
Singleton blocks have coefficient zero because
$\operatorname{Tr}X_i=0$.
\end{lemma}

\begin{proof}
By independence of the Gaussian vectors and
\Cref{pre:gaussian}, for $z_1,\ldots,z_N$ sufficiently close to zero,
\begin{align}
\label{pre:gaussian-log-mgf}
    \log\mathbb E_G\exp\!\left(\sum_{i=1}^N z_iW_i\right)
      &=-\sum_{j=1}^r
        \log\det\!\left(I-\alpha\lambda_j\sum_{i=1}^N z_iX_i\right)
        \notag\\
      &=\sum_{\ell\geq1}
        \frac{\alpha^\ell p_\ell(\lambda)}{\ell}
        \operatorname{Tr}\!\left(\sum_{i=1}^N z_iX_i\right)^\ell.
\end{align}
The second equality uses the local power series for
$-\log\det(I-M)$.
For $|J|=\ell$, the coefficient of $\prod_{i\in J}z_i$
on the right is exactly $b_J$: each index in $J$ must occur once,
and every ordering contributes its corresponding trace.

Now exponentiate \eqref{pre:gaussian-log-mgf} and compare the
coefficient of $\prod_{i\in S}z_i$.
On the left, this coefficient is
$\mathbb E_G\prod_{i\in S}W_i$.
On the right, it is obtained by selecting monomials whose index
sets are disjoint and together cover $S$, hence by choosing a
partition of $S$.
For a partition with $k$ blocks, the $k!$ orders of selecting
those blocks cancel the factor $1/k!$ in the exponential series.
Its contribution is therefore $\prod_{J\in\pi}b_J$.
Summing over partitions proves \eqref{pre:gaussian-block-moments}.
Only coefficients of finite degree in a neighborhood of the
origin are used.
\end{proof}

\subsection{Martingales, concentration, and second moments}

We use martingales to express the conditional cancellations needed below, following the standard definitions and $L^2$ theory~\cite[Chs.~10 and 12]{Williams1991ProbabilityMartingales}.
A filtration $(\mathcal F_i)_{i=0}^N$ is an increasing sequence of sigma-algebras representing the information available through round $i$.
A process $(M_i)$ of random vectors or matrices is adapted if each $M_i$ is determined by $\mathcal F_i$.
An adapted process is a martingale if $\mathbb E\|M_i\|<\infty$ and
\[
    \mathbb E[M_i\mid\mathcal F_{i-1}]=M_{i-1}
    \qquad (i\geq1),
\]
where matrix conditional expectations are taken entrywise.

Its increments $D_i=M_i-M_{i-1}$ are therefore martingale differences: they are determined by $\mathcal F_i$ and have conditional mean zero given $\mathcal F_{i-1}$.
Conversely, summing such differences starting from an integrable $\mathcal F_0$-measurable variable gives a martingale.
Thus the next increment has zero average once the past is fixed, even though increments need not be independent.

A scalar or matrix factor used at round $i$ is called predictable when it is determined by $\mathcal F_{i-1}$, before that round's outcome is observed.
For example, if $V_i$ is an outcome statistic from an adaptive measurement, then $D_i=V_i-\mathbb E[V_i\mid\mathcal F_{i-1}]$ is a martingale difference whenever $V_i$ is integrable.
The measurement may depend on the past; subtracting its conditional mean removes only the predictable part of the next observation.
Below, all martingales are square-integrable, meaning that their squared norms have finite expectation.

\begin{lemma}[Scalar concentration]
\label[lemma]{pre:binomial}
If $B\sim\operatorname{Bin}(N,p)$, then $\Pr(B\leq Np/2)\leq e^{-Np/8}$.
Consequently, when $p\geq p_*>0$, a budget $N\geq8p_*^{-1}(M+\log(1/\delta))$ supplies at least $M$ successes with probability at least $1-\delta$.
For independent $[0,1]$-valued observations with common mean $p$, Hoeffding's inequality~\cite{Hoeffding1963ProbabilityInequalities} gives $\Pr(|\widehat p-p|\geq\eta)\leq2e^{-2N\eta^2}$ for their sample mean.
\end{lemma}
\begin{proof}
The first bound follows by applying Markov's inequality to $e^{-sB}$ at $s=\log2$, using $\mathbb Ee^{-sB}\leq\exp[Np(e^{-s}-1)]$; the stated budget makes both $Np/2\geq M$ and $e^{-Np/8}\leq\delta$.
For the last bound, a variable $X\in[0,1]$ satisfies $\mathbb Ee^{s(X-\mathbb EX)}\leq e^{s^2/8}$: the second derivative of its log moment-generating function is the variance under an exponential tilt and is at most $1/4$.
Independence, Markov's inequality, and optimization over positive and negative $s$ give the result.
\end{proof}

\begin{lemma}[Positive-matrix concentration with conditional mean bounds]
\label[lemma]{pre:matrix-chernoff}
Let $Z_1,\ldots,Z_N$ be Hermitian $d\times d$ matrices adapted to a filtration $(\mathcal F_i)_{i=0}^N$, such that $0\preceq Z_i\preceq I$ and $\mathbb E[Z_i\mid\mathcal F_{i-1}]\preceq\mu I$.
For $L>N\mu>0$,
\[
 \Pr\!\left(\left\|\sum_{i=1}^NZ_i\right\|_{\mathrm{op}}\geq L\right)
 \leq d\left(\frac{eN\mu}{L}\right)^L.
\]
\end{lemma}

\begin{proof}
We first control the conditional exponential moment of each increment, then iterate this estimate and apply Markov's inequality.
For $s>0$ and $x\in[0,1]$, convexity of the exponential gives $e^{sx}\leq1+(e^s-1)x$.
Applying this scalar inequality to the eigenvalues of $Z_i$ yields
\[
    e^{sZ_i}\preceq I+(e^s-1)Z_i.
\]
Taking conditional expectation and using the assumed conditional mean bound, we obtain
\[
    \mathbb E[e^{sZ_i}\mid\mathcal F_{i-1}]
    \preceq [1+\mu(e^s-1)]I
    \preceq e^{\mu(e^s-1)}I.
\]

Write $S_i=\sum_{j=1}^iZ_j$ and $S_0=0$.
Although $S_{i-1}$ and $Z_i$ need not commute, the Golden--Thompson inequality~\cite{Forrester2014GoldenThompson} gives
$\operatorname{Tr}e^{s(S_{i-1}+Z_i)}
\leq\operatorname{Tr}(e^{sS_{i-1}}e^{sZ_i})$.
Since $S_{i-1}$ is determined by $\mathcal F_{i-1}$,
\[
 \mathbb E[\operatorname{Tr}e^{sS_i}\mid\mathcal F_{i-1}]
 \leq\operatorname{Tr}\!\left(
       e^{sS_{i-1}}\mathbb E[e^{sZ_i}\mid\mathcal F_{i-1}]
     \right)
 \leq e^{\mu(e^s-1)}\operatorname{Tr}e^{sS_{i-1}}.
\]
Here the last inequality uses positivity of $e^{sS_{i-1}}$: if $A\succeq0$ and $B\preceq cI$, then $\operatorname{Tr}(AB)\leq c\operatorname{Tr}A$.
Taking expectations and iterating from $\operatorname{Tr}e^{sS_0}=d$ gives
\[
    \mathbb E\operatorname{Tr}e^{sS_N}
    \leq d\,e^{N\mu(e^s-1)}.
\]

Because $S_N\succeq0$, the event $\|S_N\|_{\mathrm{op}}\geq L$ implies $\operatorname{Tr}e^{sS_N}\geq e^{sL}$.
Markov's inequality therefore yields
\[
    \Pr(\|S_N\|_{\mathrm{op}}\geq L)
    \leq e^{-sL}\mathbb E\operatorname{Tr}e^{sS_N}
    \leq d\,e^{N\mu(e^s-1)-sL}.
\]
The assumption $L>N\mu$ allows the positive choice $s=\log(L/(N\mu))$.
Substitution gives
\[
    \Pr(\|S_N\|_{\mathrm{op}}\geq L)
    \leq d\,e^{-N\mu}\left(\frac{eN\mu}{L}\right)^L
    \leq d\left(\frac{eN\mu}{L}\right)^L.
\]
The argument uses conditional mean bounds throughout and does not require independence.
\end{proof}

\begin{lemma}[Martingale isometry and matrix products]
\label[lemma]{pre:orthogonality}
Let $D_1,\ldots,D_N$ be square-integrable martingale differences with respect to $(\mathcal F_i)_{i=0}^N$, taking values in a finite-dimensional real or complex Hilbert space.
Let $M_0$ be square-integrable and determined by $\mathcal F_0$.
For $M_N=M_0+\sum_{i=1}^ND_i$,
\[
    \mathbb E\|M_N\|^2
    =\mathbb E\|M_0\|^2+\sum_{i=1}^N\mathbb E\|D_i\|^2.
\]
This applies to matrices with the Hilbert--Schmidt norm and does not require $M_0$ to have mean zero.
For matrix-valued differences and predictable matrices $L_i,R_i$, the products $L_iD_iR_i$ are also martingale differences whenever they are integrable; the same holds for predictable scalar multiples.
If the transformed increments are square-integrable, the same isometry applies to their sums.
\end{lemma}

\begin{proof}
The martingale-difference condition is $\mathbb E[D_j\mid\mathcal F_{j-1}]=0$.
For $i<j$, the earlier increment $D_i$ is determined by $\mathcal F_{j-1}$, so conditioning on this information gives
\[
    \mathbb E\langle D_i,D_j\rangle
    =\mathbb E\!\left[
       \left\langle D_i,\mathbb E[D_j\mid\mathcal F_{j-1}]\right\rangle
      \right]
    =0.
\]
Similarly, $M_0$ is determined by $\mathcal F_{j-1}$, and hence $\mathbb E\langle M_0,D_j\rangle=0$.
All these inner products are integrable by Cauchy--Schwarz and the square-integrability assumptions.

Expanding the squared norm now gives
\[
\begin{aligned}
    \mathbb E\|M_N\|^2
    =
    \mathbb E\|M_0\|^2+\sum_{i=1}^N\mathbb E\|D_i\|^2
    +2\operatorname{Re}\sum_{i=1}^N\mathbb E\langle M_0,D_i\rangle
    +2\operatorname{Re}\sum_{i<j}\mathbb E\langle D_i,D_j\rangle.
\end{aligned}
\]
The two sums of cross terms vanish, proving the identity.
Thus orthogonality in expectation is sufficient; the increments need not be independent.
For matrices, the Hilbert-space inner product is $\langle A,B\rangle=\operatorname{Tr}(A^\dagger B)$, whose induced norm is the Hilbert--Schmidt norm.

Finally, predictability means that $L_i,R_i$ are determined by $\mathcal F_{i-1}$.
They can therefore be taken outside the conditional expectation:
\[
    \mathbb E[L_iD_iR_i\mid\mathcal F_{i-1}]
    =L_i\,\mathbb E[D_i\mid\mathcal F_{i-1}]\,R_i=0.
\]
For unbounded predictable factors, first apply this identity on the past-measurable events where their norms are bounded, then pass to the limit using integrability of $L_iD_iR_i$.
The transformed increments are adapted and have conditional mean zero, as required.
The argument for predictable scalar multiples is identical.
\end{proof}

\begin{lemma}[Covariance of strongly log-concave measures]
\label[lemma]{pre:bras-lieb}
Suppose a probability measure on $\mathbb R^N$ has density proportional to $e^{-V(x)}$, where $V:\mathbb R^N\to\mathbb R\cup\{+\infty\}$ and $V(x)-a\|x\|^2/2$ is lower semicontinuous and convex for some $a>0$.
Then
\[
    \operatorname{Cov}(x)\preceq a^{-1}I.
\]
This includes extended-valued potentials imposing a convex support constraint.
For a mean-zero complex vector $g$, if its real representation $(\operatorname{Re}g,\operatorname{Im}g)$ has covariance bounded by $a^{-1}I$, then
\[
    \mathbb E[gg^\dagger]\preceq2a^{-1}I.
\]
\end{lemma}

\begin{proof}
We first prove the real covariance bound for smooth, finite-valued potentials, then justify the extension to nonsmooth and extended-valued potentials.
If $V$ is smooth, convexity of $V(x)-a\|x\|^2/2$ implies $\nabla^2V(x)\succeq aI$.
The Brascamp--Lieb variance inequality~\cite{Carlen2013BrascampLieb} therefore gives, for any real vector $u$,
\[
    \operatorname{Var}(u^{\mathsf T}x)
    \leq\mathbb E\!\left[u^{\mathsf T}(\nabla^2V(x))^{-1}u\right]
    \leq a^{-1}\|u\|^2.
\]
Since $\operatorname{Var}(u^{\mathsf T}x)=u^{\mathsf T}\operatorname{Cov}(x)u$, this proves the matrix inequality in the smooth case.

For the general case, write $W(x)=V(x)-a\|x\|^2/2$.
We construct smooth convex approximations to $W$ while retaining the quadratic term exactly.
For $\varepsilon>0$, define the quadratic regularization
\[
    W_\varepsilon(x)
    =\inf_{y\in\mathbb R^N}
       \left\{W(y)+\frac{\|x-y\|^2}{2\varepsilon}\right\}.
\]
For a proper lower-semicontinuous convex function, $W_\varepsilon$ is finite and convex, and $W_\varepsilon(x)\uparrow W(x)$ as $\varepsilon\downarrow0$, including at points where $W(x)=+\infty$.
Smooth $W_{1/k}$ by convolution with a symmetric, nonnegative, smooth kernel of integral one and sufficiently small compact support.
The support radius can be chosen so that the resulting smooth convex function $W_k$ satisfies
$|W_k-W_{1/k}|\leq1/k$ on $\{\|x\|\leq k\}$.
Consequently, $W_k(x)\to W(x)$ pointwise.

To justify convergence of moments, choose an affine lower bound $W(x)\geq b^{\mathsf T}x+c$, which exists for a proper lower-semicontinuous convex function.
Completing the square in the definition of $W_\varepsilon$ gives
\[
    W_\varepsilon(x)
    \geq b^{\mathsf T}x+c-\frac{\varepsilon}{2}\|b\|^2.
\]
Convolution with a symmetric kernel preserves affine functions, so the approximations above satisfy the common bound
\[
    W_k(x)\geq b^{\mathsf T}x+c-\frac12\|b\|^2.
\]
Set $V_k(x)=a\|x\|^2/2+W_k(x)$.
The common lower bound implies, for a finite constant $C$ independent of $k$,
\[
    e^{-V_k(x)}
    \leq C e^{-a\|x\|^2/4}.
\]
The right-hand side remains integrable after multiplication by $1+\|x\|^2$.
Dominated convergence therefore gives convergence of the normalizing constants and of all unnormalized first and second moments.
The limiting normalizing constant is positive by the assumed probability-density representation.
Thus the means and covariance matrices of the normalized distributions also converge.

Each $V_k$ is smooth and satisfies $\nabla^2V_k\succeq aI$, so its covariance is bounded by $a^{-1}I$ by the first part of the proof.
Passing to the limit proves the same bound for $V$.
This also handles a convex support constraint: where $V=+\infty$, the limiting density is simply zero.

Finally, let $g$ be mean zero and let $u$ be any complex vector.
The real and imaginary parts of $u^\dagger g$ are real linear functionals of $(\operatorname{Re}g,\operatorname{Im}g)$, each with coefficient vector of Euclidean norm $\|u\|$.
Both have mean zero, so the assumed real covariance bound yields
\[
    \mathbb E|u^\dagger g|^2
    =\operatorname{Var}(\operatorname{Re}u^\dagger g)
       +\operatorname{Var}(\operatorname{Im}u^\dagger g)
    \leq2a^{-1}\|u\|^2.
\]
Since $\mathbb E|u^\dagger g|^2=u^\dagger\mathbb E[gg^\dagger]u$, the complex matrix bound follows.
\end{proof}
\subsection{Single-copy protocols, transcripts, and testing}
\label{pre:learning-model}

An adaptive single-copy protocol measures a fresh copy of the same unknown state in each round and retains only classical information between copies.
Condition on an input-independent classical random seed $z$, which may contain all the protocol's random choices.
At round $i$, after the outcomes $h_{<i}=(y_1,\ldots,y_{i-1})$, the protocol chooses a POVM $\{E_{i,y}(z,h_{<i})\}_y$.
Thus the probability of a complete outcome sequence $h=(y_1,\ldots,y_T)$, conditional on $z$, is
\begin{equation}
\label{pre:transcript-law}
 p_\rho(h\mid z)
 =\prod_{i=1}^T\operatorname{Tr}\!\left[
     E_{i,y_i}(z,h_{<i})\rho\right].
\end{equation}
A nonadaptive protocol fixes the entire POVM sequence conditional on $z$, so its effects do not depend on $h_{<i}$.
Measurements may differ between rounds and share randomness; all final classical processing is unrestricted.
The copy budget includes unsuccessful filtering attempts.
We use a fixed worst-case budget, padding early termination with ignored measurements when necessary.
In the nonadaptive model these measurements are taken from the preselected schedule.

A channel or filter followed immediately by an output measurement is allowed: the composition is a POVM on one original copy, with failures included as outcomes.
For lower bounds, diagonalize each POVM effect and record its spectral label to refine it into rank-one effects $w_yP_y$, where $P_y=v_yv_y^\dagger$, $\|v_y\|=1$, and $\sum_yw_yP_y=I$.
Forgetting the extra label recovers the original outcome, so refinement only provides more information and preserves nonadaptivity.
The formulas extend to continuous POVMs by replacing sums with integrals and using the corresponding effect densities.

For an ensemble $\Pi$, a state $\rho\sim\Pi$ is drawn \emph{once}, and all input copies are copies of that state.
Its transcript law therefore satisfies
\begin{equation}
\label{pre:ensemble-transcript}
 p_\Pi(h\mid z)=\int p_\rho(h\mid z)\,\Pi(d\rho)
 =\int\prod_{i=1}^T\operatorname{Tr}\!\left[
     E_{i,y_i}(z,h_{<i})\rho\right]\Pi(d\rho).
\end{equation}
In particular, this is an average of products, not a product of averages.
We include $z$ in the full transcript $H=(z,h)$ and denote its law by $P_\Pi^{(T)}$.
An individual execution returns the transcript $H$; its probability law describes repeated executions and is the object compared in the analysis.
The adaptive measurement tree is simply the history-dependent choice of effects in \eqref{pre:transcript-law}.
Following the learning-tree approach~\cite{Chen2022QuantumMemory}, our lower bound compares the transcript laws of two ensembles; the learner's final answer is classical postprocessing of this transcript.

For probability measures $P,Q$ with densities $p,q$ relative to a common dominating measure $\mu$, define
$\operatorname{TV}(P,Q)=\tfrac12\int|p-q|\,d\mu=\sup_A|P(A)-Q(A)|$.

\begin{lemma}[Le Cam's two-point method]
\label[lemma]{pre:testing}
The largest equal-prior success probability for distinguishing $P$ from $Q$, allowing randomized tests, is $\tfrac12(1+\operatorname{TV}(P,Q))$.
Classical postprocessing cannot increase total variation.
Consequently, suppose two state ensembles $\Pi_0,\Pi_1$ satisfy
\[
 \Pr_{\rho\sim\Pi_j}(|F(\rho)-a_j|\leq b)\geq1-\eta_j,
 \qquad a_1-a_0>2(b+\epsilon).
\]
If an estimator of $F$ has error at most $\epsilon$ with probability at least $1-\delta$ on every state, their transcript laws must satisfy
\[
 \operatorname{TV}(P_{\Pi_0}^{(T)},P_{\Pi_1}^{(T)})
 \geq1-2\delta-\eta_0-\eta_1.
\]
\end{lemma}
\begin{proof}
A test choosing $P$ on a measurable set $A$ succeeds with probability $\tfrac12(1+P(A)-Q(A))$.
Maximizing over $A$ gives the first formula; randomized tests are mixtures of such decisions and cannot improve it.
Postprocessing followed by a test is itself a test on the original data, proving contraction.
For estimation, threshold the answer at $(a_0+a_1)/2$.
On the concentration and estimation events this is correct, so its average success is at least $1-\delta-(\eta_0+\eta_1)/2$.
The testing formula gives the claim, the usual two-point reduction~\cite{LeCam1986,Tsybakov2009}.
\end{proof}

\begin{lemma}[Total variation along a smooth family]
\label[lemma]{pre:tv-path}
Let $p_\theta$, $\theta\in[u,v]$, be probability densities relative to a common measure $\mu$, and denote the corresponding laws by $P_\theta$.
Suppose $\theta\mapsto p_\theta(H)$ is absolutely continuous for $\mu$-almost every $H$, its derivative is jointly measurable, and $\int_u^v\|\partial_\theta p_\theta\|_1\,d\theta<\infty$, where $\|\partial_\theta p_\theta\|_1=\int|\partial_\theta p_\theta(H)|\,\mu(dH)$.
Then
\[
 \operatorname{TV}(P_u,P_v)
 \leq\frac12\int_u^v\|\partial_\theta p_\theta\|_1\,d\theta
 \leq\frac{v-u}{2}\sup_{\theta\in[u,v]}
               \|\partial_\theta p_\theta\|_1.
\]
The same conclusion holds when the family is continuously differentiable as an $L^1(\mu)$-valued map.
\end{lemma}
\begin{proof}
The fundamental theorem of calculus gives $p_v(H)-p_u(H)=\int_u^v\partial_\theta p_\theta(H)\,d\theta$ for almost every $H$.
Take absolute values, integrate over $H$, and interchange the nonnegative integrals to obtain the first inequality.
Bounding the integrand by its supremum proves the second.
For an $L^1$-valued continuously differentiable map, apply the fundamental theorem of calculus directly in $L^1$ and use the triangle inequality for its integral.
\end{proof}

For a positive differentiable density, the Fisher information is $I_\theta=\int (\partial_\theta p_\theta)^2/p_\theta\,d\mu$.
Cauchy--Schwarz gives $\|\partial_\theta p_\theta\|_1\leq\sqrt{I_\theta}$, recovering the Fisher-information path bound used in~\cite[App.~G]{Zhou2026AnsatzFreeHamiltonian} and related to the statistical-distance geometry of~\cite{Wootters1981StatisticalDistance}.
Our proof applies the $L^1$ path bound on good records and bounds the complementary part of total variation by its endpoint probabilities, as formalized in \Cref{lo:tv-split}.

\section{Upper bounds}
\label{sec:proof-upper}

We reduce higher-order moment estimation to lower-order estimation on a smaller, classically specified state space.
All bounds count the pilot copies and unsuccessful filtering attempts.
Write $p_j(\rho)=\operatorname{Tr}(\rho^j)$; constants may depend on the fixed moment order, accuracy, and failure probability, but not on the dimension or input state.

\begin{theorem}[Single-copy upper bound for state moments]
\label[theorem]{up:main}
For every integer $t\geq2$ and $\epsilon,\delta\in(0,1)$, an adaptive single-copy protocol estimates $p_t(\rho)$ to additive error $\epsilon$ with probability at least $1-\delta$, for every $d$-dimensional state $\rho$, using a worst-case budget of $\mathcal{O}_{t,\epsilon,\delta}(d^{h/(h+1)})$ copies, where $h=\lceil\log_2t\rceil$.
Only classical information is retained between copies.
\end{theorem}

\subsection{Purity estimation}
\label{up:purity-section}

The base case uses collisions between outcomes measured in a common random basis, a standard purity-estimation approach~\cite{VanEnk2012RandomMeasurements,Gong2026PurityInnerProduct}.

\begin{lemma}[Purity subroutine]
\label[lemma]{up:purity}
For every $\epsilon,\delta\in(0,1)$, purity on $D$-dimensional states can be estimated to error $\epsilon$ with failure probability at most $\delta$, using $\mathcal{O}_{\epsilon,\delta}(\sqrt D)$ single-copy measurements.
\end{lemma}

\begin{proof}
Assume $D\geq2$, as $D=1$ is immediate.
Measure $n=\max\{2,\lceil\sqrt D\rceil\}$ copies in a common Haar-random basis, with outcomes $Y_1,\ldots,Y_n$ and conditional probabilities $q_1,\ldots,q_D$.
Set
\[
    U=\binom n2^{-1}\sum_{a<b}\mathbf1\{Y_a=Y_b\},
    \qquad Z=(D+1)U-1,
    \qquad Q_j=\sum_iq_i^j.
\]
\Cref{pre:haar} gives
\[
    \mathbb EQ_2=\frac{1+p_2(\rho)}{D+1},
    \qquad
    \mathbb EQ_3=\frac{1+3p_2(\rho)+2p_3(\rho)}{(D+1)(D+2)}.
\]
Since $\mathbb E[U\mid q]=Q_2$, the estimator satisfies $\mathbb EZ=p_2(\rho)$.
Conditional on the basis, disjoint collision pairs are independent; an individual pair has variance at most $Q_2$, and overlapping pairs have covariance at most $Q_3$.
Counting these pairs gives
\[
    \operatorname{Var}(U\mid q)
    \leq\frac{Q_2}{\binom n2}
       +\frac{6\binom n3Q_3}{\binom n2^2}
    \leq\frac{4Q_2}{n^2}+\frac{4Q_3}{n}.
\]
The random basis contributes an additional variance term:
$\operatorname{Var}(U)=\mathbb E\operatorname{Var}(U\mid q)+\operatorname{Var}(Q_2)$.
Since $Q_2^2\leq Q_3$ by Cauchy--Schwarz and $\sum_iq_i=1$, we have $\operatorname{Var}(Q_2)\leq\mathbb EQ_3$.
Using $p_2,p_3\leq1$ in the Haar formulas therefore yields
\[
    \operatorname{Var}(Z)
    \leq(D+1)^2\left(
        \frac{8}{(D+1)n^2}
        +\frac{24}{(D+1)(D+2)n}
        +\frac6{(D+1)(D+2)}\right)
    \leq64.
\]
Average $\lceil64/(\epsilon^2\delta)\rceil$ independent repetitions, using fresh bases and copies.
Chebyshev's inequality proves the claim; clipping the answer to $[0,1]$ cannot increase its error.
\end{proof}

\subsection{State filtering and recursive moment estimation}
\label{up:filter-section}
\label{up:simulation-section}
\label{up:doubling-section}
\label{up:third-fourth}

We first learn a low-rank filter from pilot measurements, then explain how its conditional output encodes higher moments.
Let $\nu_{\mathrm H}$ be the normalized Haar probability measure on unit vectors in $\mathbb C^d$, and write $P_v=|v\rangle\langle v|$.
The covariant POVM $dP_v\,\nu_{\mathrm H}(dv)$~\cite{Harrow2013SymmetricSubspace} satisfies $\mathbb E_{\mathrm{cov},\rho}P_v=(I+\rho)/(d+1)$ by \Cref{pre:haar}.
Given $x\in[0,1]$, use this measurement on a fresh copy with probability $\gamma_x=x(d+1)/(d+x)$, and otherwise sample a Haar vector classically.
Each trial consumes at most one copy and satisfies
\begin{equation}
\label{up:pilot-mean}
    \mathbb EP_v=(1-\gamma_x)\frac Id
        +\gamma_x\frac{I+\rho}{d+1}
        =\frac{I+x\rho}{d+x}.
\end{equation}
After $n$ independent trials, the protocol knows $S_x=\sum_{i=1}^nP_{v_i}$, whose rank is at most $n$.
The matrices involving $\rho$ below are used only in the analysis.

\begin{lemma}[Concentration of the pilot matrix]
\label[lemma]{up:pilot-concentration}
Let $S_x$ be constructed from the independent pilot trials above, and define $B_x=(d+x)\sqrt\rho S_x\sqrt\rho/n$ and $H_x=\rho+x\rho^2$.
We use $\mathbb E_{\nu_{\mathrm H}}$ for expectation over a Haar-random unit vector, and $\mathbb E_{\mathrm{cov},\rho}$ for expectation over the covariant-measurement outcome, whose probability law is $d\,\langle v|\rho|v\rangle\,\nu_{\mathrm H}(dv)$.
Unsubscripted $\mathbb E$ and $\Pr$ refer to the mixture used by the pilot protocol, jointly over the independent trials when several vectors are involved.
The input state $\rho$ is fixed in all these expectations.

Uniformly over $\rho$, $x\in[0,1]$, and integers $n\geq1$,
\begin{equation}
\label{up:pilot-second-moments}
    \mathbb EB_x=H_x,
    \qquad
    \mathbb E\|B_x-H_x\|_{\mathrm{HS}}^2\leq\frac{24}{n},
    \qquad
    \mathbb E|\operatorname{Tr}(B_x-H_x)|^2\leq\frac{24}{n}.
\end{equation}
For $L>2n/d$,
\begin{equation}
\label{up:pilot-operator-tail}
    \Pr[\|S_x\|_{\mathrm{op}}\geq L]
        \leq d\left(\frac{2en}{Ld}\right)^L.
\end{equation}
\end{lemma}

\begin{proof}
The mean identity follows directly from \eqref{up:pilot-mean}:
\[
    \mathbb EB_x
    =\frac{d+x}{n}\sqrt\rho\,(\mathbb ES_x)\sqrt\rho
    =\sqrt\rho(I+x\rho)\sqrt\rho
    =H_x.
\]
We next bound the fluctuations of an individual trial and then average over the $n$ independent trials.

For one trial, put $Y=\sqrt\rho P_v\sqrt\rho$ and $z=\langle v|\rho|v\rangle$.
The matrix $Y$ is positive semidefinite and has rank at most one, with trace $z$.
Consequently, $\operatorname{Tr}Y=z$ and $\|Y\|_{\mathrm{HS}}^2=z^2$, so the same scalar second moment controls both quantities.

Under Haar sampling, \Cref{pre:haar} gives
\[
    \mathbb E_{\nu_{\mathrm H}} z^2
    =\frac{1+p_2(\rho)}{d(d+1)}
    \leq\frac2{d^2}.
\]
For the covariant measurement, the Born rule weights each Haar direction $v$ by $d\,z(v)$.
Thus the second moment under this outcome distribution is
\[
    \mathbb E_{\mathrm{cov},\rho}z^2
    =\int z(v)^2\,[d\,z(v)]\,\nu_{\mathrm H}(dv)
    =d\,\mathbb E_{\nu_{\mathrm H}} z^3.
\]
Applying the third-moment identity in \Cref{pre:haar} therefore yields
\[
    \mathbb E_{\mathrm{cov},\rho}z^2
    =\frac{1+3p_2(\rho)+2p_3(\rho)}{(d+1)(d+2)}
    \leq\frac6{d^2},
\]
where we used $p_2(\rho),p_3(\rho)\leq1$.
The pilot protocol mixes these two distributions, so its second moment is their weighted average:
\[
    \mathbb Ez^2
    =(1-\gamma_x)\mathbb E_{\nu_{\mathrm H}} z^2
       +\gamma_x\mathbb E_{\mathrm{cov},\rho}z^2
    \leq\frac6{d^2}.
\]

For the independent pilot vectors, write $Y_i=\sqrt\rho P_{v_i}\sqrt\rho$.
Since $\mathbb EB_x=H_x$, we have
\[
    B_x-H_x
    =\frac{d+x}{n}\sum_{i=1}^n(Y_i-\mathbb EY_i).
\]
The centered matrices from different trials are independent.
For $i\ne j$, this gives
$\mathbb E\operatorname{Tr}[(Y_i-\mathbb EY_i)(Y_j-\mathbb EY_j)]=0$.
Hence only the diagonal terms remain when we expand the squared Hilbert--Schmidt norm:
\[
\begin{aligned}
    \mathbb E\|B_x-H_x\|_{\mathrm{HS}}^2
    &=\frac{(d+x)^2}{n^2}
        \sum_{i=1}^n
        \mathbb E\|Y_i-\mathbb EY_i\|_{\mathrm{HS}}^2\\
    &=\frac{(d+x)^2}{n}
        \left(\mathbb E\|Y\|_{\mathrm{HS}}^2
                   -\|\mathbb EY\|_{\mathrm{HS}}^2\right)\\
    &\leq\frac{(d+x)^2}{n}\mathbb Ez^2
     \leq\frac{24}{n}.
\end{aligned}
\]
The last inequality uses $d+x\leq2d$ and $\mathbb Ez^2\leq6/d^2$.

For the trace error, $\operatorname{Tr}Y_i=z_i$ gives
$\operatorname{Tr}(B_x-H_x)=(d+x)n^{-1}\sum_{i=1}^n(z_i-\mathbb Ez_i)$.
Independence again eliminates the cross terms, and therefore
\[
    \mathbb E|\operatorname{Tr}(B_x-H_x)|^2
    =\frac{(d+x)^2}{n}\mathbb E|z-\mathbb Ez|^2
    \leq\frac{(d+x)^2}{n}\mathbb Ez^2
    \leq\frac{24}{n}.
\]

Finally, every pilot projector satisfies $0\preceq P_{v_i}\preceq I$.
Since $\rho\preceq I$, its mean obeys
\[
    \mathbb EP_{v_i}
    =\frac{I+x\rho}{d+x}
    \preceq\frac{1+x}{d+x}I
    \preceq\frac2d I.
\]
The trials are independent, so the same bound holds for the conditional mean given all preceding trials.
Applying \Cref{pre:matrix-chernoff} to their sum $S_x$, with $N=n$ and $\mu=2/d$, proves \eqref{up:pilot-operator-tail}.
\end{proof}

If $\|S_x\|_{\mathrm{op}}>L$, abort this pilot.
Otherwise $K_x=\sqrt{S_x/L}$ satisfies $K_x^\dagger K_x=S_x/L\preceq I$ and is therefore a valid Kraus operator.
Applied to a fresh copy of $\rho$, it succeeds with probability $s_x=\operatorname{Tr}(K_x\rho K_x^\dagger)$ and, upon success, produces the normalized state $K_x\rho K_x^\dagger/s_x$.
Thus
\begin{equation}
\label{up:filter-state}
    s_x=\frac{\operatorname{Tr}(S_x\rho)}L,
    \qquad
    \sigma_x=\frac{\sqrt{S_x}\rho\sqrt{S_x}}
                       {\operatorname{Tr}(S_x\rho)}.
\end{equation}
The output is supported on the known space $\operatorname{range}(S_x)$, whose dimension is at most $n$.
We show below that its moments approximate those of the normalized matrix $\rho+x\rho^2$.
This is a spectral approximation; the corresponding density matrices need not be close in the same basis.

\begin{lemma}[Moments and success probability of the filter]
\label[lemma]{up:normalized-filter}
For every integer $m\geq2$ and $\eta\in(0,1/2)$, there is an event of probability at least $1-48/(n\eta^2)$ over the pilot samples with the following property.
On this event, whenever $\|S_x\|_{\mathrm{op}}\leq L$, the filter is valid and satisfies
\begin{equation}
\label{up:filtered-moment}
    \left|p_m(\sigma_x)
      -\frac{\operatorname{Tr}[(\rho+x\rho^2)^m]}
                    {(1+xp_2(\rho))^m}\right|
       \leq4m\eta,
    \qquad
    s_x\geq\frac{n}{4Ld}.
\end{equation}
\end{lemma}

\begin{proof}
We first obtain simultaneous control of the matrix error and the trace error.
By Markov's inequality and \eqref{up:pilot-second-moments},
\[
    \Pr\!\left(\|B_x-H_x\|_{\mathrm{HS}}>\eta\right)
    \leq\frac{\mathbb E\|B_x-H_x\|_{\mathrm{HS}}^2}{\eta^2}
    \leq\frac{24}{n\eta^2}.
\]
The same argument gives
$\Pr(|\operatorname{Tr}(B_x-H_x)|>\eta)\leq24/(n\eta^2)$.
A union bound therefore shows that, with probability at least $1-48/(n\eta^2)$, both
\[
    \|B_x-H_x\|_{\mathrm{HS}}\leq\eta,
    \qquad
    |\operatorname{Tr}B_x-\operatorname{Tr}H_x|\leq\eta
\]
hold simultaneously.
For the rest of the proof, fix pilot samples in this event and suppose that the filter is valid.

We next check that normalization preserves the approximation.
Since $H_x=\rho+x\rho^2$, its trace satisfies
$\operatorname{Tr}H_x=1+xp_2(\rho)\in[1,2]$.
The trace-error bound and $\eta<1/2$ consequently give
$1/2\leq\operatorname{Tr}B_x\leq5/2$.
In particular, the normalization denominator stays bounded away from zero.
Decomposing the normalized difference gives
\[
    \frac{B_x}{\operatorname{Tr}B_x}
      -\frac{H_x}{\operatorname{Tr}H_x}
    =
    \frac{B_x-H_x}{\operatorname{Tr}B_x}
    +\frac{H_x}{\operatorname{Tr}H_x}
       \frac{\operatorname{Tr}H_x-\operatorname{Tr}B_x}
            {\operatorname{Tr}B_x}.
\]
The matrix $H_x/\operatorname{Tr}H_x$ is a state and hence has Hilbert--Schmidt norm at most one.
Taking norms in this identity therefore yields
\[
    \left\|
      \frac{B_x}{\operatorname{Tr}B_x}
       -\frac{H_x}{\operatorname{Tr}H_x}
    \right\|_{\mathrm{HS}}
    \leq\frac{\eta}{1/2}+\frac{\eta}{1/2}
    =4\eta.
\]

To connect this approximation to the actual filtered state, set $V=\sqrt{S_x}\sqrt\rho$.
Then
\[
    VV^\dagger=\sqrt{S_x}\rho\sqrt{S_x},
    \qquad
    V^\dagger V=\sqrt\rho S_x\sqrt\rho.
\]
These two positive matrices have the same nonzero eigenvalues, namely the squared singular values of $V$, and the same trace $\operatorname{Tr}(S_x\rho)$.
Moreover, the scalar factor $(d+x)/n$ in $B_x$ cancels upon normalization, so
\[
    \frac{B_x}{\operatorname{Tr}B_x}
    =\frac{\sqrt\rho S_x\sqrt\rho}
           {\operatorname{Tr}(S_x\rho)}.
\]
Thus $\sigma_x$ and $B_x/\operatorname{Tr}B_x$ have the same spectrum, although they need not be the same matrix.
Since moments depend only on eigenvalues, we obtain
$p_m(\sigma_x)=p_m(B_x/\operatorname{Tr}B_x)$.

Both normalized matrices are states, so the moment-continuity bound in \Cref{pre:moment-lipschitz} now gives
\[
\begin{aligned}
    \left|
      p_m(\sigma_x)
       -p_m\!\left(\frac{H_x}{\operatorname{Tr}H_x}\right)
    \right|
    &\leq m
       \left\|
         \frac{B_x}{\operatorname{Tr}B_x}
          -\frac{H_x}{\operatorname{Tr}H_x}
       \right\|_{\mathrm{HS}}\\
    &\leq4m\eta.
\end{aligned}
\]
Substituting $H_x=\rho+x\rho^2$ and $\operatorname{Tr}H_x=1+xp_2(\rho)$ proves the moment bound in \eqref{up:filtered-moment}.

Finally, taking the trace in the definition of $B_x$ gives
$\operatorname{Tr}B_x=(d+x)\operatorname{Tr}(S_x\rho)/n$.
The filtering success probability can therefore be written as
\[
    s_x=\frac{n\,\operatorname{Tr}B_x}{L(d+x)}.
\]
Using $1/2\leq\operatorname{Tr}B_x\leq5/2$ and $d\leq d+x\leq2d$, we obtain
\[
    \frac{n}{4Ld}\leq s_x\leq\frac{5n}{2Ld}.
\]
This proves the claimed lower bound and shows that, for fixed $L$, the success probability is $\Theta(n/d)$.
\end{proof}

\begin{lemma}[Single-copy simulation after filtering]
\label[lemma]{up:filter-simulation}
Let $K$ be known, $K^\dagger K\preceq I$, and $s=\operatorname{Tr}(K\rho K^\dagger)\geq s_0>0$.
For every $\beta\in(0,1)$, any adaptive single-copy protocol using at most $N\geq1$ copies of $\sigma=K\rho K^\dagger/s$ can be simulated on original copies, with an additional failure probability at most $\beta$, using at most
\[
    M=\left\lceil\frac8{s_0}\left(N+\log\frac1\beta\right)\right\rceil
\]
copies and retaining only classical information between them.
\end{lemma}

\begin{proof}
For each POVM $\{E_y\}$ requested by the inner protocol, measure a fresh original copy with effects
$F_y=K^\dagger E_yK$ and $F_{\mathrm{fail}}=I-K^\dagger K$.
Extend a POVM on the known output support arbitrarily to its complement if necessary.
These effects are positive and sum to $I$, so each trial is a single POVM on one original copy.
The success probability is $s$, independently of the inner history, and
$\Pr[y\mid\mathrm{success}]=\operatorname{Tr}(E_y\sigma)$.
After failure retry without changing that history; after success pass $y$ to the inner protocol.
Consequently the simulated successful transcript has exactly the desired law.
The success indicators are independent Bernoulli variables of parameter $s$; padding after termination if necessary defines $X\sim\operatorname{Bin}(M,s)$ with mean $\mu\geq8(N+\log(1/\beta))$.
\Cref{pre:binomial} gives $\Pr[X<N]\leq\Pr[X<\mu/2]\leq e^{-\mu/8}\leq\beta$.
Coupling to the unlimited successful transcript proves the claim.
\end{proof}

The filter therefore realizes lower-dimensional single-copy protocols at a copy overhead $\mathcal{O}(Ld/n)$.
The information it encodes is the polynomial
\begin{equation}
\label{up:encoding-polynomial}
    F_m(x)=\operatorname{Tr}[(\rho+x\rho^2)^m]
          =\sum_{j=0}^m\binom mjx^jp_{m+j}(\rho).
\end{equation}
Thus its $m$th moments at finitely many filter parameters recover every order from $m$ through $2m$.

We now prove a more precise version of \Cref{lem:moment-doubling} from the technical overview, with simultaneous estimation guarantees and explicit control of the accuracy and failure probability.

\begin{lemma}[Moment-order doubling]
\label[lemma]{up:doubling}
Fix an integer $m\geq2$ and $a\in[1/2,1)$.
Suppose that, for every dimension $d\geq1$ and every fixed accuracy and failure probability, $p_m$ on $d$-dimensional states has an adaptive single-copy estimator using $\mathcal{O}(d^a)$ copies in the worst case, with a constant independent of $D$ and the input state.
Then, for every $\epsilon,\delta\in(0,1)$, all of $p_m,\ldots,p_{2m}$ can be estimated simultaneously to error $\epsilon$, with joint success probability at least $1-\delta$, using $\mathcal{O}_{m,a,\epsilon,\delta}(d^{1/(2-a)})$ copies.
\end{lemma}

\begin{proof}
Put $b=1/(2-a)$, $n=\lceil d^b\rceil$, and $L=\lceil2/(1-b)\rceil$.
The threshold $L$ depends only on $a$, not on $d$, and its choice makes the pilot operator-norm failure probability vanish.
For nodes $x_i=i/m$, let $\ell_i$ be the Lagrange polynomials and define
\[
    A_m=\max_{0\leq j\leq m}\sum_{i=0}^m
           \left|\frac{[x^j]\ell_i(x)}{\binom mj}\right|,
    \qquad \xi=\frac\epsilon{A_m}.
\]
Here $A_m\geq1$ since $x_0=0$.
By \Cref{pre:interpolation}, node errors at most $\xi$ suffice for all desired moment errors to be at most $\epsilon$.
Take the fixed internal tolerances
\begin{equation}
\label{up:doubling-accuracies}
    \tau=\frac{\xi}{3\,2^m},\qquad
    \zeta=\frac{\xi}{3m2^{m-1}},\qquad
    \eta=\frac{\xi}{12m2^m},\qquad
    \beta=\frac\delta{4(m+1)}.
\end{equation}
First estimate $p_2(\rho)$ to error $\zeta$ and failure probability $\delta/4$, and clip the result $\widehat p_2$ to $[0,1]$.
At each node independently construct its pilot filter.
By Lemmas~\ref{up:pilot-concentration} and~\ref{up:normalized-filter}, the probability that either validity or the conclusions of~\eqref{up:filtered-moment} fail is at most
\[
    d\left(\frac{2en}{Ld}\right)^L+\frac{48}{n\eta^2}
    \leq\left(\frac{4e}{L}\right)^Ld^{1-L(1-b)}
          +\frac{48}{d^b\eta^2}=o(1).
\]
Here $L(1-b)\geq2$; in particular the failure probability is at most $\beta$ for $d$ above a fixed threshold.

On a good pilot, estimate $p_m(\sigma_{x_i})$ to error $\tau$ and failure probability $\beta$, using at most $C_{m,\tau,\beta}n^a$ successful outputs.
Its support is known and has dimension at most $n$.
\Cref{up:filter-simulation} realizes this subroutine using $\mathcal{O}((d/n)(n^a+1))$ original copies and additional failure probability $\beta$.
Clip the inner estimate $\widehat b_i$ to $[0,1]$, and form
$\widehat F_i=(1+x_i\widehat p_2)^m\widehat b_i$.
On the good events, the filtered-moment estimate differs from $F_m(x_i)/(1+x_ip_2)^m$ by at most $\tau+4m\eta$.
The normalization factor is at most $2^m$, and its change when replacing $p_2$ by $\widehat p_2$ is at most $m2^{m-1}\zeta$ by the mean-value theorem.
Thus
\[
    |\widehat F_i-F_m(x_i)|
       \leq2^m(\tau+4m\eta)+m2^{m-1}\zeta=\xi.
\]
Interpolation of~\eqref{up:encoding-polynomial} now gives the desired estimates.
The purity failure contributes $\delta/4$, and the pilot, inner-estimation, and simulation failures at all nodes contribute at most $3(m+1)\beta=3\delta/4$.
An invalid pilot or an exhausted simulation budget is declared a failure and produces an arbitrary estimate.
The total budget accounts for purity estimation, pilot trials, successful filtered samples, and the additional attempts needed for a fixed simulation failure probability:
\begin{equation}
\label{up:doubling-cost}
    \mathcal{O}_{m,a,\epsilon,\delta}\left(\sqrt d+n+dn^{a-1}+\frac dn\right)
       =\mathcal{O}_{m,a,\epsilon,\delta}(d^b),
\end{equation}
because $b=1+b(a-1)$ and $b\geq2/3$.
The finitely many smaller dimensions can be handled by single-copy tomography to Hilbert--Schmidt accuracy $\epsilon/(2m)$, followed by moment evaluation and \Cref{pre:moment-lipschitz}; increasing the dimension-independent constant covers these cases.
\end{proof}

\begin{proof}[Proof of \Cref{up:main}]
Starting with purity, apply \Cref{up:doubling} for $m=2,4,8,\ldots$.
The exponents obey
\[
    a_1=\frac12,\qquad a_{h+1}=\frac1{2-a_h},
    \qquad a_h=\frac h{h+1}.
\]
Level $h$ supplies every order $2^{h-1}<t\leq2^h$, so choose $h=\lceil\log_2t\rceil$.
The number of recursive levels depends only on $t$, and each invocation permits arbitrary fixed internal error and failure tolerances.
Their allocation therefore introduces no dimension dependence.
\Cref{up:filter-simulation} counts all attempts at every level and ensures that no quantum information passes between copies.
\end{proof}

For example, $F_2(x)=p_2+2xp_3+x^2p_4$ gives
\[
    p_3=2F_2(1/2)-\tfrac12F_2(1)-\tfrac32p_2,
    \qquad
    p_4=2F_2(1)-4F_2(1/2)+2p_2.
\]
Taking $n=\lceil d^{2/3}\rceil$ and $L=6$, purity estimation on each filtered state costs $\mathcal{O}(\sqrt n)$ successful measurements, hence $\mathcal{O}(d/\sqrt n)$ original copies.
Both this cost and the pilot cost are $\mathcal{O}(d^{2/3})$, while the normalization costs an additional $\mathcal{O}(\sqrt d)$ copies.
The two fixed interpolation formulas show directly why orders three and four have the same upper bound.

\subsection{Cyclic overlaps, weighted moments, partial transposition, and fixed-rank spectra}
\label{up:overlap-section}
\label{up:weighted-section}
\label{up:partial-transpose-section}
\label{up:spectrum-section}

We prove the four upper bounds in \Cref{cor:moment-consequences} by reductions to ordinary moments.
The first three recover the target quantities from finitely many moments of suitably constructed states.
The last uses uniform continuity of the inverse moment map at fixed rank.

\begin{theorem}[Upper bound for cyclic overlaps]
\label[theorem]{up:cyclic-overlap}
Fix an integer $t\geq2$ and $\epsilon,\delta\in(0,1)$, and let $h=\lceil\log_2t\rceil$.
For arbitrary unknown states $\rho_0,\ldots,\rho_{t-1}$ on $\mathbb C^d$, an adaptive single-copy protocol estimates the cyclic overlap
$z=\operatorname{Tr}(\rho_0\cdots\rho_{t-1})$
with complex-modulus error at most $\epsilon$ and failure probability at most $\delta$, using
$\mathcal{O}_{t,\epsilon,\delta}(d^{h/(h+1)})$
copies in total across all input states.
Each round measures one fresh copy from a classically selected input source, and only classical information is retained between rounds.
\end{theorem}

\begin{proof}
For $t=2$, the overlap is real and satisfies
\[
    \operatorname{Tr}(\rho_0\rho_1)
    =2p_2\!\left(\frac{\rho_0+\rho_1}{2}\right)
      -\frac{p_2(\rho_0)+p_2(\rho_1)}2.
\]
A copy of the mixture is obtained by choosing either input source with equal probability.
Estimating the three purities to error $\epsilon/3$ and failure probability $\delta/3$ therefore proves the claim using \Cref{up:purity}.

Now suppose $t\geq3$.
The target can be complex, and moments of ordinary mixtures generally combine different product orderings.
To separate the desired ordering and its real and imaginary parts, introduce known auxiliary states
\[
    A_j=\ket{a_j}\bra{a_j},
    \qquad
    \ket{a_j}=\frac{\ket{j}+e^{i\phi_j}\ket{j+1}}{\sqrt2},
    \qquad 0\leq j<t,
\]
where the auxiliary space has dimension $t$ and basis labels are interpreted modulo $t$.
These states are used to describe the reduction; below we absorb them into single-copy POVMs on the original inputs.

For each nonempty subset $S\subseteq\{0,\ldots,t-1\}$, consider the physical mixture
\[
    \Omega_S=\frac1{|S|}\sum_{j\in S}\rho_j\otimes A_j.
\]
Expanding its $t$th moment and applying inclusion--exclusion gives
\begin{equation}
\label{up:overlap-polarization}
    \sum_{\varnothing\ne S\subseteq\{0,\ldots,t-1\}}
       (-1)^{t-|S|}|S|^t p_t(\Omega_S)
    =t\,2^{1-t}\operatorname{Re}\!\left[
        e^{-i\sum_j\phi_j}z
      \right].
\end{equation}
To verify this identity, consider a product in the expansion whose set of input labels is $J$.
Its coefficient in the subset sum is
$\sum_{S\supseteq J}(-1)^{t-|S|}$,
which vanishes unless $J=\{0,\ldots,t-1\}$.
Since every product has length $t$, the surviving products contain each input label exactly once.

For an ordering $(i_1,\ldots,i_t)$, its auxiliary factor is
\[
    \operatorname{Tr}(A_{i_1}\cdots A_{i_t})
    =\langle a_{i_1}|a_{i_2}\rangle
       \cdots
       \langle a_{i_t}|a_{i_1}\rangle.
\]
Distinct auxiliary states have zero overlap unless their labels are adjacent on the cycle.
Consequently, only the forward and reverse cyclic orderings survive, each with $t$ cyclic rotations.
Their auxiliary factors are $2^{-t}e^{-i\sum_j\phi_j}$ and its complex conjugate, respectively.
Their state factors are $z$ and $\overline z$, because the input states are Hermitian.
Adding these contributions proves \eqref{up:overlap-polarization}.

Taking all phases zero recovers $\operatorname{Re}z$ after multiplying the left-hand side by $2^{t-1}/t$.
Taking $\phi_0=\pi/2$ and all other phases zero similarly recovers $\operatorname{Im}z$, since $\operatorname{Re}(-iz)=\operatorname{Im}z$.
For explicit error control, set
\[
    C_t=\frac{2^{t-1}}t
        \sum_{\varnothing\ne S\subseteq\{0,\ldots,t-1\}}|S|^t.
\]
Estimate every moment in the two phase settings to error $\epsilon/(2C_t)$ and failure probability $\delta/[2(2^t-1)]$.
On the joint success event, each recovered component has error at most $\epsilon/2$, so the resulting estimate has complex-modulus error at most $\epsilon$.
The union bound gives total failure probability at most $\delta$.
There are $2(2^t-1)$ moment estimates, each on a $td$-dimensional state.
Since $t$ is fixed, \Cref{up:main} gives total complexity
$\mathcal{O}_{t,\epsilon,\delta}(d^{h/(h+1)})$.

Finally, each requested measurement on $\Omega_S$ can be implemented directly on one original input copy.
Choose $j$ uniformly from $S$.
If the moment estimator requests a POVM $\{E_y\}$ on the enlarged space, measure $\rho_j$ with
\[
    F_{y\mid j}
       =(I\otimes\bra{a_j})E_y(I\otimes\ket{a_j}).
\]
These effects are positive and sum to $I$, and averaging over the chosen source gives
\[
    \frac1{|S|}\sum_{j\in S}
       \operatorname{Tr}(F_{y\mid j}\rho_j)
    =\operatorname{Tr}(E_y\Omega_S).
\]
Passing only the outcome $y$ to the inner estimator therefore reproduces its measurement statistics exactly.
This construction applies at every adaptive history, so no auxiliary state needs to be physically prepared and no quantum information is retained between copies.
Each simulated mixture sample consumes exactly one original input copy, establishing the claimed total copy budget.
\end{proof}

\begin{theorem}[Upper bound for weighted moments]
\label[theorem]{up:weighted}
For a fixed integer $t\geq2$, known Hermitian $O$ with $\|O\|_{\mathrm{op}}\leq1$, and $\epsilon,\delta\in(0,1)$, estimating $\operatorname{Tr}(O\rho^t)$ to error $\epsilon$ and failure probability $\delta$ uses $\mathcal{O}_{t,\epsilon,\delta}(d^{h/(h+1)})$ single-copy measurements, where $h=\lceil\log_2t\rceil$.
The bound is uniform over $O$ and $\rho$.
\end{theorem}

\begin{proof}
For $x\in[-1/2,1/2]$, the known filter $K_x=\sqrt{(I+xO)/2}$ succeeds with probability $q_x\in[1/4,3/4]$ and produces $\sigma_x=K_x\rho K_x/q_x$.
Equality of the nonzero spectra of $AA^\dagger$ and $A^\dagger A$ gives the degree-at-most-$t$ polynomial
\begin{equation}
\label{up:weighted-polynomial}
    F(x)=(2q_x)^tp_t(\sigma_x)
         =\operatorname{Tr}[(\rho+x\sqrt\rho O\sqrt\rho)^t].
\end{equation}
Cyclicity of the trace, without any commutation assumption on $O$, gives
\[
    [x]F(x)=\sum_{j=0}^{t-1}
       \operatorname{Tr}(\rho^j\sqrt\rho O\sqrt\rho\rho^{t-1-j})
       =t\operatorname{Tr}(O\rho^t).
\]
Choose $t+1$ distinct fixed nodes in the filter interval.
At each node estimate $q_x$ by Bernoulli trials to error $\zeta$, clipping to $[1/4,3/4]$, and estimate $p_t(\sigma_x)$ to error $\tau$, clipping to $[0,1]$.
The latter uses \Cref{up:main} and \Cref{up:filter-simulation} with $s_0=1/4$, so filtering has constant copy overhead.
The estimate $\widehat F(x)=(2\widehat q_x)^t\widehat p_t(\sigma_x)$ satisfies
\[
    |\widehat F(x)-F(x)|
       \leq(3/2)^t\tau+2t(3/2)^{t-1}\zeta
\]
whenever its two component estimates are accurate.
\Cref{pre:interpolation} thus permits fixed $\tau,\zeta>0$, depending only on $t,\epsilon$, that recover $[x]F/t$ to error $\epsilon$.
Assign failure probability $\delta/[3(t+1)]$ to each node's probability estimate, moment estimate, and simulation.
The union bound and the constant number of nodes prove the claim.
\end{proof}

\begin{theorem}[Upper bound for partial-transpose moments]
\label[theorem]{up:partial-transpose}
Fix integers $t\geq2$, $q\geq1$, and $\epsilon,\delta\in(0,1)$.
If $\rho_{AB}$ has total dimension $d$ and either subsystem has dimension $q$, then $\operatorname{Tr}[(\rho_{AB}^{\mathrm{T}_A})^t]$ can be estimated to error $\epsilon$ and failure probability $\delta$ using $\mathcal{O}_{t,q,\epsilon,\delta}(d^{h/(h+1)})$ single-copy measurements, where $h=\lceil\log_2t\rceil$.
\end{theorem}

\begin{proof}
Suppose first that $\dim A=q$.
Define the maps
\[
    \Phi(X)=\frac{X^\mathrm{T}+\operatorname{Tr}(X)I_A}{q+1},
    \qquad
    \mathcal D(X)=\frac{\operatorname{Tr}(X)}q I_A.
\]
With the convention $J(\Psi)=\sum_{a,b}|a\rangle\langle b|\otimes\Psi(|a\rangle\langle b|)$, their unnormalized Choi matrices are $(F+I)/(q+1)$ and $I/q$, respectively, where $F$ is the swap.
Both matrices are positive semidefinite and have output partial trace $I$, so both maps are quantum channels.
Moreover, $X^\mathrm{T}=(q+1)\Phi(X)-q\mathcal D(X)$.

Let $\omega_1=(\Phi\otimes\operatorname{id}_B)(\rho_{AB})$ and $\omega_0=(\mathcal D\otimes\operatorname{id}_B)(\rho_{AB})$.
Then $\rho_{AB}^{\mathrm{T}_A}=(q+1)\omega_1-q\omega_0$.
Expanding its $t$th power while preserving the order of the factors gives
\[
    \operatorname{Tr}[(\rho_{AB}^{\mathrm{T}_A})^t]
    =
    \sum_{\boldsymbol{s}\in\{0,1\}^t}
       (q+1)^{|\boldsymbol{s}|}(-q)^{t-|\boldsymbol{s}|}
       \operatorname{Tr}(\omega_{s_1}\cdots\omega_{s_t}),
\]
where $|\boldsymbol{s}|=\sum_{j=1}^t s_j$.
Thus the target is a linear combination of $2^t$ ordered overlaps, each covered by \Cref{up:cyclic-overlap}; repeated input states are allowed.

The sum of the absolute coefficients is $(2q+1)^t$.
Estimate each overlap to complex-modulus error $\epsilon/(2q+1)^t$ and failure probability $\delta/2^t$, and output the real part of their weighted sum.
The triangle inequality and union bound give additive error at most $\epsilon$ with probability at least $1-\delta$.
Since $t$ and $q$ are fixed, \Cref{up:cyclic-overlap} gives a total copy budget of
$\mathcal{O}_{t,q,\epsilon,\delta}(d^{h/(h+1)})$, where $h=\lceil\log_2t\rceil$.

Each requested copy of $\omega_0$ or $\omega_1$ uses one fresh copy of $\rho_{AB}$.
The corresponding channel followed by a measurement can be combined into a POVM on that original copy by applying the channel adjoint to the measurement effects.
This remains valid for adaptive measurement choices, so the reduction preserves single-copy access.

If instead $\dim B=q$, apply the same construction to the transpose on $B$.
The identity $(\rho_{AB}^{\mathrm{T}_A})^\mathrm{T}=\rho_{AB}^{\mathrm{T}_B}$ implies that the two partial transposes have the same spectrum and hence the same moments.
\end{proof}

\begin{theorem}[Upper bound for fixed-rank spectrum estimation]
\label[theorem]{up:spectrum}
For a fixed integer $r\geq2$ and $\epsilon,\delta\in(0,1)$, the sorted spectrum of any $d$-dimensional state of rank at most $r$ can be estimated to $\ell_1$ error $\epsilon$, with failure probability at most $\delta$, using $\mathcal{O}_{r,\epsilon,\delta}(d^{h/(h+1)})$ single-copy measurements, where $h=\lceil\log_2r\rceil$.
No eigenvalue-gap assumption is required.
\end{theorem}

\begin{proof}
Let $s=\min\{r,d\}$ and represent the unknown spectrum by a sorted probability vector $\lambda$ of length $s$, padding with zeros when the rank is smaller than $s$.
If $s=1$, the spectrum is known exactly, so assume $s\geq2$.

We first choose the accuracy required for moment estimation.
For each integer $2\leq k\leq r$, \Cref{pre:spectrum-continuity}, after shrinking its threshold if necessary, supplies $\kappa_k>0$ such that any two sorted probability vectors of length $k$ satisfy
\[
    \max_{2\leq j\leq k}|p_j(\lambda)-p_j(\mu)|
    \leq\kappa_k
    \quad\Longrightarrow\quad
    \|\lambda-\mu\|_1\leq\epsilon.
\]
Set $\kappa=\min_{2\leq k\leq r}\kappa_k>0$.
This choice depends only on $r,\epsilon$, not on $d$.
Using \Cref{up:main}, estimate each moment $p_j(\rho)$, $2\leq j\leq s$, to error $\kappa/4$ and failure probability $\delta/(s-1)$.
Write the resulting estimates as $\widehat p_j$.
By the union bound, all estimates are simultaneously accurate with probability at least $1-\delta$.

We reconstruct the spectrum by comparing these estimates with the moments of finitely many candidate spectra.
Choose a finite $\ell_1$-net $\mathcal N$ of radius $\kappa/(4s)$ in the compact set of sorted probability vectors of length $s$.
Thus every possible spectrum lies within this distance of some candidate in $\mathcal N$.
Output a candidate $\widehat\lambda$ minimizing
\[
    \max_{2\leq j\leq s}|p_j(\mu)-\widehat p_j|
    \qquad\text{over }\mu\in\mathcal N.
\]
This minimization involves only classical calculations.

To verify the reconstruction, condition on all moment estimates being accurate.
There exists $\lambda^*\in\mathcal N$ with
$\|\lambda^*-\lambda\|_1\leq\kappa/(4s)$.
Since $|x^j-y^j|\leq j|x-y|$ for $x,y\in[0,1]$, we have
\[
    |p_j(\lambda^*)-p_j(\lambda)|
    \leq j\|\lambda^*-\lambda\|_1
    \leq\kappa/4,
    \qquad 2\leq j\leq s.
\]
Together with the estimation error, this shows that $\lambda^*$ differs from every estimated moment by at most $\kappa/2$.
The minimizing candidate $\widehat\lambda$ has no larger maximum discrepancy.
Consequently,
\[
    |p_j(\widehat\lambda)-p_j(\lambda)|
    \leq |p_j(\widehat\lambda)-\widehat p_j|
       +|\widehat p_j-p_j(\lambda)|
    \leq\frac{3\kappa}{4}
    \leq\kappa_s,
    \qquad 2\leq j\leq s.
\]
The defining property of $\kappa_s$ therefore gives
$\|\widehat\lambda-\lambda\|_1\leq\epsilon$.
Appending $d-s$ zeros produces an estimate of the full spectrum without changing this error.

Finally, there are at most $r-1$ moment estimators.
Their accuracies and failure probabilities depend only on $r,\epsilon,\delta$, and their dimension exponents are at most $h/(h+1)$, where $h=\lceil\log_2r\rceil$.
Their total copy budget is therefore
$\mathcal{O}_{r,\epsilon,\delta}(d^{h/(h+1)})$.
The argument uses moment continuity on the entire sorted probability simplex, including repeated and zero eigenvalues, so no eigenvalue-gap assumption is needed.
\end{proof}

\subsection{Distributed estimation of higher-order moments}
\label{up:distributed-section}

We now allow Alice and Bob to perform arbitrary joint operations on their own copies of $\rho$ and $\sigma$, respectively, with classical communication but no quantum communication or preshared entanglement.
The preceding filtering construction still applies, while Bob's local joint measurements remove the dimension-dependent cost of estimating the filtered moments.

\begin{theorem}[Upper bound for distributed alternating moments]
\label[theorem]{up:distributed}
For every fixed integer $t\geq2$ and $\epsilon,\delta\in(0,1)$, Alice and Bob can estimate $\operatorname{Tr}[(\rho\sigma)^t]$ for arbitrary $d$-dimensional states $\rho,\sigma$ to additive error $\epsilon$ with probability at least $1-\delta$, using a worst-case total budget of $\mathcal{O}_{t,\epsilon,\delta}(\sqrt d)$ copies.
One-way classical communication from Alice to Bob suffices, and the budget includes unsuccessful filtering attempts.
\end{theorem}

\begin{proof}
We adapt the filtering analysis in \Cref{up:filter-section}, using Alice's state to construct the filter and Bob's state as its input.
Set $n=\lceil\sqrt d\rceil$ and $L=6$.
For each $x\in[0,1]$, Alice performs the pilot trials of \eqref{up:pilot-mean} and sends $S_x=\sum_{i=1}^nP_{v_i}$ to Bob.
The analysis now uses
\[
    B_x=\frac{d+x}{n}\sqrt{\sigma}S_x\sqrt{\sigma},
    \qquad H_x=\sigma+x\sqrt{\sigma}\rho\sqrt{\sigma}.
\]
The mean identity remains $\mathbb E B_x=H_x$, and $\operatorname{Tr}H_x=1+x\operatorname{Tr}(\rho\sigma)\in[1,2]$.
To extend the concentration proof to different input states, the only changed calculation is
\[
    \mathbb E_{\mathrm{cov},\rho}(v^\dagger\sigma v)^2
    =d\,\mathbb E_{\nu_{\mathrm H}}[(v^\dagger\rho v)(v^\dagger\sigma v)^2]
    \leq\frac6{d^2},
\]
which follows from H\"older's inequality and the Haar third-moment formula in \Cref{pre:haar}.
Thus the Hilbert--Schmidt and trace second-moment bounds of \Cref{up:pilot-concentration} remain $24/n$, and the operator-norm bound for $S_x$ is unchanged.

Fix an internal accuracy $\eta\in(0,1/2)$.
The same argument as in \Cref{up:normalized-filter} shows that, with arbitrarily high fixed probability for sufficiently large $d$, Bob's filter $K_x=\sqrt{S_x/L}$ is valid, succeeds with probability $s_x=\Theta(n/d)$, and produces
\[
    \omega_x=\frac{\sqrt{S_x}\sigma\sqrt{S_x}}
                       {\operatorname{Tr}(S_x\sigma)}.
\]
On this event, $b_x:=\operatorname{Tr}B_x=L(d+x)s_x/n$ lies in $[1/2,5/2]$.
Equality of the nonzero spectra identifies $p_t(\omega_x)$ with the moment of $B_x/b_x$.
Combining the normalized-moment estimate of \Cref{pre:moment-lipschitz} with $|b_x-\operatorname{Tr}H_x|\leq\eta$ and the bounded traces gives
\begin{equation}
\label{up:distributed-moment-identity}
    b_x^t p_t(\omega_x)
    =\operatorname{Tr}(B_x^t)
    =F_t(x)+\mathcal{O}_t(\eta),
    \qquad F_t(x):=\operatorname{Tr}(H_x^t).
\end{equation}

Bob now replaces the single-copy moment subroutine by a joint measurement of $(V_t+V_t^\dagger)/2$, where $V_t$ cyclically permutes $t$ copies.
This observable has norm at most one and expectation $p_t(\omega_x)$.
Consequently, fixed accuracy and failure probability require only a dimension-independent number of successful outputs.
The binomial collection argument used in \Cref{up:filter-simulation} supplies these within $\mathcal{O}_{t,\eta,\delta}(d/n)$ original copies; here Bob retains the outputs for the joint measurement.

Normalization costs the same order of copies.
Conditional on a fixed good pilot $S_x$, the empirical success frequency $\widehat s_x$ from $M$ separate filtering attempts gives an unbiased estimate $L(d+x)\widehat s_x/n$ of $b_x$, with
\[
    \operatorname{Var}\!\left(\left.\frac{L(d+x)}n\widehat s_x\,\right|S_x\right)
    \leq\frac{5L(d+x)}{2nM}.
\]
Chebyshev's inequality therefore estimates $b_x$ to fixed accuracy using $\mathcal{O}_{\eta,\delta}(d/n)$ attempts.
Clipping the estimates of $b_x$ and $p_t(\omega_x)$ to $[0,5/2]$ and $[0,1]$, respectively, and combining them as in \eqref{up:distributed-moment-identity}, estimates $F_t(x)$ to error $\mathcal{O}_t(\eta)$.

Finally, $F_t$ is a polynomial of degree at most $t$, with leading coefficient
\[
    [x^t]F_t(x)
    =\operatorname{Tr}[(\sqrt{\sigma}\rho\sqrt{\sigma})^t]
    =\operatorname{Tr}[(\rho\sigma)^t].
\]
Estimate $F_t$ at $t+1$ fixed nodes and apply \Cref{pre:interpolation}, as in the proof of \Cref{up:doubling}.
Choose $\eta$ sufficiently small depending on $t,\epsilon$, and allocate failure probability $\delta/[4(t+1)]$ to each node's pilot, collection, moment estimate, and normalization estimate.
The union bound gives the desired success probability.
Invalid pilots or exhausted budgets return an arbitrary estimate.
Since all internal tolerances are independent of $d$, the total worst-case cost is
\[
    \mathcal{O}_{t,\epsilon,\delta}(n+d/n)
    =\mathcal{O}_{t,\epsilon,\delta}(\sqrt d).
\]
The finitely many smaller dimensions can be covered by local tomography, with Alice sending her classical estimate to Bob, enlarging only the implicit constant.
\end{proof}

The matching lower bound follows from the distributed pure-state comparison problem of~\cite{Anshu2022DistributedInnerProduct}.
For pure inputs the target is $|\langle\psi|\phi\rangle|^{2t}$, which separates identical states from independent Haar-random states at constant accuracy, requiring $\Omega(\sqrt d)$ copies even with arbitrary local joint operations and classical communication.

\section{Lower bounds}
\label{sec:proof-lower}

We compare two ensembles constructed from a smooth spectral path whose weights have matching low-order power sums.
We split the endpoint total variation distance according to whether the accumulated projector sum is bounded.
On records satisfying this bound, we control the probability derivative using moment matching and martingale estimates for matrix chains and cycles.
The remaining contribution is bounded directly by the probabilities of exceptional records at the two endpoints.

\begin{theorem}[Lower bound for adaptive single-copy moment estimation]
\label[theorem]{lo:main}
For every integer $t\geq2$, there are constants $\epsilon_t,c_t>0$ and $d_t$ such that the following holds for all $d\geq d_t$.
Any adaptive single-copy protocol that, for every state $\rho$ on $\mathbb C^d$, estimates $\operatorname{Tr}(\rho^t)$ to additive error at most $\epsilon_t$ with probability at least $2/3$ has worst-case sample complexity at least
\[
    c_t d^{\lceil\log_2t\rceil/(1+\lceil\log_2t\rceil)}.
\]
The same conclusion holds for every smaller error tolerance.
The protocol may use arbitrary single-copy POVMs and arbitrary classical adaptation, but no quantum memory between copies.
\end{theorem}

Fix $t$ throughout this section.
For the adaptive lower bound, use the default parameters
\begin{equation}
\label{lo:parameters}
    h=\lceil\log_2t\rceil,
    \qquad r=2^{h-1}+1,
    \qquad \kappa=2(h+1).
\end{equation}
In particular, $r\leq t\leq2r-2$.
The likelihood and record bounds below allow any fixed $r$ and $\kappa$; \Cref{sec:proof-nonadaptive} instead uses $r=t$ and $\kappa=2t$.
Constants may depend on $t$, but not on the dimension, strategy, or interpolation parameter.
Unsubscripted logarithms are natural.

\subsection{Hard ensembles and the testing reduction}
\label{lo:spectra-subsection}
\label{lo:priors-subsection}

Write $p_k(\lambda)=\sum_j\lambda_j^k$.
Matching these moments along the entire path eliminates their contributions to the likelihood derivative.

\begin{lemma}[Moment-matching spectral path]
\label[lemma]{lo:spectral-path}
Let $r\geq2$ be fixed and assume $r\leq t\leq2r-2$.
For $0\leq\theta\leq\pi/r$, define
\begin{equation}
\label{lo:lambda}
    \lambda_j(\theta)
    =\frac{1+\cos(2\pi j/r+\theta)}r,
    \qquad 0\leq j<r.
\end{equation}
These vectors are probability vectors, and $p_k(\lambda(\theta))$ is constant in $\theta$ for $1\leq k\leq r-1$.
Their endpoint $t$th moments differ by
\begin{equation}
\label{lo:gap}
    \Delta_t
    :=p_t(\lambda(0))-p_t(\lambda(\pi/r))
    =2^{2-t}r^{1-t}\binom{2t}{t-r}>0.
\end{equation}
Moreover,
\begin{equation}
\label{lo:spectral-derivative}
    0\leq\lambda_j\leq2/r,
    \qquad \sum_j|\lambda_j'|\leq1,
    \qquad \frac{|p_\ell'(\lambda(\theta))|}{\ell}
    \leq(2/r)^{\ell-1}
    \quad(\ell\geq1).
\end{equation}
\end{lemma}

\begin{proof}
Nonnegativity is immediate, and summing the equally spaced phases gives $\sum_j\lambda_j=1$.
The binomial identity
\[
    (1+\cos u)^k
    =2^{-k}\sum_{a=-k}^{k}\binom{2k}{k-a}e^{iau}
\]
follows by writing $1+\cos u=(e^{iu/2}+e^{-iu/2})^2/2$.
Summing over the $r$ phases annihilates every frequency not divisible by $r$.
For $k<r$, only the zero frequency remains, proving moment matching through order $r-1$.
For $r\leq t\leq2r-2$, only the frequencies $0,r,-r$ remain, so
\[
    p_t(\lambda(\theta))
    =2^{-t}r^{1-t}
      \left[\binom{2t}{t}
      +2\binom{2t}{t-r}\cos(r\theta)\right].
\]
Evaluating at the endpoints proves \eqref{lo:gap}.
Finally, $|\lambda_j'|\leq1/r$, and
\[
    |p_\ell'|
    \leq\ell\sum_j\lambda_j^{\ell-1}|\lambda_j'|
    \leq\ell(2/r)^{\ell-1}.
\]
\end{proof}

For $t=3$ or $4$, the same path has $r=3$ and endpoints $(2/3,1/6,1/6)$ and $(1/2,0,1/2)$, up to permutation.
The purity equals $1/2$ throughout the path, whereas $\Delta_3=1/18$ and $\Delta_4=2/27$.
Thus these two orders can be treated using the same pair of hard ensembles.

\paragraph{Physical Gaussian ensembles.}
Let $g_0,\ldots,g_{r-1}$ be independent standard complex Gaussian vectors in $\mathbb C^d$.
We write $\mathbb E_G$ for expectation under this product Gaussian law.
Choose a positive constant $\alpha\leq1/100$, to be made smaller later, and define
\begin{equation}
\label{lo:gaussian-state}
\begin{split}
    R_\theta(g)&=\frac\alpha d\sum_j\lambda_j(\theta)\|g_j\|^2,\\
    \rho_\theta(g)&=
    \frac{1-R_\theta(g)}d I
    +\frac\alpha d\sum_j\lambda_j(\theta)g_jg_j^\dagger.
\end{split}
\end{equation}
To ensure positivity, we truncate the Gaussian prior with a continuously differentiable cutoff.
Define
\begin{equation}
\label{lo:barrier}
    \phi(s)=
    \begin{cases}
       0,&s\leq1/3,\\
       (s-1/3)^2/(1/2-s),&1/3<s<1/2,\\
       +\infty,&s\geq1/2,
    \end{cases}
    \qquad \chi(s)=e^{-\phi(s)}.
\end{equation}
The prior at parameter $\theta$ has density $\chi(R_\theta)/B_\theta$ with respect to the Gaussian law, where $B_\theta=\mathbb E_G\chi(R_\theta)$.
The function $\phi$ is convex, nondecreasing, and $C^1$ on $(-\infty,1/2)$, though only piecewise $C^2$ at $1/3$.
Both $\chi$ and $\chi'$ are bounded and vanish at $1/2$.
Every state in this prior satisfies
\begin{equation}
\label{lo:physical-floor}
    \rho_\theta(g)\succeq\frac{I}{2d}.
\end{equation}
The prior is unitarily invariant.

\begin{lemma}[Concentration of the hard ensembles]
\label[lemma]{lo:prior-concentration}
Uniformly over the interpolation parameter, $B_\theta=1-\mathcal{O}(e^{-c d})$ for an absolute constant $c>0$.
For every fixed integer $k\geq2$ and every fixed parameter $\theta$, the random state drawn from this prior satisfies
\begin{equation}
\label{lo:moment-limit}
    \operatorname{Tr}(\rho_\theta^k)
    \longrightarrow\alpha^k p_k(\lambda(\theta))
    \quad\text{in probability as }d\longrightarrow\infty.
\end{equation}
\end{lemma}

\begin{proof}
The Gaussian exponential-moment formula in \Cref{pre:gaussian} gives
\[
    \mathbb E_G e^{dR_\theta}
    =\prod_j(1-\alpha\lambda_j)^{-d}
    \leq\exp\!\left(\frac{\alpha d}{1-\alpha}\right).
\]
Here $-\log(1-x)\leq x/(1-\alpha)$ for $0\leq x\leq\alpha$, and $\sum_j\lambda_j=1$.
Markov's inequality therefore bounds $\mathbb P_G(R_\theta>1/3)$ by $e^{-c d}$, uniformly in $\theta$.
Since $\chi=1$ on $(-\infty,1/3]$ and $0\leq\chi\leq1$, this proves the assertion about $B_\theta$.
The total variation distance between the tilted and Gaussian laws is at most $(1-B_\theta)/B_\theta=\mathcal{O}(e^{-c d})$, by integrating $|\chi(R_\theta)/B_\theta-1|$.

Since $r$ is fixed, the normalized Gram matrix $(g_j^\dagger g_k/d)_{j,k}$ converges in probability to $I_r$.
The low-rank term in \eqref{lo:gaussian-state} and the following $r\times r$ matrix have the same nonzero eigenvalues:
\[
    \left(\frac{\alpha\sqrt{\lambda_j(\theta)\lambda_k(\theta)}}d
                  g_j^\dagger g_k\right)_{j,k}.
\]
The eigenvalues of this matrix converge to $\alpha\lambda_j(\theta)$.
The scalar shift in \eqref{lo:gaussian-state} is $\mathcal{O}(1/d)$ on the tilted support, and the remaining eigenvalues contribute at most $d^{1-k}$.
The total variation comparison transfers Gaussian convergence to the tilted prior, proving \eqref{lo:moment-limit}.
\end{proof}

Draw one state from the prior and give the protocol copies of that same state, without resampling the Gaussian vectors.
Let $P_\theta^{(T)}$ be its transcript law, including any classical random seed, as in the learning model of \Cref{pre:learning-model}.

Choose
\begin{equation}
\label{lo:epsilon}
    \epsilon_t=\frac{\alpha^t\Delta_t}{8}.
\end{equation}
The endpoint limits differ by $8\epsilon_t$.
On an ensemble event of probability $1-o(1)$, any successful estimate of accuracy $\epsilon_t$ is within $2\epsilon_t$ of its endpoint limit.
Thresholding halfway between the limits therefore distinguishes the priors with success at least $2/3-o(1)$.
Le Cam's testing identity~\cite{LeCam1986}, in the form of \Cref{pre:testing}, therefore requires
\begin{equation}
\label{lo:necessary-tv}
    \operatorname{TV}(P_0^{(T)},P_{\pi/r}^{(T)})
    \geq\frac13-o(1).
\end{equation}
We will contradict this inequality when $T$ is a sufficiently small constant multiple of $d^{h/(h+1)}$.

\subsection{Original record probabilities and the two-part bound}
\label{lo:records-subsection}

As in the learning-tree representation of~\cite{Chen2022QuantumMemory}, refine every POVM into rank-one effects $w_yP_y$, where $P_y=v_yv_y^\dagger$, $\|v_y\|=1$, and $\sum_yw_yP_y=I$.
Retaining refinement labels can only improve distinguishability, so a lower bound for the refined transcript suffices.
Sums become integrals for continuous outcomes.

Once a history $H$ is frozen, its observed projectors $P_i$ and weights $w_i$ are fixed, even though the protocol chose them adaptively.
For each round, put
\begin{equation}
\label{lo:likelihood-factors}
    X_i=P_i-I/d,
    \qquad
    f_i=1+\alpha\sum_j\lambda_j g_j^\dagger X_i g_j
       =d\operatorname{Tr}(P_i\rho_\theta).
\end{equation}
On the support of the tilted prior,
\begin{equation}
\label{lo:f-bounds}
    1/2\leq f_i\leq1+d/2.
\end{equation}
The factors $w_i/d$ are independent of $\theta$; the protocol run on $I/d$ therefore gives a common probability measure $\nu$ on records, including the random seed.

Define the physical and auxiliary normalizers by
\begin{equation}
\label{lo:normalizers}
    Z_\theta(H)
    =\mathbb E_G\!\left[\chi(R_\theta)\prod_{i\in H}f_i\right],
    \qquad
    \widetilde Z_\theta(H)=\mathbb E_G\prod_{i\in H}f_i.
\end{equation}
Only $Z_\theta$ describes a physical prior; $\widetilde Z_\theta$ is an auxiliary polynomial expectation.
The notation $H\setminus J$ deletes likelihood factors while keeping the original frozen projectors, without rerunning the protocol.

Thus the density of the original record law relative to $\nu$ is
\begin{equation}
\label{lo:original-density}
    p_\theta(H)=\frac{Z_\theta(H)}{B_\theta}.
\end{equation}
For discrete records, $p_\theta(H)=P_\theta^{(T)}(H)/\nu(H)$ whenever $\nu(H)>0$.
For a fixed state the record probability is a product of single-copy probabilities, and its derivative differentiates one factor at a time.
The average over the unknown state is taken after forming this product, since all copies come from the same sampled state.
We compare the endpoint densities through the identity
\begin{equation}
\label{lo:derivative-target}
    p_{\pi/r}(H)-p_0(H)
       =\int_0^{\pi/r}\partial_\theta p_\theta(H)\,d\theta.
\end{equation}
Below we use this identity on good records and bound the remaining contribution to total variation directly by its probability.
For fixed $d,T$, the bounds $|f_i|,|\partial_\theta f_i|\leq1+\alpha\sum_j\|g_j\|^2$, the bounded cutoff and its derivative, and Gaussian moments justify differentiation under the integral.
They also give continuous differentiability in $L^1(\nu)$, as required by \Cref{pre:tv-path}.
The passage from local derivatives to endpoint distinguishability is the same smooth-path principle underlying Fisher-information bounds such as~\cite[App.~G]{Zhou2026AnsatzFreeHamiltonian}; the estimates below exploit the specific moment-matching structure of our ensembles.

\paragraph{Splitting the measurement records.}
All measurements below belong to the original experiment, which runs for all $T$ rounds.
For its observed projectors, define
\begin{equation}
\label{lo:good-event}
    S_n=\sum_{i=1}^nP_i,
    \qquad S_0=0,
    \qquad
    G_\kappa
    =\{H:\|S_T\|_{\mathrm{op}}<\kappa\},
\end{equation}
where $\kappa\geq1$ is a fixed threshold to be chosen later.
To analyze the records, introduce the following indicators, each determined by the preceding record, and auxiliary matrix sums
\begin{equation}
\label{lo:retained-sum}
    a_i=\boldsymbol 1_{\{\|S_{i-1}\|_{\mathrm{op}}<\kappa\}},
    \qquad
    \widetilde S_n=\sum_{i=1}^na_iP_i.
\end{equation}
Because $S_n$ is increasing in the positive-semidefinite order, the indicators remain zero after the threshold is reached.
The auxiliary sum retains the threshold-crossing observation and therefore satisfies $\|\widetilde S_n\|_{\mathrm{op}}\leq\kappa+1$ for every record.
On $G_\kappa$, all indicators equal one, and $\widetilde S_n=S_n$.
These indicators are only analytical tools: the actual protocol, its later outcomes, and the likelihood factors $X_i=P_i-I/d$ and $f_i$ are unchanged.

\begin{lemma}[Posterior mean and probability of exceptional records]
\label[lemma]{lo:posterior-records}
If $\alpha\leq1/(8\kappa)$, then every history $H_n$ satisfying $\|S_n\|_{\mathrm{op}}<\kappa$ obeys
\begin{equation}
\label{lo:posterior-mean}
    \frac{I}{2d}\preceq
    \mathbb E_\theta[\rho_\theta\mid H_n]
    \preceq\frac{2I}{d}.
\end{equation}
For every $T\leq d$ with $2T<\kappa d$, uniformly over the interpolation parameter and the adaptive measurement strategy,
\begin{equation}
\label{lo:rare-records}
    P_\theta^{(T)}(G_\kappa^c)
    \leq d\left(\frac{2eT}{\kappa d}\right)^\kappa.
\end{equation}
\end{lemma}

\begin{proof}
Fix a history $H_n$ with $\|S_n\|_{\mathrm{op}}<\kappa$ and regard all its projectors as fixed.
By Bayes' rule, its posterior density with respect to the product Gaussian law is $\chi(R_\theta)\prod_{i=1}^nf_i/Z_\theta(H_n)$.
The prior normalizer $B_\theta$ and the factors $w_i/d$ cancel in this conditional density.
Writing the Gaussian coordinates as one real vector $x$, the negative log posterior is therefore, up to an additive constant,
\[
    V(x)=\|x\|^2+\phi(R_\theta(x))
                      -\sum_{i=1}^n\log f_i(x).
\]

We establish its strong convexity directly from the quadratic likelihood factors.
For $x$ and $x+u$ in the convex support $\{x:R_\theta(x)<1/2\}$, write $u=(u_0,\ldots,u_{r-1})$ in complex coordinates.
The exact quadratic expansion and the inequality $\log z\leq z-1$ give
\begin{align*}
    f_i(x+u)
       &=f_i(x)+\mathrm D f_i(x)[u]
          +\alpha\sum_j\lambda_j u_j^\dagger X_i u_j,\\
    \log f_i(x+u)
       &\leq\log f_i(x)+\mathrm D\log f_i(x)[u]
          +2\alpha\sum_j\lambda_j u_j^\dagger P_i u_j.
\end{align*}
Here $\mathrm D$ denotes the real directional derivative.
For the second inequality, first use $X_i\preceq P_i$ and then $f_i(x)\geq1/2$; the latter comparison is applied only to the nonnegative quadratic form involving $P_i$.
Summing over $i$, and using $S_n\preceq\kappa I$ and $\lambda_j\leq1$, bounds the total quadratic remainder by $2\alpha\kappa\|u\|^2$.
The barrier $\phi\circ R_\theta$ is convex because $R_\theta$ is a nonnegative quadratic form and $\phi$ is convex and nondecreasing.
Consequently,
\begin{equation}
\label{lo:posterior-strong-convexity}
    V(x+u)\geq V(x)+\mathrm D V(x)[u]
             +\left(1-2\alpha\kappa\right)\|u\|^2
    \geq V(x)+\mathrm D V(x)[u]+\frac12\|u\|^2.
\end{equation}
The stated bound on $\alpha$ implies the last inequality.
Thus $V(x)-\|x\|^2/2$ is convex on this support.
The barrier diverges at the boundary of this convex support, so extension by $+\infty$ preserves convexity on the full real space.

The posterior is invariant under multiplying any individual $g_j$ by a complex phase, so $\mathbb E_\theta[g_j\mid H_n]=0$; this does not assert independence of the posterior vectors.
We use the covariance bound from the Brascamp--Lieb inequality~\cite{Carlen2013BrascampLieb}, stated in \Cref{pre:bras-lieb}.
Applying it separately to the real and imaginary parts of $v^\dagger g_j$ yields
\[
    \mathbb E_\theta[|v^\dagger g_j|^2\mid H_n]
       \leq2\|v\|^2,
    \qquad
    \mathbb E_\theta[g_jg_j^\dagger\mid H_n]\preceq2I.
\]
Dropping the negative scalar term in \eqref{lo:gaussian-state} gives
\[
    \mathbb E_\theta[\rho_\theta\mid H_n]
       \preceq\frac{1+2\alpha}{d}I
       \preceq\frac{2I}{d}.
\]
The lower bound follows pointwise from \eqref{lo:physical-floor}.

It remains to control the probability of exceptional records under the original experiment.
If $a_i=1$, the posterior bound implies that the next outcome $y$ has conditional probability at most $2w_y/d$.
POVM completeness therefore gives
\[
    \mathbb E_\theta[a_iP_i\mid H_{i-1}]
       \preceq\frac2d\sum_yw_yP_y=\frac{2I}{d}.
\]
If $a_i=0$, this conditional expectation is zero, so the inequality still holds without any posterior bound for that history.
We apply the positive-matrix concentration bound in \Cref{pre:matrix-chernoff}, whose proof uses the Golden--Thompson inequality~\cite{Forrester2014GoldenThompson}.
For $0\preceq a_iP_i\preceq I$, it gives
\[
    \mathbb P_\theta(\|\widetilde S_T\|_{\mathrm{op}}\geq\kappa)
       \leq d\left(\frac{2eT}{\kappa d}\right)^\kappa.
\]
The event on the left equals $G_\kappa^c$: whenever the raw sum first reaches the threshold, the corresponding observation is included in $\widetilde S_T$, and conversely $\widetilde S_T\preceq S_T$.
This proves \eqref{lo:rare-records}.
\end{proof}

\paragraph{Separating good and rare records.}
The event $G_\kappa$ depends only on the record and the fixed measurement strategy, so it is the same event for every $\theta$.
We split the endpoint total variation distance using this event.

\begin{lemma}[Total variation from good records]
\label[lemma]{lo:tv-split}
Under the assumptions of \Cref{lo:posterior-records},
\begin{equation}
\label{lo:tv-good-rare}
    \operatorname{TV}(P_0^{(T)},P_{\pi/r}^{(T)})
    \leq\frac12\int_0^{\pi/r}
          \int_{G_\kappa}|\partial_\theta p_\theta(H)|
                    \,d\nu(H)\,d\theta
       +d\left(\frac{2eT}{\kappa d}\right)^\kappa.
\end{equation}
\end{lemma}

\begin{proof}
By the definition of total variation,
\[
\begin{aligned}
    \operatorname{TV}(P_0^{(T)},P_{\pi/r}^{(T)})
    ={}&\frac12\int_{G_\kappa}|p_{\pi/r}-p_0|\,d\nu
       +\frac12\int_{G_\kappa^c}|p_{\pi/r}-p_0|\,d\nu.
\end{aligned}
\]
On $G_\kappa$, the fundamental theorem of calculus and the triangle inequality give
\[
    \frac12\int_{G_\kappa}|p_{\pi/r}-p_0|\,d\nu
    \leq\frac12\int_0^{\pi/r}
          \int_{G_\kappa}|\partial_\theta p_\theta|\,d\nu\,d\theta.
\]
On $G_\kappa^c$, nonnegativity of the densities gives
\[
\begin{aligned}
    \frac12\int_{G_\kappa^c}|p_{\pi/r}-p_0|\,d\nu
    &\leq\frac12\left[
         P_0^{(T)}(G_\kappa^c)+P_{\pi/r}^{(T)}(G_\kappa^c)
         \right]\\
    &\leq d\left(\frac{2eT}{\kappa d}\right)^\kappa,
\end{aligned}
\]
where the last step uses \Cref{lo:posterior-records} at the two endpoints.
Adding the bounds proves \eqref{lo:tv-good-rare}.
\end{proof}

Thus it remains to bound the probability derivative only on $G_\kappa$.
All probabilities refer to the original experiment, which still runs for all $T$ rounds.

\subsection{Matrix martingales and cycle estimates}
\label{lo:centered-subsection}
\label{lo:cycles-subsection}

All expectations below remain under the original joint law of the latent vectors and the complete measurement record.
For fixed $\theta$, let $\mathcal F_i=\sigma(g,z,y_1,\ldots,y_i)$ contain the latent Gaussian vectors, the random seed, and the first $i$ outcomes.
The latent vectors are included only in the analysis and are not given to the learner.
The indicators $a_i$ from \eqref{lo:retained-sum} are predictable: they are determined before outcome $i$ is observed.
Define
\begin{equation}
\label{lo:Y}
    \widehat Y_i=X_i/f_i,\qquad Y_i=a_i\widehat Y_i.
\end{equation}
The hatted matrices use every observation, whereas the unhatted matrices retain observations only through the first threshold crossing.
They coincide throughout every record in $G_\kappa$.
This modification concerns only the auxiliary matrices; the original experiment still runs for all $T$ rounds.

\begin{lemma}[Martingale differences and uniform matrix bounds]
\label[lemma]{lo:Y-properties}
Uniformly over $\theta$ and all adaptive strategies, both $\widehat Y_i$ and $Y_i$ satisfy the following conditional bounds (written for $Y_i$):
\begin{align}
    \mathbb E_\theta[Y_i\mid\mathcal F_{i-1}]&=0,
       \label{lo:Y-center}\\
    \mathbb E_\theta[Y_i^2\mid\mathcal F_{i-1}]&\preceq\frac2d I,
       \label{lo:Y-matrix-variance}\\
    \mathbb E_\theta[|\operatorname{Tr}(MY_i)|^2
                        \mid\mathcal F_{i-1}]
       &\leq\frac2d\|M\|_{\mathrm{HS}}^2
       \label{lo:Y-scalar-variance}
\end{align}
for every complex matrix $M$ determined by $\mathcal F_{i-1}$.
For every copy budget $T$, the auxiliary matrices $Y_i$ additionally satisfy the deterministic bound
\begin{equation}
\label{lo:Y-subsets}
    \left\|\sum_{i\in J}Y_i\right\|_{\mathrm{op}}
    \leq K:=2(\kappa+2)
    \qquad\text{for every }J\subseteq\{1,\ldots,T\}.
\end{equation}
These constants are independent of the eventual small choice of $\alpha$.
\end{lemma}

\begin{proof}
Conditional on the latent vectors and the past, outcome $y$ has probability $w_yf_y/d$.
First consider the original matrix $\widehat Y_i$.
POVM completeness gives
\[
    \mathbb E_\theta[\widehat Y_i\mid\mathcal F_{i-1}]
    =\frac1d\sum_yw_y(P_y-I/d)=0,
\]
where $\sum_yw_y=d$ follows by taking traces in the POVM normalization.
For the matrix second moment, $\rho_\theta\succeq I/(2d)$ implies $f_y\geq1/2$, and hence
\[
    \mathbb E_\theta[\widehat Y_i^2\mid\mathcal F_{i-1}]
    =\frac1d\sum_y\frac{w_y}{f_y}(P_y-I/d)^2
    \preceq\frac2d\sum_yw_y(P_y-I/d)^2
    =\frac2d(1-1/d)I.
\]
Indeed, $P_y^2=P_y$, $\sum_yw_yP_y=I$, and $\sum_yw_y=d$ give
$\sum_yw_y(P_y-I/d)^2=(1-2/d)I+I/d=(1-1/d)I$.
Similarly,
\[
    \mathbb E_\theta[|\operatorname{Tr}(M\widehat Y_i)|^2
                         \mid\mathcal F_{i-1}]
    \leq\frac2d\sum_yw_y
        |\operatorname{Tr}(M(P_y-I/d))|^2.
\]
The sum on the right satisfies
\begin{align*}
    \sum_yw_y|\operatorname{Tr}(M(P_y-I/d))|^2
      &=\sum_yw_y|v_y^\dagger Mv_y|^2
               -\frac{|\operatorname{Tr}M|^2}{d}\\
      &\leq\sum_yw_y v_y^\dagger M^\dagger Mv_y
       =\|M\|_{\mathrm{HS}}^2.
\end{align*}
Here $|v_y^\dagger Mv_y|^2\leq v_y^\dagger M^\dagger Mv_y$ follows from Cauchy--Schwarz and $\|v_y\|=1$.
This proves the scalar bound, including for non-Hermitian $M$.
Since $a_i\in\{0,1\}$ is predictable, multiplication by $a_i$ preserves the martingale-difference property by \Cref{pre:orthogonality} and can only decrease the two conditional second moments.

It remains to bound every subset sum deterministically.
Fix $J\subseteq\{1,\ldots,T\}$ and collect the positive projector terms into
\[
    M_J=\sum_{i\in J}\frac{a_i}{f_i}P_i.
\]
Since $0\leq a_i/f_i\leq2a_i$ and every $P_i$ is positive semidefinite,
\[
    0\preceq M_J
       \preceq2\sum_{i\in J}a_iP_i
       \preceq2\sum_{i=1}^T a_iP_i
       =2\widetilde S_T.
\]
Recall why $\|\widetilde S_T\|_{\mathrm{op}}\leq\kappa+1$.
Immediately before the first threshold crossing, the accumulated projector sum has norm less than $\kappa$.
The crossing observation adds one projector of norm one, and all subsequent auxiliary terms vanish.
If the threshold is never reached, the norm remains below $\kappa$.
Thus, in either case, $\|M_J\|_{\mathrm{op}}\leq2(\kappa+1)$.

To recover the centered matrices $Y_i$, use $\operatorname{Tr}P_i=1$:
\[
    \operatorname{Tr}M_J=\sum_{i\in J}\frac{a_i}{f_i},
    \qquad
    \sum_{i\in J}Y_i
       =M_J-\frac{\operatorname{Tr}M_J}{d}I.
\]
The scalar $\operatorname{Tr}M_J/d$ is the average eigenvalue of $M_J$.
Because $M_J\succeq0$, both this average and every eigenvalue of $M_J$ lie in $[0,\|M_J\|_{\mathrm{op}}]$.
Subtracting the average therefore leaves every eigenvalue with absolute value at most $\|M_J\|_{\mathrm{op}}$.
Consequently,
\[
    \left\|\sum_{i\in J}Y_i\right\|_{\mathrm{op}}
       =\left\|M_J-\frac{\operatorname{Tr}M_J}{d}I\right\|_{\mathrm{op}}
       \leq\|M_J\|_{\mathrm{op}}
       \leq2(\kappa+1)\leq K.
\]
This holds for every record and every subset $J$, without a restriction on $T$.
\end{proof}

In particular, the partial sums of either $\widehat Y_i$ or $Y_i$ are matrix martingales under the original joint law.
The following lemma controls an increment multiplied on both sides by matrices determined by the past.

\begin{lemma}[Second moment of a matrix product]
\label[lemma]{lo:sandwich-lemma}
Under the assumptions of \Cref{lo:Y-properties}, let $L,R\in\mathbb C^{d\times d}$ be matrices determined by $\mathcal F_{i-1}$.
Then
\begin{equation}
\label{lo:sandwich}
    \mathbb E_\theta\!\left[
        \|LY_iR\|_{\mathrm{HS}}^2
        \,\middle|\,\mathcal F_{i-1}\right]\leq\frac2d\min\left\{
        \|L\|_{\mathrm{op}}^2\|R\|_{\mathrm{HS}}^2,\,
        \|L\|_{\mathrm{HS}}^2\|R\|_{\mathrm{op}}^2
    \right\}.
\end{equation}
Moreover,
$\mathbb E_\theta[\|Y_i\|_{\mathrm{HS}}^2
\mid\mathcal F_{i-1}]\leq2$.
\end{lemma}

\begin{proof}
Conditional on $\mathcal F_{i-1}$, the matrices $L,R$ and their norms are fixed.
We first prove the bound involving
$\|L\|_{\mathrm{op}}^2\|R\|_{\mathrm{HS}}^2$.

By the definition of the operator norm,
\[
    0\preceq L^\dagger L
      \preceq\|L\|_{\mathrm{op}}^2I.
\]
Multiplying this inequality on the left by $(Y_iR)^\dagger$
and on the right by $Y_iR$ preserves the positive-semidefinite order.
Since $Y_i$ is Hermitian, this gives
\[
    R^\dagger Y_iL^\dagger LY_iR
    \preceq
    \|L\|_{\mathrm{op}}^2R^\dagger Y_i^2R.
\]
Taking traces and using the definition of the Hilbert--Schmidt norm, we obtain
\[
\begin{aligned}
    \|LY_iR\|_{\mathrm{HS}}^2
    &=\operatorname{Tr}
        \bigl((LY_iR)^\dagger(LY_iR)\bigr)\\
    &=\operatorname{Tr}
        (R^\dagger Y_iL^\dagger LY_iR)\\
    &\leq\|L\|_{\mathrm{op}}^2
        \operatorname{Tr}(R^\dagger Y_i^2R).
\end{aligned}
\]

We now take conditional expectation.
Because $L$ and $R$ are determined by the past, only $Y_i^2$ remains inside the conditional expectation.
The bound
$\mathbb E_\theta[Y_i^2\mid\mathcal F_{i-1}]
\preceq2I/d$
from \eqref{lo:Y-matrix-variance} therefore yields
\[
\begin{aligned}
    &\mathbb E_\theta\!\left[
        \|LY_iR\|_{\mathrm{HS}}^2
        \,\middle|\,\mathcal F_{i-1}\right]\\
    &\quad\leq
        \|L\|_{\mathrm{op}}^2
        \operatorname{Tr}\!\left(
            R^\dagger
            \mathbb E_\theta[Y_i^2\mid\mathcal F_{i-1}]
            R\right)\\
    &\quad\leq
        \frac2d\|L\|_{\mathrm{op}}^2
        \operatorname{Tr}(R^\dagger R)\\
    &\quad=
        \frac2d\|L\|_{\mathrm{op}}^2
        \|R\|_{\mathrm{HS}}^2.
\end{aligned}
\]

To obtain the other bound, use invariance of the Hilbert--Schmidt norm under conjugate transpose:
\[
    \|LY_iR\|_{\mathrm{HS}}
    =\|(LY_iR)^\dagger\|_{\mathrm{HS}}
    =\|R^\dagger Y_iL^\dagger\|_{\mathrm{HS}}.
\]
Apply the bound just proved with $L$ replaced by $R^\dagger$
and $R$ replaced by $L^\dagger$.
Since conjugate transpose also preserves the operator norm,
\[
\begin{aligned}
    \mathbb E_\theta\!\left[
        \|LY_iR\|_{\mathrm{HS}}^2
        \,\middle|\,\mathcal F_{i-1}\right]
    &\leq\frac2d
        \|R^\dagger\|_{\mathrm{op}}^2
        \|L^\dagger\|_{\mathrm{HS}}^2\\
    &=\frac2d
        \|L\|_{\mathrm{HS}}^2
        \|R\|_{\mathrm{op}}^2.
\end{aligned}
\]
Both bounds hold simultaneously, so we may take their minimum.

Finally, Hermiticity of $Y_i$ gives
$\|Y_i\|_{\mathrm{HS}}^2=\operatorname{Tr}(Y_i^2)$.
Taking a trace in \eqref{lo:Y-matrix-variance}, we find
\[
    \mathbb E_\theta\!\left[
        \|Y_i\|_{\mathrm{HS}}^2
        \,\middle|\,\mathcal F_{i-1}\right]
    =
    \operatorname{Tr}\!\left(
        \mathbb E_\theta[Y_i^2\mid\mathcal F_{i-1}]
    \right)
    \leq\operatorname{Tr}(2I/d)=2.
\]
\end{proof}

\paragraph{The dyadic hierarchy.}

For $s\geq1$ and $N\leq T$, define the ordered distinct-index chain
\begin{equation}
\label{lo:chain-definition}
    \mathcal A_{s,N}
    =\sum_{\substack{i_1,\ldots,i_s\leq N\\\text{all distinct}}}
        Y_{i_1}\cdots Y_{i_s},
    \qquad
    \mathcal C_{s,N}=\operatorname{Tr}(\mathcal A_{s,N}).
\end{equation}
The sum is zero for $s>N$; products retain their order.
Write $\widehat{\mathcal C}_{s,N}$ for the same cycle sum formed from $\widehat Y_i$.
On $G_\kappa$, the original and auxiliary sums agree for every choice of the latent vectors in the physical prior's support.

\begin{lemma}[All-order chain and cycle bounds]
\label[lemma]{lo:cycle-bounds}
Set $H(s)=1+\lfloor\log_2s\rfloor$.
Thus $H(s)-1$ counts the number of halvings, rounding down at each step, needed to reduce $s$ to one.
The matrices of \Cref{lo:Y-properties} satisfy, for all $s\geq1$ and $N\leq T$,
\begin{equation}
\label{lo:chain-bound}
    \mathbb E_\theta\|\mathcal A_{s,N}\|_{\mathrm{HS}}^2
    \leq s^{2s}(2K^2)^s
        \frac{N^{H(s)}}{d^{H(s)-1}}.
\end{equation}
For every $\ell\geq2$,
\begin{equation}
\label{lo:cycle-bound-explicit}
    \mathbb E_\theta|\mathcal C_{\ell,N}|^2
    \leq
    2\ell^2(\ell-1)^{2\ell-2}(2K^2)^{\ell-1}
    \frac{N^{\lceil\log_2\ell\rceil+1}}
         {d^{\lceil\log_2\ell\rceil}}.
\end{equation}
Consequently, there is a constant $A=A(K)\geq1$, independent of the order, such that
\begin{equation}
\label{lo:cycle-bound-quantitative}
    \|\mathcal C_{\ell,N}\|_{L^2}
    \leq
    \exp[A\ell\log(\ell+1)]
    \frac{N^{(\lceil\log_2\ell\rceil+1)/2}}
         {d^{\lceil\log_2\ell\rceil/2}}.
\end{equation}
Here
$\|\mathcal C_{\ell,N}\|_{L^2}
=(\mathbb E_\theta|\mathcal C_{\ell,N}|^2)^{1/2}$
denotes the root-mean-square magnitude of the random scalar
$\mathcal C_{\ell,N}$.
The expectation is taken under the original joint law of the latent vectors, the random seed, and the measurement outcomes.
\end{lemma}

\begin{proof}
We first bound the matrix chains and then pass to their traces.
To handle the distinct-index restriction, independently assign each round a uniformly random color from $\{1,\ldots,s\}$ and define
\[
    B_{a,n}
    =\sum_{\substack{i\leq n\\\mathrm{color}(i)=a}}Y_i.
\]
Expanding $B_{1,N}\cdots B_{s,N}$ selects one index from each color class, so no index can occur twice.
Conversely, each ordered tuple of $s$ distinct indices receives the prescribed colors with probability $s^{-s}$.
Therefore, for every fixed realization of the matrices,
\begin{equation}
\label{lo:coloring-identity}
    \mathcal A_{s,N}
    =s^s\mathbb E_{\mathrm{colors}}
        [B_{1,N}\cdots B_{s,N}].
\end{equation}
Jensen's inequality gives
\[
    \mathbb E_\theta\|\mathcal A_{s,N}\|_{\mathrm{HS}}^2
    \leq
    s^{2s}\mathbb E_{\mathrm{colors}}
        \mathbb E_\theta
        \|B_{1,N}\cdots B_{s,N}\|_{\mathrm{HS}}^2.
\]
It thus suffices to bound these products uniformly over all fixed colorings.

Fix a coloring.
For an ordered list $w=(a_1,\ldots,a_u)$ of distinct colors, write
$M_{w,n}=B_{a_1,n}\cdots B_{a_u,n}$.
Every factor has operator norm at most $K$ by \eqref{lo:Y-subsets}.
At round $n$, only the factor corresponding to the color of $n$ can change.
Hence
$\Delta M_{w,n}:=M_{w,n}-M_{w,n-1}$
is either zero or $LY_nR$, where $L,R$ are products of factors evaluated at time $n-1$; empty products mean $I$.
These factors are determined by the past, so
$\mathbb E_\theta[Y_n\mid\mathcal F_{n-1}]=0$
implies
$\mathbb E_\theta[\Delta M_{w,n}\mid\mathcal F_{n-1}]=0$.
Thus $M_{w,n}$ is a matrix martingale starting at zero.
This uses the fact that only one factor changes in each round; no independence between color classes is needed.

We prove by induction on $u$, uniformly over colorings and ordered lists, that
\begin{equation}
\label{lo:word-bound}
    \mathbb E_\theta\|M_{w,N}\|_{\mathrm{HS}}^2
    \leq
    (2K^2)^u\frac{N^{H(u)}}{d^{H(u)-1}}.
\end{equation}
For $u=1$, the martingale isometry and
$\mathbb E_\theta\|Y_n\|_{\mathrm{HS}}^2\leq2$
give
$\mathbb E_\theta\|M_{w,N}\|_{\mathrm{HS}}^2
\leq2N\leq2K^2N$.

For $u\geq2$, put $u_0=\lfloor u/2\rfloor$.
In a nonzero increment $LY_nR$, the two sides $L,R$ contain $u-1$ old factors in total, so at least one side contains $u_0$ factors.
Retain a contiguous block of $u_0$ factors on that side, denoted by $M_{w',n-1}$.
In \eqref{lo:sandwich}, use the Hilbert--Schmidt norm for this side and the operator norm for the other.
Bounding all factors outside the retained block by $K$ gives
\[
    \mathbb E_\theta
        \|\Delta M_{w,n}\|_{\mathrm{HS}}^2
    \leq
    \frac{2K^{2(u-1-u_0)}}d
    \mathbb E_\theta
        \|M_{w',n-1}\|_{\mathrm{HS}}^2.
\]
The retained block may depend on the fixed color at round $n$, which is harmless because the induction bound is uniform over all such blocks.
If the color does not occur in $w$, the increment is zero.

The martingale isometry~\cite[Sec.~12.1]{Williams1991ProbabilityMartingales}, in the form of \Cref{pre:orthogonality}, now yields
\begin{align*}
    \mathbb E_\theta\|M_{w,N}\|_{\mathrm{HS}}^2
    &=\sum_{n=1}^N
        \mathbb E_\theta
        \|\Delta M_{w,n}\|_{\mathrm{HS}}^2\\
    &\leq
        \frac{2K^{2(u-1-u_0)}(2K^2)^{u_0}}
             {d^{H(u_0)}}
        \sum_{n=1}^N(n-1)^{H(u_0)}\\
    &\leq
        (2K^2)^u
        \frac{N^{H(u)}}{d^{H(u)-1}}.
\end{align*}
For the last step, use
$\sum_{n=1}^N(n-1)^{H(u_0)}
\leq N^{H(u_0)+1}$
and $H(u)=H(u_0)+1$.
The constants satisfy
$2K^{2(u-1-u_0)}(2K^2)^{u_0}
=2^{u_0+1}K^{2u-2}\leq(2K^2)^u$,
since $u_0+1\leq u$ and $K\geq1$.
This proves \eqref{lo:word-bound}.
Applying it with $u=s$ in the coloring estimate proves \eqref{lo:chain-bound}.
The recursion shows explicitly why each halving of the product length introduces one additional factor of $N/d$.

We next bound the cycles.
The difference
$\mathcal C_{\ell,n}-\mathcal C_{\ell,n-1}$
contains exactly the terms involving the new index $n$.
Separate these terms according to the $\ell$ possible positions of $Y_n$.
Cyclically moving $Y_n$ to the front and reindexing the remaining ordered tuples gives the same sum for every position, so
\begin{equation}
\label{lo:cycle-increment}
    \mathcal C_{\ell,n}-\mathcal C_{\ell,n-1}
    =\ell\operatorname{Tr}
        (Y_n\mathcal A_{\ell-1,n-1}).
\end{equation}
Because $\mathcal A_{\ell-1,n-1}$ is determined by the past, these increments have conditional mean zero.
Thus $\mathcal C_{\ell,n}$ is a scalar martingale starting at zero.
Its isometry, the conditional trace bound \eqref{lo:Y-scalar-variance}, and \eqref{lo:chain-bound} give
\begin{align*}
    \mathbb E_\theta|\mathcal C_{\ell,N}|^2
    &=\sum_{n=1}^N
        \mathbb E_\theta
        |\mathcal C_{\ell,n}-\mathcal C_{\ell,n-1}|^2\\
    &\leq
        \frac{2\ell^2}d
        \sum_{n=1}^N
        \mathbb E_\theta
        \|\mathcal A_{\ell-1,n-1}\|_{\mathrm{HS}}^2\\
    &\leq
        2\ell^2(\ell-1)^{2\ell-2}(2K^2)^{\ell-1}
        \frac{N^{H(\ell-1)+1}}{d^{H(\ell-1)}}.
\end{align*}
Since $H(\ell-1)=\lceil\log_2\ell\rceil$, this proves \eqref{lo:cycle-bound-explicit}.

Finally, taking square roots gives the prefactor
$\sqrt2\,\ell(\ell-1)^{\ell-1}(2K^2)^{(\ell-1)/2}$.
Its logarithm is at most $A\ell\log(\ell+1)$ for a sufficiently large constant $A$ depending only on $K$.
Together with the definition of the scalar $L^2$ norm, this proves \eqref{lo:cycle-bound-quantitative}.
\end{proof}

Each halving of the product length contributes one further factor of $N/d$ to the chain estimate.
Its powers of $N,d$ are therefore constant throughout $2^j\leq s<2^{j+1}$ for each integer $j\geq0$.
The cycle estimate uses a chain of length $\ell-1$, shifting these blocks to $2^j+1\leq\ell\leq2^{j+1}$; the constants may still depend on the order.
In particular, cycle variances are $\mathcal{O}(T^3/d^2)$ for lengths $3,4$ and $\mathcal{O}_\ell(T^4/d^3)$ for lengths $5,\ldots,8$.
For the eventual tail sum we also need an exponential, rather than superexponential, bound in the degree.

\begin{lemma}[Deterministic bound for long chains]
\label[lemma]{lo:long-chains}
For every $\ell\geq1$ and every realization satisfying \eqref{lo:Y-subsets},
\begin{equation}
\label{lo:long-chain-bound}
    \|\mathcal A_{\ell,N}\|_{\mathrm{op}}\leq(6K)^\ell,
    \qquad |\mathcal C_{\ell,N}|\leq d(6K)^\ell.
\end{equation}
\end{lemma}
\begin{proof}
Fix a realization of $Y_1,\ldots,Y_N$.
For independent Bernoulli variables $\xi_i$ of mean $z\in[0,1]$, define
\[
    F(z)=\mathbb E_\xi\left(\sum_{i=1}^N\xi_iY_i\right)^\ell.
\]
Each selected subset sum has operator norm at most $K$, so $\|F(z)\|_{\mathrm{op}}\leq K^\ell$ on $[0,1]$.
In the expansion, $\xi_i^m=\xi_i$ for every positive integer $m$.
A product involving $s$ distinct indices therefore contributes the factor $z^s$, regardless of repetitions or matrix order.
Thus $F$ has degree at most $\ell$, and its coefficient of $z^\ell$ is exactly the distinct-index chain $\mathcal A_{\ell,N}$.
The exact finite-difference formula for this coefficient is
\[
    \mathcal A_{\ell,N}
       =\frac{\ell^\ell}{\ell!}
          \sum_{j=0}^\ell(-1)^{\ell-j}\binom\ell j F(j/\ell).
\]
Indeed, the alternating sum annihilates powers below $\ell$ and equals $\ell!/\ell^\ell$ on $z^\ell$.
Thus $\|\mathcal A_{\ell,N}\|_{\mathrm{op}}
\leq(2\ell)^\ell K^\ell/\ell!\leq(2eK)^\ell\leq(6K)^\ell$.
Taking a trace proves the second bound.
\end{proof}

\subsection{Expanding and bounding the probability derivative}
\label{lo:score-subsection}
\label{lo:completion-subsection}

The cycle estimates become useful after averaging the product likelihood over the latent Gaussian vectors.
We first justify replacing the smooth cutoff by the untruncated Gaussian law in this algebraic calculation; the resulting error is exponentially small even after normalization.

\begin{lemma}[Uniform barrier replacement]
\label[lemma]{lo:barrier-replacement}
For the spectral path and cutoff defined above, fix $r$ and the sufficiently small constant $\alpha$ used in the prior.
Suppose that the copy budget satisfies $T\log(d+T+1)=o(d)$.
Then there is $c_r>0$ such that, for all sufficiently large $d$, the following bounds hold uniformly in $\theta$, every frozen history $H$ of length at most $T$, and every subset $J\subseteq H$:
\begin{align}
    |Z_\theta(H)-\widetilde Z_\theta(H)|
      +|\partial_\theta Z_\theta(H)
            -\partial_\theta\widetilde Z_\theta(H)|
      +|B_\theta'|
      &\leq e^{-c_r d},
      \label{lo:barrier-absolute}\\
    |\partial_\theta\log Z_\theta(H)
            -\partial_\theta\log\widetilde Z_\theta(H)|
      &\leq e^{-c_r d},
      \label{lo:barrier-log}\\
    \left|
      \frac{\widetilde Z_\theta(H\setminus J)}
           {\widetilde Z_\theta(H)}
      -\frac{Z_\theta(H\setminus J)}
            {Z_\theta(H)}
    \right|
      &\leq e^{-c_r d}.
      \label{lo:barrier-ratio}
\end{align}
In particular, every auxiliary denominator in these formulas is positive.
\end{lemma}

\begin{proof}
We first control the absolute errors and then show that division by the normalizers preserves their exponential decay.
Throughout the proof, the projectors in $H$ are fixed.
The notation $H\setminus J$ deletes only the indicated likelihood factors, without changing the remaining projectors.

Let $U=\sum_j\|g_j\|^2$.
Since $\|X_i\|_{\mathrm{op}}\leq1$,
$\sum_j\lambda_j=1$, and
$\sum_j|\lambda_j'|\leq1$, uniformly over all frozen projectors,
\[
    |f_i|\leq1+\alpha U,
    \qquad
    |\partial_\theta f_i|\leq\alpha U,
    \qquad
    |R_\theta'|\leq\alpha U/d.
\]
The cutoff and its derivative are bounded, so these polynomial bounds and Gaussian integrability justify differentiation under the expectation.

The difference between the physical and auxiliary integrals is
\[
    Z_\theta(H)-\widetilde Z_\theta(H)
    =
    \mathbb E_G\!\left[
        (\chi(R_\theta)-1)\prod_{i\in H}f_i
    \right].
\]
Differentiating gives
\begin{align*}
    &\partial_\theta Z_\theta(H)
      -\partial_\theta\widetilde Z_\theta(H)
    =
    \mathbb E_G\!\left[
        \chi'(R_\theta)R_\theta'
        \prod_{i\in H}f_i
    \right]
    +
    \mathbb E_G\!\left[
        (\chi(R_\theta)-1)
        \sum_{i\in H}
        (\partial_\theta f_i)
        \prod_{\substack{k\in H\\k\ne i}}f_k
    \right].
\end{align*}
Moreover, $B_\theta'=\mathbb E_G[\chi'(R_\theta)R_\theta']$.
Both $\chi(R_\theta)-1$ and $\chi'(R_\theta)$ vanish when $R_\theta<1/3$.
Thus all three errors are supported on the same Gaussian tail.
Using the preceding bounds, their sum is at most
\[
    C(T+1)\,
    \mathbb E_G\!\left[
        \boldsymbol 1_{\{R_\theta\geq1/3\}}
        (1+\alpha U)^{T+1}
    \right].
\]
The factor $T+1$ accounts for differentiating the cutoff or one of the at most $T$ likelihood factors.

To bound this weighted tail, Cauchy--Schwarz gives
\begin{align*}
    \mathbb E_G\!\left[
        \boldsymbol 1_{\{R_\theta\geq1/3\}}
        (1+\alpha U)^{T+1}
    \right]\leq
    \mathbb P_G(R_\theta\geq1/3)^{1/2}
    \left(\mathbb E_G(1+\alpha U)^{2T+2}\right)^{1/2}.
\end{align*}
The first factor is exponentially small by the uniform tail estimate in \Cref{lo:prior-concentration}.
For the second factor, $U$ has the Gamma distribution with shape $rd$, and hence
\[
    \mathbb E_G U^m
    =\prod_{j=0}^{m-1}(rd+j)
    \leq(rd+m)^m.
\]
Using $(1+x)^m\leq2^{m-1}(1+x^m)$ for $x\geq0$, we obtain
\[
    \mathbb E_G(1+\alpha U)^m
    \leq
    2^{m-1}\left[1+\alpha^m(rd+m)^m\right]
    \leq
    \exp[\mathcal{O}_r(m\log(d+m+1))].
\]
Taking $m=2T+2$ therefore bounds the sum of the absolute errors by
\begin{equation}
\label{lo:barrier-error-scale}
    \exp\!\left[
        -cd+\mathcal{O}_r((T+1)\log(d+T+1))
    \right].
\end{equation}
The budget assumption makes the positive term in the exponent $o(d)$, proving \eqref{lo:barrier-absolute} after choosing a smaller exponential constant.
The same moment estimates, without the tail indicator, also give
\[
    |\partial_\theta Z_\theta(H)|
    +|\partial_\theta\widetilde Z_\theta(H)|
    \leq
    \exp[\mathcal{O}_r((T+1)\log(d+T+1))]
    =\exp[o(d)].
\]

We next control the denominators.
On the support of the cutoff,
$1/2\leq f_i\leq1+d/2$.
Since $B_\theta\geq1/2$ for sufficiently large $d$,
\[
    2^{-T-1}
    \leq2^{-|H|}B_\theta
    \leq Z_\theta(H)
    \leq(1+d/2)^T.
\]
The budget assumption implies $T=o(d)$, so the exponentially small absolute error in \eqref{lo:barrier-absolute} is eventually less than $2^{-T-2}$.
Consequently,
\[
    \widetilde Z_\theta(H)
    \geq
    Z_\theta(H)
    -|Z_\theta(H)-\widetilde Z_\theta(H)|
    \geq2^{-T-2}>0.
\]
In particular, both reciprocal denominators are at most
$2^{T+2}=\exp[o(d)]$.
The same bounds hold after deleting any subset of likelihood factors.

For the logarithmic derivatives, suppressing the arguments $\theta,H$ temporarily, write
\[
    \frac{Z'}Z-\frac{\widetilde Z'}{\widetilde Z}
    =
    \frac{Z'-\widetilde Z'}Z
    +
    \frac{\widetilde Z'(\widetilde Z-Z)}
         {Z\widetilde Z}.
\]
Each numerator difference is exponentially small by
\eqref{lo:barrier-absolute}.
The remaining derivative and reciprocal-denominator factors are all bounded by $\exp[o(d)]$.
Thus the right-hand side remains exponentially small, proving \eqref{lo:barrier-log} after decreasing $c_r$.

Finally, the ratio difference has the exact decomposition
\begin{align*}
    \frac{\widetilde Z_\theta(H\setminus J)}
          {\widetilde Z_\theta(H)}
     -\frac{Z_\theta(H\setminus J)}
           {Z_\theta(H)}=
    \frac{\widetilde Z_\theta(H\setminus J)
          -Z_\theta(H\setminus J)}
         {\widetilde Z_\theta(H)}+
    \frac{Z_\theta(H\setminus J)
          \bigl(Z_\theta(H)-\widetilde Z_\theta(H)\bigr)}
         {Z_\theta(H)\widetilde Z_\theta(H)}.
\end{align*}
Again, both numerator differences are exponentially small.
The additional numerator satisfies
$Z_\theta(H\setminus J)\leq(1+d/2)^T=\exp[o(d)]$,
and the reciprocal denominators have the same subexponential bounds as above.
Hence this difference is also bounded by $e^{-c_r d}$ after a final decrease of $c_r$, proving \eqref{lo:barrier-ratio}.
All estimates are uniform over the frozen histories and deleted subsets.
\end{proof}

Since $p_\theta=Z_\theta/B_\theta$ is positive relative to $\nu$, its logarithmic derivative is well defined and satisfies
\[
    \partial_\theta\log p_\theta(H)
       =\partial_\theta\log Z_\theta(H)-B_\theta'/B_\theta.
\]

\begin{lemma}[Expansion of the logarithmic derivative]
\label[lemma]{lo:score-identity}
Consider the spectral path above with a fixed number $r$ of spectral weights and the fixed constant $\alpha$ chosen in the prior construction.
Suppose that $T\log(d+T+1)=o(d)$.
For all sufficiently large $d$, uniformly over every original record $H$ of length $T$ and every $\theta$,
\begin{equation}
\label{lo:score-expansion}
    \partial_\theta\log p_\theta(H)
    =\mathbb E_\theta\!\left[
        \left.
        \sum_{\ell=r}^T
        \frac{\alpha^\ell p_\ell'(\lambda(\theta))}{\ell}
            \widehat{\mathcal C}_{\ell,T}
        \,\right|\,H\right]+e_\theta(H),
    \qquad |e_\theta(H)|\leq e^{-c_r d}.
\end{equation}
Here $p_\ell'(\lambda(\theta))$ denotes the derivative of
$\sum_j\lambda_j(\theta)^\ell$ with respect to $\theta$.
The conditional expectation is over the latent vectors under their physical posterior given the complete original record $H$.
The cycles $\widehat{\mathcal C}_{\ell,T}$ are formed from the original matrices $\widehat Y_i=X_i/f_i$, without auxiliary truncation.
\end{lemma}

\begin{proof}
Fix a complete record $H$ of $T$ rounds and regard its observed
projectors, and hence the matrices $X_i=P_i-I/d$, as fixed.
Subsets of $H$ below refer to subsets of its round indices, even
if some observed projectors coincide.
We first work with the unbarriered expectation
$\widetilde Z_\theta(H)=\mathbb E_G\prod_{i\in H}f_i$.

Since $f_i-1=\alpha\sum_j\lambda_j g_j^\dagger X_i g_j$
and $\operatorname{Tr}X_i=0$, we may apply
\Cref{pre:gaussian-blocks} with $W_i=f_i-1$.
Let $b_J(\theta)$ denote the block coefficients defined in
\eqref{pre:gaussian-block-coefficient}.
Expanding each $f_i=1+(f_i-1)$ and using
\eqref{pre:gaussian-block-moments} gives
\[
    \widetilde Z_\theta(H)
    =\sum_{S\subseteq H}
       \sum_{\pi\in\mathcal P(S)}
          \prod_{J\in\pi}b_J(\theta),
\]
where $\mathcal P(S)$ is the set of partitions of $S$ and the
empty-set contribution is one.

Differentiating this finite expression selects one block $J$
and replaces its coefficient $b_J$ by $b_J'$.
Once this differentiated block is fixed, the remaining blocks
may partition any subset of $H\setminus J$.
Summing their contributions reconstructs
$\widetilde Z_\theta(H\setminus J)$, so
\[
    \partial_\theta\widetilde Z_\theta(H)
       =\sum_{\varnothing\ne J\subseteq H}
          b_J'(\theta)\widetilde Z_\theta(H\setminus J).
\]
Here deleting $J$ removes only its likelihood factors; all other
frozen projectors remain unchanged.

For this fixed record, the matrices $X_i$ do not depend on
$\theta$, so differentiating $b_J$ differentiates only its
spectral coefficient $p_{|J|}(\lambda(\theta))$.
Divide by $\widetilde Z_\theta(H)$, which is positive for
sufficiently large $d$ by \Cref{lo:barrier-replacement}, and
group the blocks by their size.
The coefficient formula in
\eqref{pre:gaussian-block-coefficient} then gives
\begin{equation}
\label{lo:marked-blocks}
    \frac{\partial_\theta\widetilde Z_\theta(H)}
         {\widetilde Z_\theta(H)}
    =
    \sum_{\ell=2}^{T}
      \frac{\alpha^\ell p_\ell'(\lambda)}{\ell}
      \sum_{\substack{i_1,\ldots,i_\ell\in H\\
                      \text{all distinct}}}
      \operatorname{Tr}(X_{i_1}\cdots X_{i_\ell})\times
      \frac{\widetilde Z_\theta(H\setminus\{i_1,\ldots,i_\ell\})}
           {\widetilde Z_\theta(H)}.
\end{equation}
Singleton blocks vanish because $\operatorname{Tr}X_i=0$.
Moreover, moment matching makes $p_\ell'(\lambda)=0$ for
$\ell<r$, so only degrees $\ell\geq r$ remain.

We next return to the physical prior.
By \Cref{lo:barrier-replacement}, the logarithmic derivative
and the ratios in \eqref{lo:marked-blocks} can be replaced by
their counterparts involving $Z_\theta$, with exponentially
small errors.
These errors remain exponentially small after summation:
there are at most
$\exp[\mathcal{O}(T\log(T+1))]$ ordered tuples, each trace has
absolute value at most $d$ because $\|X_i\|_{\mathrm{op}}\leq1$,
and each scalar coefficient has absolute value at most one.
The assumption $T\log(d+T+1)=o(d)$ therefore absorbs all these
factors into the exponential error.
Furthermore,
\[
    \partial_\theta\log p_\theta(H)
       =\frac{\partial_\theta Z_\theta(H)}{Z_\theta(H)}
          -\frac{B_\theta'}{B_\theta}.
\]
The last term is also exponentially small, since
$B_\theta\geq1/2$ for sufficiently large $d$ and
\Cref{lo:barrier-replacement} bounds $B_\theta'$.

It remains to identify the physical ratios as posterior
expectations.
By Bayes' formula, the conditional density of the latent
vectors given $H$, relative to the original Gaussian law, is
$\chi(R_\theta)\prod_{i\in H}f_i/Z_\theta(H)$.
Multiplying this density by $\prod_{i\in J}f_i^{-1}$ cancels
exactly the likelihood factors indexed by $J$, giving
\begin{equation}
\label{lo:cavity-posterior}
    \mathbb E_\theta\!\left[
       \left.\prod_{i\in J}f_i^{-1}\,\right|\,H\right]
    =
    \frac{\mathbb E_G\!\left[
          \chi(R_\theta)\prod_{i\in H\setminus J}f_i\right]}
         {Z_\theta(H)}
    =
    \frac{Z_\theta(H\setminus J)}{Z_\theta(H)}.
\end{equation}
All inverse factors are well defined on the physical prior's
support because $f_i\geq1/2$.

Substitute \eqref{lo:cavity-posterior} into the physical version
of \eqref{lo:marked-blocks}.
The trace factors are fixed once $H$ is given, so each summand
becomes a posterior expectation of
\[
    \operatorname{Tr}(X_{i_1}\cdots X_{i_\ell})
       \prod_{a=1}^{\ell}f_{i_a}^{-1}
    =
    \operatorname{Tr}(
       \widehat Y_{i_1}\cdots\widehat Y_{i_\ell}),
    \qquad \widehat Y_i=X_i/f_i.
\]
Summing over all ordered distinct indices gives
$\widehat{\mathcal C}_{\ell,T}$, proving
\eqref{lo:score-expansion} after decreasing $c_r$ if necessary.
This identity uses the original matrices and therefore holds
for every record.
On good records, these cycles also coincide with the auxiliary
cycles used in the martingale estimates.
\end{proof}

\paragraph{The contribution of records in $G_\kappa$.}
For a record in $G_\kappa$, every $a_i$ equals one, so all original and auxiliary cycles coincide.
Since this event depends only on $H$, the identity \eqref{lo:score-expansion} gives
\[
    \boldsymbol1_{G_\kappa}\partial_\theta\log p_\theta(H)
      =\boldsymbol1_{G_\kappa}
         \mathbb E_\theta\!\left[
           \left.\sum_{\ell=r}^T
              \frac{\alpha^\ell p_\ell'}{\ell}\mathcal C_{\ell,T}
           \right|H\right]
        +\boldsymbol1_{G_\kappa}e_\theta(H).
\]
Conditional Jensen and Cauchy--Schwarz therefore imply
\begin{equation}
\label{lo:good-derivative-sum}
    \int_{G_\kappa}|\partial_\theta p_\theta|\,d\nu
      =\mathbb E_\theta[\boldsymbol1_{G_\kappa}|\partial_\theta\log p_\theta|]
      \leq\sum_{\ell=r}^T
          \frac{\alpha^\ell|p_\ell'|}{\ell}
                \|\mathcal C_{\ell,T}\|_{L^2}
          +e^{-c_r d}.
\end{equation}
All norms on the right are taken under the original, unconditional joint law.
In particular, the martingale estimates are applied before restricting to $G_\kappa$; conditioning the experiment on this future-dependent event would not preserve the martingale-difference property.

We now sum every degree in \eqref{lo:good-derivative-sum}.
The constants in the variance bounds grow with the degree, so a fixed-degree estimate alone is insufficient for the entire tail.
With $A$ from \Cref{lo:cycle-bounds}, choose the dimension-independent constants
\begin{equation}
\label{lo:cutoff-choices}
    \zeta=\frac1{32A(h+1)\log2},
    \qquad
    0<\alpha\leq
       \min\left\{\frac1{100},\frac1{8\kappa},
                   \frac{r e^{-4/\zeta}}{12K}\right\}.
\end{equation}

\begin{proposition}[Derivative bound on good records]
\label[proposition]{lo:derivative-bound}
For $T=\lfloor c_0d^{h/(h+1)}\rfloor$, with fixed $0<c_0\leq1$, uniformly over $\theta$ and all adaptive single-copy protocols,
\begin{equation}
\label{lo:good-derivative}
    \int_{G_\kappa}|\partial_\theta p_\theta|\,d\nu
       \leq C_h\frac{T^{(h+1)/2}}{d^{h/2}}+o(1)
       \leq C_hc_0^{(h+1)/2}+o(1).
\end{equation}
The error tends to zero for fixed $t,\alpha,c_0$ as $d\to\infty$.
\end{proposition}

\begin{proof}
Put $L_d=\lfloor\zeta\log d\rfloor$.
For sufficiently large $d$, we split the sum in \eqref{lo:good-derivative-sum} into three ranges: the first surviving block $r\leq\ell\leq2^h$, the intermediate degrees $2^h<\ell\leq L_d$, and the long cycles $L_d<\ell\leq T$.
The first range gives the main contribution.
We control the second range using the additional dimension decay in the cycle second-moment bound, and the third using the deterministic cycle bound together with the small fixed value of $\alpha$.

\emph{The first surviving block.}
Recall that $r=2^{h-1}+1$, so $\lceil\log_2\ell\rceil=h$ throughout $r\leq\ell\leq2^h$.
The quantitative cycle bound \eqref{lo:cycle-bound-quantitative} therefore gives
\[
    \|\mathcal C_{\ell,T}\|_{L^2}
    \leq \exp[A\ell\log(\ell+1)]
          \frac{T^{(h+1)/2}}{d^{h/2}}.
\]
Thus every degree in this block has the same dependence on $T$ and $d$.
Since $h$ is fixed, the block contains only finitely many degrees independent of $d$.
Using $\alpha^\ell|p_\ell'|/\ell\leq1$ and absorbing their degree-dependent constants into $C_h$, their total contribution is at most $C_hT^{(h+1)/2}/d^{h/2}$.

\emph{The intermediate degrees.}
For $\ell>2^h$, the exponent $\lceil\log_2\ell\rceil$ is strictly larger than $h$.
Relative to the first block, the cycle bound therefore gains the factor
$(T/d)^{(\lceil\log_2\ell\rceil-h)/2}$.
Dividing the intermediate contribution by $T^{(h+1)/2}/d^{h/2}$ and again using $\alpha^\ell|p_\ell'|/\ell\leq1$, we obtain the upper bound
\begin{equation}
\label{lo:intermediate-tail}
    \sum_{\ell=2^h+1}^{L_d}
       \exp[A\ell\log(\ell+1)]
       \left(\frac Td\right)^{(\lceil\log_2\ell\rceil-h)/2}.
\end{equation}
Although $T/d\leq d^{-1/(h+1)}$ tends to zero, the prefactor grows with $\ell$.
We must therefore control this entire sum uniformly up to $L_d$, rather than only checking each fixed degree.

First consider $2^h<\ell<2^{2h+2}$.
There are only finitely many such degrees, depending on $h$, and each satisfies $\lceil\log_2\ell\rceil-h\geq1$.
Their total contribution to \eqref{lo:intermediate-tail} is consequently $\mathcal{O}_h(d^{-1/[2(h+1)]})=o(1)$.

For the remaining degrees $\ell\geq2^{2h+2}$, we have
$\lceil\log_2\ell\rceil-h\geq\log\ell/(2\log2)$ and
$\log(\ell+1)\leq2\log\ell$.
Together with $T/d\leq d^{-1/(h+1)}$, these inequalities bound the logarithm of each summand by
\[
    2A\ell\log\ell
      -\frac{\log\ell\log d}{4(h+1)\log2}
    \leq-\frac{\log\ell\log d}{8(h+1)\log2}.
\]
For the last inequality, use $\ell\leq\zeta\log d$ and
$\zeta=1/[32A(h+1)\log2]$.
Indeed, the positive term is then at most
$\log\ell\log d/[16(h+1)\log2]$, so the negative term dominates.

Moreover, $\ell\geq2^{2h+2}$ implies that the final expression is at most $-\tfrac14\log d$.
Each of these summands is therefore at most $d^{-1/4}$.
There are only $\mathcal{O}_h(\log d)$ summands, so their total is $\mathcal{O}_h((\log d)d^{-1/4})=o(1)$.
Combining the two parts shows that \eqref{lo:intermediate-tail} is $o(1)$.
Thus the intermediate degrees contribute $o(T^{(h+1)/2}/d^{h/2})$ to the derivative bound.

\emph{The long cycles.}
For $\ell>L_d$, we use the deterministic estimate
$|\mathcal C_{\ell,T}|\leq d(6K)^\ell$ from \eqref{lo:long-chain-bound}.
Since this bound holds for every realization, it also gives
$\|\mathcal C_{\ell,T}\|_{L^2}\leq d(6K)^\ell$ for each degree.
Combining it with $|p_\ell'|/\ell\leq(2/r)^{\ell-1}$ bounds the degree-$\ell$ contribution by
$(rd/2)(12K\alpha/r)^\ell$.
Our choice of $\alpha$ ensures $12K\alpha/r\leq e^{-4/\zeta}$, so summing the resulting geometric series gives
\begin{equation}
\label{lo:long-tail}
    \sum_{\ell=L_d+1}^{T}
        \frac{\alpha^\ell|p_\ell'|}{\ell}
        \|\mathcal C_{\ell,T}\|_{L^2}
    \leq \frac{rd}{2}
          \frac{e^{-4(L_d+1)/\zeta}}{1-e^{-4/\zeta}}
    =\mathcal{O}_h(d^{-3}).
\end{equation}
Here $L_d+1>\zeta\log d$ implies
$e^{-4(L_d+1)/\zeta}\leq d^{-4}$.
The small fixed value of $\alpha$ therefore controls the entire long-degree tail once the deterministic estimate has reduced the degree dependence to an exponential.

Adding the three ranges and the exponentially small error in \eqref{lo:good-derivative-sum} yields
\[
    \int_{G_\kappa}|\partial_\theta p_\theta|\,d\nu
    \leq C_h\frac{T^{(h+1)/2}}{d^{h/2}}
       +o\!\left(\frac{T^{(h+1)/2}}{d^{h/2}}\right)
       +\mathcal{O}_h(d^{-3})+e^{-c_rd}.
\]
Finally, $T\leq c_0d^{h/(h+1)}$ gives
$T^{(h+1)/2}/d^{h/2}\leq c_0^{(h+1)/2}\leq1$.
The remaining terms are therefore $o(1)$, proving \eqref{lo:good-derivative}.
\end{proof}

\paragraph{Completion of the adaptive lower bound.}
Apply \Cref{lo:tv-split} with the good-record derivative bound \eqref{lo:good-derivative}.
Since $\kappa=2(h+1)$ and $T\leq d^{h/(h+1)}$, the exceptional-record term is $\mathcal{O}_h(d^{-1})$.
Consequently,
\begin{equation}
\label{lo:tv-full}
\begin{aligned}
    \operatorname{TV}(P_0^{(T)},P_{\pi/r}^{(T)})
       &\leq C_h\frac{T^{(h+1)/2}}{d^{h/2}}
          +d\left(\frac{2eT}{\kappa d}\right)^\kappa+o(1)\\
       &\leq C_hc_0^{(h+1)/2}+o(1).
\end{aligned}
\end{equation}

\begin{proof}[Proof of \Cref{lo:main}]
Choose $\alpha$ by \eqref{lo:cutoff-choices}, $\epsilon_t$ by \eqref{lo:epsilon}, and a fixed $c_0>0$ making the constant term in \eqref{lo:tv-full} less than $1/12$.
For large enough $d$, the endpoint total variation is less than $1/6$, contradicting \eqref{lo:necessary-tv} for any estimator using at most $\lfloor c_0d^{h/(h+1)}\rfloor$ copies.
Padding with unused rounds handles smaller budgets; smaller errors obey the same reduction.
Absorbing floors and dimension restrictions into $c_t,d_t$ proves the theorem.
\end{proof}

For $t=3,4$, the surviving lengths $3,4$ both have variance $\mathcal{O}(T^3/d^2)$, giving the common threshold $d^{2/3}$.
In general, moment matching leaves the first block $2^{h-1}<\ell\leq2^h$, with variance scale $T^{h+1}/d^h$; this explains the dyadic hierarchy.
Combining Theorems~\ref{up:main} and~\ref{lo:main} proves \Cref{thm:optimal-moments} for every fixed $0<\epsilon\leq\epsilon_t$.

\subsection{The stronger lower bound for nonadaptive protocols}
\label{sec:proof-nonadaptive}

We prove the lower bound in \Cref{thm:nonadaptive-moments} directly for the original transcript distribution.
Condition on the input-independent random seed, which fixes the entire POVM schedule in the nonadaptive model of \Cref{pre:learning-model}.
Rank-one refinement and padding to $T$ prescribed measurements preserve nonadaptivity.
The POVMs may differ between rounds, and the final classical processing is unrestricted.
All estimates below are uniform in the seed.

For this subsection, replace the adaptive parameters by
\begin{equation}
\label{lo:nonadaptive-parameters}
    r=t,\qquad \kappa=2t.
\end{equation}
The spectral path and physical prior remain unchanged in form, and the bound $K=2(\kappa+2)$ is evaluated at this threshold.
The first $t-1$ spectral moments agree along the path, whereas the endpoint gap is
\begin{equation}
\label{lo:nonadaptive-gap}
    p_t(\lambda(0))-p_t(\lambda(\pi/t))
    =2^{2-t}t^{1-t}>0.
\end{equation}
At any fixed accuracy at most $\alpha^t2^{-t-1}t^{1-t}$, the estimation-to-testing reduction therefore requires endpoint total variation at least $1/3-o(1)$.

\paragraph{Independent cycles in the original experiment.}
The adaptive proof used martingale differences and their orthogonality across rounds.
Conditional on the latent vectors and the seed, the original nonadaptive experiment has independent outcomes, giving additional orthogonality between different index sets.
This conditional independence does not require the POVMs to be identical across rounds.
We use the raw matrices here, since the auxiliary indicators depend on earlier outcomes and need not preserve independence.
Recall the raw matrices and cycles
\[
    \widehat Y_i
      =\frac{P_i-I/d}{d\operatorname{Tr}(P_i\rho_\theta)},
    \qquad
    \widehat{\mathcal C}_{\ell,T}
      =\sum_{\substack{i_1,\ldots,i_\ell\leq T\\\mathrm{all\ distinct}}}
         \operatorname{Tr}(\widehat Y_{i_1}\cdots\widehat Y_{i_\ell}).
\]
The one-step centering and variance calculations of \Cref{lo:Y-properties} apply to these raw matrices: they require only POVM normalization and the physical floor.

\begin{lemma}[Independent-cycle estimate]
\label[lemma]{lo:nonadaptive-cycles}
For every $2\leq\ell\leq T$, the original nonadaptive experiment satisfies
\begin{equation}
\label{lo:nonadaptive-raw-cycle-bound}
    \mathbb E_\theta\bigl|\widehat{\mathcal C}_{\ell,T}\bigr|^2
    \leq \ell!\,2^\ell\frac{T^\ell}{d^{\ell-1}}.
\end{equation}
In particular, for the good event $G_\kappa=\{\|\sum_{i=1}^T P_i\|_{\mathrm{op}}<\kappa\}$ from \eqref{lo:good-event},
\begin{equation}
\label{lo:nonadaptive-cycle-bound}
    \mathbb E_\theta\!\left[
       \mathbf1_{G_\kappa}\bigl|\widehat{\mathcal C}_{\ell,T}\bigr|^2\right]
    \leq \ell!\,2^\ell\frac{T^\ell}{d^{\ell-1}}.
\end{equation}
The expectations include the latent prior, the random seed, and the measurement outcomes.
\end{lemma}

\begin{proof}
Condition on the latent vectors and seed, and write $\mathbb E_*$ for this conditional expectation.
Then the $\widehat Y_i$ are independent, centered, and satisfy
\[
    \mathbb E_*\widehat Y_i^2\preceq\frac2d I,
    \qquad
    \mathbb E_*|\operatorname{Tr}(M\widehat Y_i)|^2
       \leq\frac2d\|M\|_{\mathrm{HS}}^2
\]
for every fixed complex matrix $M$.
In particular, $\mathbb E_*\|\widehat Y_i\|_{\mathrm{HS}}^2\leq2$.

For an ordered distinct-index tuple $I=(i_1,\ldots,i_\ell)$, put
$W_I=\operatorname{Tr}(\widehat Y_{i_1}\cdots\widehat Y_{i_\ell})$.
Independence allows us to fix the other factors when integrating any one matrix.
Integrating the last factor using scalar variance, and then the remaining factors using matrix variance, gives
\[
    \mathbb E_*|W_I|^2
    \leq\frac2d\,
       \mathbb E_*\|\widehat Y_{i_1}\cdots\widehat Y_{i_{\ell-1}}\|_{\mathrm{HS}}^2
    \leq\frac{2^\ell}{d^{\ell-1}}.
\]
For the second step, independence and $\mathbb E_*\widehat Y_i^2\preceq2I/d$ imply
$\mathbb E_*\|Q\widehat Y_i\|_{\mathrm{HS}}^2\leq(2/d)\mathbb E_*\|Q\|_{\mathrm{HS}}^2$ whenever $Q$ is a product of the other factors.
Iterate this inequality down to one factor, whose squared Hilbert--Schmidt norm has expectation at most $2$.
To estimate the complete cycle, expand
\[
    \mathbb E_*\bigl|\widehat{\mathcal C}_{\ell,T}\bigr|^2
    =\sum_{I,J}\mathbb E_*[W_I\overline{W_J}].
\]
If $I$ and $J$ have different underlying index sets, some index appears in just one tuple.
Its matrix is independent of all other factors and appears linearly, so its zero mean gives
$\mathbb E_*[W_I\overline{W_J}]=0$.

For each $I$, exactly $\ell!$ tuples $J$ have the same index set.
Cauchy--Schwarz bounds each surviving cross term in absolute value by $2^\ell/d^{\ell-1}$.
There are at most $T^\ell$ choices of $I$, proving \eqref{lo:nonadaptive-raw-cycle-bound} after averaging over the latent vectors and seed.
The second bound follows by dropping the indicator $\mathbf1_{G_\kappa}$ from a nonnegative integrand.
In particular, no independence claim is made after conditioning on $G_\kappa$.
\end{proof}

\paragraph{Bounding the derivative on good records.}
Choose fixed parameters
\begin{equation}
\label{lo:nonadaptive-alpha}
    0<\alpha\leq
    \min\left\{\frac1{100},\frac1{8\kappa},
                  \frac{t e^{-4}}{12K}\right\},
    \qquad T=\left\lfloor c\,d^{1-1/t}\right\rfloor,
    \qquad 0<c\leq1.
\end{equation}
Write $p_\theta$ for the density of the original transcript law $P_\theta^{(T)}$ with respect to the common measure $\nu$.
Since $T\log(d+T+1)=o(d)$, the marked-cycle derivative identity \eqref{lo:score-expansion} gives
\begin{equation}
\label{lo:nonadaptive-score}
\begin{split}
    \partial_\theta p_\theta(H)
    &=p_\theta(H)\,\mathbb E_\theta\!\left[
       \left.
       \sum_{\ell=t}^{T}
          \frac{\alpha^\ell p_\ell'(\lambda(\theta))}{\ell}
              \widehat{\mathcal C}_{\ell,T}
       \,\right|\,H\right]
       +p_\theta(H)e_\theta(H),\\
    |e_\theta(H)|&\leq e^{-c_t'd}.
\end{split}
\end{equation}
The indicator of $G_\kappa$ is determined by the transcript and does not depend on $\theta$.
As in the adaptive proof, conditional Jensen and Cauchy--Schwarz bound each degree's contribution by its absolute coefficient times
$\bigl(\mathbb E_\theta[\mathbf1_{G_\kappa}|\widehat{\mathcal C}_{\ell,T}|^2]\bigr)^{1/2}$.
Substituting \eqref{lo:nonadaptive-cycle-bound} and
$|p_\ell'|/\ell\leq(2/t)^{\ell-1}$ gives the bound
\[
    b_\ell
      :=\alpha^\ell(2/t)^{\ell-1}
        \sqrt{\ell!\,2^\ell}\,
        \frac{T^{\ell/2}}{d^{(\ell-1)/2}}
\]
for the contribution of degree $\ell$.

Put $L=\lceil\log d\rceil$.
For $t\leq\ell<L$, the successive ratio satisfies
\[
    \frac{b_{\ell+1}}{b_\ell}
       =\frac{2\alpha}{t}\sqrt{2(\ell+1)T/d}
       \leq\frac{2\alpha}{t}\sqrt{2(L+1)T/d}=o(1),
\]
because $T/d\leq d^{-1/t}$.
For sufficiently large $d$, this ratio is at most $1/2$ throughout the range, so the sum is controlled by its first term:
\begin{equation}
\label{lo:nonadaptive-short-tail}
    \sum_{\ell=t}^{L}b_\ell
       \leq2b_t
       \leq C_t\frac{T^{t/2}}{d^{(t-1)/2}}.
\end{equation}

For the remaining degrees, the raw matrices agree on $G_\kappa$ with the retained matrices used only in the analysis.
The finite-difference bound of \Cref{lo:long-chains} therefore gives
$\mathbf1_{G_\kappa}|\widehat{\mathcal C}_{\ell,T}|\leq d(6K)^\ell$.
Since $12K\alpha/t\leq e^{-4}$, the entire long-degree contribution is bounded pointwise by
\begin{equation}
\label{lo:nonadaptive-long-tail}
\begin{split}
    \mathbf1_{G_\kappa}
    \sum_{\ell=L+1}^{T}
       \frac{\alpha^\ell|p_\ell'|}{\ell}
                 |\widehat{\mathcal C}_{\ell,T}|
    \leq\frac{td}{2}\sum_{\ell>L}e^{-4\ell}=\frac{td}{2}\frac{e^{-4(L+1)}}{1-e^{-4}}
      =\mathcal{O}_t(d^{-3}).
\end{split}
\end{equation}
Taking expectations bounds its contribution to the derivative integral by the same quantity.
Combining \eqref{lo:nonadaptive-score}--\eqref{lo:nonadaptive-long-tail} proves
\begin{equation}
\label{lo:nonadaptive-good-derivative}
    \int_{G_\kappa}|\partial_\theta p_\theta(H)|\,\nu(dH)
       \leq C_t\frac{T^{t/2}}{d^{(t-1)/2}}+o(1)
       \leq C_t c^{t/2}+o(1).
\end{equation}

\paragraph{Bad records and completion.}
The rare-record bound \eqref{lo:rare-records}, with $\kappa=2t$, yields
\[
    P_\theta^{(T)}(G_\kappa^c)
       \leq d\left(\frac{2eT}{\kappa d}\right)^\kappa
       \leq\left(\frac et\right)^{2t}c^{2t}d^{-1}.
\]
Combining this endpoint probability bound with \Cref{lo:tv-split} and \eqref{lo:nonadaptive-good-derivative} yields
\begin{equation}
\label{lo:nonadaptive-tv}
\begin{aligned}
    \operatorname{TV}\!\left(P_0^{(T)},P_{\pi/t}^{(T)}\right)
       &\leq\frac12\int_0^{\pi/t}
          \int_{G_\kappa}|\partial_\theta p_\theta(H)|
                            \,d\nu(H)\,d\theta
          +\mathcal{O}_t(d^{-1})\\
       &\leq C_t c^{t/2}+o(1).
\end{aligned}
\end{equation}
A sufficiently small fixed $c$ makes this less than $1/4+o(1)$, contradicting the testing requirement at the stated accuracy or smaller.
Thus arbitrary nonadaptive single-copy protocols require $\Omega_t(d^{1-1/t})$ copies.
Together with the nonadaptive collision upper bound~\cite{Li2026SingleSetting}, this proves \Cref{thm:nonadaptive-moments}.

\paragraph{Comparison with adaptive protocols.}
The adaptive and nonadaptive proofs follow the same framework: construct a moment-matching spectral path, expand the transcript derivative into cycles, control good and exceptional records, and integrate to bound the endpoint total variation.
The main technical difference is the second-moment estimate for the cycle statistics.

For adaptive protocols, we group the cycle terms by their latest round.
The increment identity
$\mathcal C_{\ell,n}-\mathcal C_{\ell,n-1}
=\ell\operatorname{Tr}(Y_n\mathcal A_{\ell-1,n-1})$
and martingale orthogonality give
\[
    \mathbb E_\theta|\mathcal C_{\ell,T}|^2
    \leq\frac{2\ell^2}{d}\sum_{n=\ell}^T
         \mathbb E_\theta\|\mathcal A_{\ell-1,n-1}\|_{\mathrm{HS}}^2
    \leq C_\ell
         \frac{T^{\lceil\log_2\ell\rceil+1}}
              {d^{\lceil\log_2\ell\rceil}}.
\]
The chain estimate recursively reduces the product length by roughly one half, producing the dyadic dependence on $\ell$.
For nonadaptive protocols, conditional independence instead allows us to expand the cycle's squared magnitude directly.
Cross terms with different underlying index sets vanish, leaving at most $\ell!T^\ell$ terms, each bounded by $2^\ell/d^{\ell-1}$.
This gives the stronger estimate
$\mathbb E_\theta|\widehat{\mathcal C}_{\ell,T}|^2
\leq\ell!2^\ell T^\ell/d^{\ell-1}$.
The auxiliary and raw cycles agree on good records, so both estimates control the same contribution to the transcript derivative.

For a concrete example of the additional cancellation, take
$I=(1,2,3,5)$ and $J=(1,2,4,5)$.
The corresponding cross term is
\[
    W_I\overline{W_J}
    =
    \operatorname{Tr}(
       \widehat Y_1\widehat Y_2\widehat Y_3\widehat Y_5)
    \overline{\operatorname{Tr}(
       \widehat Y_1\widehat Y_2\widehat Y_4\widehat Y_5)}.
\]
In the nonadaptive model, after fixing the latent vectors, the seed, and all matrices other than $\widehat Y_3$, this expression is linear in the independent, mean-zero matrix $\widehat Y_3$.
Its conditional expectation is therefore zero.
In an adaptive protocol, the fifth measurement may depend on the third outcome, so $\widehat Y_5$ can depend on $\widehat Y_3$.
Fixing the other matrices then need not preserve the zero mean of $\widehat Y_3$, and this cancellation is no longer guaranteed.

\ifanonymous
\else
    \section*{Acknowledgments}
    We thank Zihao Li for pointing out that the $\mathcal{O}(d^{1-1/t})$ sample complexity achieved in their work~\cite{Li2026SingleSetting} is optimal among all single-copy protocols that use the same measurement basis on every copy. This work grew out of extensive exchanges between the author and generative AI systems.
    Generative AI played an important role in developing the results and their derivations, as well as in writing the manuscript.
    The author takes full responsibility for the correctness of the results and the content of the article.
\fi


\phantomsection
\addcontentsline{toc}{section}{References}

\bibliographystyle{alphaurl}

\newcommand{\etalchar}[1]{$^{#1}$}

\end{document}